\documentclass{LMCS}

\def\dOi{10(4:3)2014}
\lmcsheading%
{\dOi}
{1--38}
{}
{}
{Jan.~13, 2013}
{Nov.~12, 2014}
{}

\ACMCCS{[{\bf Theory of computation}]: Logic---Finite Model Theory;
  [{\bf Mathematics of computing}]: Discrete mathematics---Graph theory---Trees}

\usepackage{comment}
\usepackage{hyperref}
\usepackage{amsmath,amssymb,amsthm} 
\usepackage{mathrsfs}
\usepackage{graphicx}
\usepackage{xspace}
\newcommand{\fo}{\mathrm{FO}\xspace} 
 
\newcommand{\EF}{Ehrenfeucht-Fra\"iss\' e\xspace}

\begin{document}

\title[quantifier structure hierarchy]{On the strictness of the quantifier structure hierarchy in first-order logic\rsuper*}

\author{Yuguo He}	%required
\address{School of Computer Science, Beijing Institute of Technology, Beijing 100081, China.}	%required
\email{hugo274@gmail.com}  %optional
%\thanks{thanks 1, optional.}	%optional

%% etc.

%% required for running head on odd and even pages, use suitable
%% abbreviations in case of long titles and many authors:

%% mandatory lists of keywords and classifications:
\titlecomment{{\lsuper*}This paper is from a part of the author's PhD thesis. It is an extension and revision of a conference paper published in Proceedings of the 25th Annual IEEE Symposium on Logic in Computer Science.}
%%%%%%%%%%%%%%%%%%%%%%%%%%%%%%%%%%%%%%%%%%%%%%%%%%%%%%%%%%%%%%%%%%%%%%%%%%%

%% the abstract has to PRECEED the command \maketitle:
%% be sure not to issue the \maketitle command twice!

\begin{abstract}
  \noindent We study a natural hierarchy in first-order logic, namely the quantifier structure hierarchy, which gives a systematic classification of first-order formulas based on structural quantifier resource. We define a variant of Ehrenfeucht-Fra\"iss\' e games that characterizes quantifier classes and use it to prove that this hierarchy is strict over finite structures, using strategy compositions. Moreover, we prove that this hierarchy is strict even over ordered finite structures, which is interesting in the context of descriptive complexity. 
\end{abstract}

\maketitle

%% start the paper here:
\section{Introduction}\label{S:one}
%\subsection*{\TeX-nical matters}
One of the major interests of finite model theory is to separate the expressive power of different logics or fragments of logics. Quantifiers are an important logical resource for measuring the logical complexity of problems. The study of fragments of first-order logic ($\fo$) based on quantifier structures, especially quantifier prefixes, has a long history in model theory \cite{GradelM96}. However, so far there are few results about the expressive power of such fragments. Walkoe \cite{W.Walkoe_partially-ordered_1970} proved that there exists a sentence with prefix $p$ which is different from any sentence with prefix $q$ if $p$ and $q$ are different but with the same length. In the proof, the structures are assumed to be infinite. Afterwards Keisler and Walkoe \cite{kw73diversity} improved this result by showing its validity over finite structures. Chandra and Harel \cite{ch82fohierarchy} proved that $\Sigma_k\subsetneq \Sigma_{k+1}$ over finite digraphs. 
Sipser \cite{sipser_borel_1983} proved a similar result in the context of unbounded fan-in bounded depth circuits.  

In 1996, Gr\"adel and McColm \cite{GradelM96} established a strict hierarchy based on quantifier classes in the infinitary logic over finite structures and resolved a conjecture of Immerman, i.e.\ $\Sigma_i^{TC}\subsetneq \Sigma_{i+1}^{TC}$ for each $i$. At the same time, they proposed a conjecture on the expressive power of the fragments of $\fo$  based on prefixes, which generalized the previous results \cite{W.Walkoe_partially-ordered_1970}, \cite{kw73diversity} and \cite{ch82fohierarchy}. In 1998, Rosen \cite{rosen05prifix} confirmed this conjecture and called the strict hierarchy based on these fragments of $\fo$ the first-order prefix hierarchy. Actually, Rosen proved a stronger result, which states that, over a single binary relation, for any prefix $p$ there is a first-order sentence $\varphi_p$ in prenex normal form with prefix $p$, such that for any sentences $\psi$ in infinitary logic, $\varphi_p$ is not equivalent to $\psi$ if $p$ is not embeddable in the ``quantifier structure'' of $\psi$.\footnote{Here, the notion ``quantifier structure'' is from Gr\"{a}del and McColm \cite{GradelM96}, which is different from ours (cf. Definition \ref{quantifier-structure-1}).}  However, a stronger version of the conjecture remains open, i.e.\ whether it holds over finite structures or not \cite{rosen_some-aspects_2002,rosen05prifix}. One way to prove the conjecture is to prove a finite version of Rosen's main theorem.  

In this paper, we continue this line of study. 
We define a variant of Ehrenfeucht-Fra\"iss\' e games that characterizes quantifier classes and prove the following main result:

\textit{Let $S_1$ and $S_2$ be two finite $\Gamma$-labeled forests. Over the class of all finite digraphs,
\[
         \mbox{if}\hspace{3pt} S_1\npreceq_e S_2, \hspace{2pt} \mbox{then}\hspace{3pt} \fo\{S_1\}\nsubseteq \fo\{S_2\}.
\]}

\noindent The structures we use in the proof are finite trees, which
makes it easy to prove a stronger result: The above main result holds
even when the structures have a linear order. Here we introduce the ideas that are used to deal with linear order in a simpler context: 

\textit{Over the class of all ordered finite digraphs,
\[
         \mbox{if}\hspace{3pt} \mathscr{W}(S_1)\nsubseteq \mathscr{W}(S_2),\hspace{2pt} \mbox{then}\hspace{3pt} \fo\{S_1\}\nsubseteq \fo\{S_2\}.
\]}

\section{Preliminaries}

\subsection{General background}

Let $\mathbb{N}$ and $\mathbb{N}^+$ denote the set of natural numbers (non-negative integers) and positive natural numbers respectively. 

We assume that the readers have basic knowledge about finite model theory. 
In the following we briefly introduce some necessary background. The readers can cf. the textbook \cite{Libkin04Elements} for more of it. 

A \textit{relational signature} consists of a sequence of relation and constant symbols. In this paper, a signature is relational and finite, whenever mentioned.  

Let $\sigma=\langle R_1,\cdots\!,R_m,c_1,\cdots\!,c_n\rangle$ be a signature,  a $\sigma$-\textit{structure} $\mathfrak{A}$  consists of a universe $|\mathfrak{A}|$ together with an interpretation of\begin{itemize}
\item each $k$-ary relation symbol $R_i\in \sigma$ as a $k$-ary relation on $|\mathfrak{A}|$, denoted by $R_i^{\mathfrak{A}}$;
\item each constant symbol $c_i\in \sigma$ as an element in $|\mathfrak{A}|$. 
\end{itemize}   
A structure is called \textit{finite} if its universe is a finite set. 

A $\sigma$-structure $\mathfrak{A}^{\prime}$ is a \textit{substructure} of $\mathfrak{A}$ if the following hold:
\begin{enumerate}
\item $|\mathfrak{A}^{\prime}|\subseteq |\mathfrak{A}|$;
\item For any $k$-ary relation $R\in \sigma\cup\{=\}$, $R^{\mathfrak{A}^{\prime}}=R^{\mathfrak{A}}\cap |\mathfrak{A}^{\prime}|^k$;
\item For any constant $c\in \sigma$, $c^{\mathfrak{A}^{\prime}}=c^{\mathfrak{A}}$.
\end{enumerate}  

Let $\sigma^{\prime}\subseteq \sigma$. The $\sigma^{\prime}$-\textit{reduct} of  $\mathfrak{A}$, denoted $\mathfrak{A}|\sigma^{\prime}$, is obtained from $\mathfrak{A}$ by leaving all the symbols in $\sigma\setminus\sigma^{\prime}$ uninterpreted.

Let $\mathfrak{A}$ and $\mathfrak{B}$ be wo structures of the same signature. 
An \textit{isomorphism} between $\mathfrak{A}$ and $\mathfrak{B}$ is a bijection $h: |\mathfrak{A}|\rightarrow|\mathfrak{B}|$ such that the following hold: \begin{enumerate}
\item For any $k$-ary relation $R\in \sigma\cup\{=\}$ and $(a_1,\ldots,a_k)\in |\mathfrak{A}|^k$, \[(a_1,\ldots,a_k)\in R^{\mathfrak{A}} \hspace{3pt}\mathrm{iff}\hspace{3pt} (h(a_1),\ldots,h(a_k))\in R^{\mathfrak{B}};\]
\item For any constant $c\in \sigma$, $h(c^{\mathfrak{A}})=c^{\mathfrak{B}}$. 

\end{enumerate}

\noindent Say that two structures $\mathfrak{A}$ and $\mathfrak{B}$ over the same signature are isomorphic if there is an isomorphism between them, denoted $\mathfrak{A}\cong\mathfrak{B}$.

Let $\bar a=(a_1,\cdots,a_k)\in |\mathfrak{A}|^k$, $\bar b=(b_1,\cdots,b_k)\in |\mathfrak{B}|^k$. Say that $(\bar a,\bar b)$ defines a\textit{ partial isomorphism} between $\mathfrak{A}$ and $\mathfrak{B}$ if $\bar a$ contains  the elements that interpret all the constants of $\mathfrak{A}$, $\bar b$ contains  the elements that interpret all the constants of $\mathfrak{B}$, and the substructure of $\mathfrak{A}$ that is generated by $\bar a$ is isomorphic to the substructure of $\mathfrak{B}$ that is generated by $\bar b$.   
More precisely, the following hold:  \begin{enumerate}
\item for any $m$-ary relation symbol $R\in \sigma\cup\{=\}$ and any sequence $(i_1,\cdots,i_m)$ of numbers from [k],  
\[
(a_{i_1},\cdots,a_{i_m})\in R^{\mathfrak{A}} 
\hspace{5pt}\mathrm{iff}\hspace{5pt}
(b_{i_1},\cdots,b_{i_m})\in R^{\mathfrak{B}}.
\]
\item for any constant $c\in \sigma$ and any $i\in [k]$, 
\[
a_i=c^{\mathfrak{A}} \hspace{5pt} \mathrm{iff} \hspace{5pt} b_i=c^{\mathfrak{B}}.
\]
\end{enumerate}   

\noindent We assume that the readers have basic knowledge about first-order logic, 
especially what is the meaning of ``a formula is true in a structure''. 
Without loss of generality, we assume that all the formulas and sentences are
in \textit{negation normal form}, i.e. all negations can only occur immediately before 
atoms.

Let $\mathfrak{A}$ be a $\sigma$-structure and $\psi$ be a first-order sentence. We use $\mathfrak{A}\models \psi$ to denote that $\psi$ is true in $\mathfrak{A}$, and we call $\mathfrak{A}$ a \textit{model} for $\psi$. Let $\mathrm{Mod}(\psi)$ be the set of models of $\psi$.
A \textit{property} $Q$ over $\sigma$ is a set of $\sigma$-structures closed under isomorphism. Say that $Q$ is \textit{expressible}, or \textit{definable}, in $\fo$ if there is a sentence $\varphi$ in $\fo$ such that for every $\mathfrak{A}$,  $\mathfrak{A}\in\mathrm{Mod}(\varphi)$ iff $\mathfrak{A}\in Q$. 

A \textit{linear order} is a binary relation that is transitive, antisymmetric and total.   
Let $\tau$ be a signature. And let $\tau^{\text{ORD}}:= \tau \cup \{\leq\}$ where $\leq$ is interpreted in a $\tau^{\text{ORD}}$-structure as a linear order of its universe.

  \subsection{$\Gamma$-labeled forests}
 Let $n\in\mathbb{N}^+$. Given a graph $\mathcal{G}=(V,E)$, a \textit{directed path} $P$ in $\mathcal{G}$ is a sequence of vertices $(v_0,\cdots, v_{n})$  such that there is an arc from $v_i$ to $v_{i+1}$ for any $i<n$. The length of $P$ is $n$. A directed path is \textit{nontrivial} if the length of the path is nonzero. 

Trees are defined in the usual way in computer science. If there is an arrow from a node $a$ to a node $b$, then we call $a$ a \textit{father} of $b$ and $b$ a \textit{child} of $a$.  In a tree, each node has zero or more children and each node has at most one father. A node which has no father is called a \textit{root} and a node which has no child is called a \textit{leaf}. An \textit{inner node} is any node that has child nodes.  
A \textit{tree} is a connected acyclic digraph that has a root and some leaves. A \textit{degenerate tree} is a directed path. 

The \textit{height of a tree} is the length of a longest directed path in the tree. A \textit{forest} is composed of disjoint trees. 
Let $S$ be a forest. Define its height, denoted $h(S)$, as the maximum height of its trees. And define its rank, denoted $rk(S)$, as $h(S)+1$ when $S$ is not empty and 0 otherwise. Let $\Gamma=\{\exists,\forall\}$. A forest is a \textit{$\Gamma$-labeled forest} if all its nodes are labeled with ``$\exists$'' or ``$\forall$''. We call those nodes labeled with ``$\exists$'' $\mathscr{E}$ nodes and the other nodes $\mathscr{A}$ nodes. 

A\textit{ $\Gamma$-labeled perfect binary tree} is a $\Gamma$-labeled tree where each node, except the leaves, has exactly one $\mathscr{E}$ child and one $\mathscr{A}$ child, and all the leaves are at the same depth. 

A \textit{$\exists_n$-perfect binary tree}, denoted   $^*\mathcal{T}^{\exists}_n$, is a $\Gamma$-labeled perfect binary tree, whose root is labeled with $\exists$ and height is $(n-1)$. Likewise, a \textit{$\forall_n$-perfect binary tree}, denoted  $^*\mathcal{T}^{\forall}_n$, is a $\Gamma$-labeled perfect binary tree, whose root is labeled with $\forall$ and height is $(n-1)$.

\subsection{Prefixes}

The following terminology and conventions come from   Gr\"adel-McColm \cite{GradelM96} and Rosen \cite{rosen05prifix}.  A \textit{prefix} $p$ is a finite string in $\Gamma^*$.  The dual of $p$, denoted by $\bar p$, is the prefix obtained from $p$ by swapping $\exists$ with $\forall$.  Let $\mathcal{P}\subseteq \Gamma^*$. Then $\overline{\mathcal{P}}:=\{\bar p\in \Gamma^*\mid\hspace{3pt} p\in \mathcal{P}\}$. A prefix $p$ is a \textit{subsequence} of a prefix $q$ if $p$ can be obtained from $q$ by possibly deleting some elements of $q$, without changing the order of the remaining elements of $q$. A partial order on $\Gamma^*$, called prefix embedding, can be defined as follows: $p\preceq q$ iff $p$ is a subsequence of $q$. Here we use the curly symbol to distinguish it from the usual symbol of linear orders. Nevertheless, whether a symbol stands for a linear order or (prefix) embedding should be easily decided from the context. 
We use the same notation ``$\preceq$'' to denote the embedding relation between two sets of prefixes. For $\mathcal{P}_1$, $\mathcal{P}_2\subseteq\Gamma^*$, $\mathcal{P}_1\preceq \mathcal{P}_2\Leftrightarrow \forall p\in \mathcal{P}_1, \exists q\in \mathcal{P}_2$ s.t.\ $p\preceq q.$ $\mathcal{P}_1\prec\mathcal{P}_2$ if $\mathcal{P}_1\preceq \mathcal{P}_2$ but $\mathcal{P}_2\npreceq\mathcal{P}_1$. We use ``$*$'' to denote the concatenation of words. For any $\alpha\in \Gamma$ and $\mathcal{P}\subseteq \Gamma^*$, $\alpha*\mathcal{P}:=\{\alpha*p\mid p\in \mathcal{P}\}$. We define ${\mathcal{P}}^-:=\{p\mid \exists q\in \mathcal{P} \hspace{4pt} s.t.\ p\preceq q \}$ as the downward closure of $\mathcal{P}$. Let $\Gamma_c=\{\exists,\forall,\exists^*,\forall^*\}$ where $\exists^*$ and  $\forall^*$ are characters.  We interpret a word in $\Gamma_c^*$ as a regular expression. $\gamma: \Gamma_c^*\rightarrow \wp({\Gamma^*})$ maps such a regular expression to the regular language it denotes, where $ \wp({\Gamma^*})$ is the power set of $\Gamma^*$. We define $\gamma^-: \Gamma_c^*\rightarrow \wp(\Gamma^*)$ so that for any $v\in \Gamma_c^*$, $\gamma^-(v)=\{q\in \Gamma^*\mid$ there is $q^{\prime}\in \gamma(v)$ and $q\preceq q^{\prime}\}$, the downward closure of $\gamma(v)$. 

For a prefix $p$, $|p|$ is the length of $p$. $p[i]$ is the $i$-th letter of $p$. Let $l(p)$ be the last letter of $p$. For $0\leq i<|p|$, let $p^{-i}$ be the prefix obtained from $p$ by removing the first $i$ letters in $p$, i.e.\ $p=p[1]*\cdots p[i]*p^{-i}$.

Finally, let $\epsilon$ be the empty string.

\begin{lem}\label{conca-duals}
Let $p,q$ be prefixes. The following hold:
\begin{enumerate}
\item $\overline{\overline{p}}=p$.
\item  $\bar p*\bar q=\overline{p*q}$.
\end{enumerate}
\end{lem}
\begin{proof}\hfill
\begin{enumerate}
\item By definition.
\item Let $|p|=n$, $|q|=m$. 

By definition, $\bar p* \bar q=\overline{p[1]*\cdots*p[n]}*\overline{q[1]*\cdots*q[m]}=(\overline{p[1]}*\cdots*\overline{p[n]})*(\overline{q[1]}*\cdots*\overline{q[m]})$. Because the concatenation operation on  words satisfies the associative law, it follows that  $(\overline{p[1]}*\cdots*\overline{p[n]})*(\overline{q[1]}*\cdots*\overline{q[m]})=\overline{p[1]}*\cdots*\overline{p[n]}*\overline{q[1]}*\cdots*\overline{q[m]}$. Hence, by definition, $\bar p*\bar q=\overline{p*q}$.\qedhere 
\end{enumerate}
\end{proof}

\begin{defi} (Rosen, \cite{rosen05prifix}). \label{function-f} 
Let $s\in \Gamma$ and $s^i$ denote the string consisting of $i$ repetitions of $s$. Define $f: \Gamma^*\rightarrow \Gamma_c^*$ as follows: 
\begin{enumerate}
\item If $p=\exists^n$, then $f(p):=a_1*\cdots *a_{2n-1},$ where $a_i=\forall^*$ for $i$ odd, and $a_i=\exists$ for $i$ even;
\item If $p=\forall^n$, then $f(p):=a_1*\cdots *a_{2n-1},$ where $a_i=\exists^*$ for $i$ odd, and $a_i=\forall$ for $i$ even;
\item If $p=s_1^{i_1}*\cdots *s_n^{i_n}$ ($s_i\in\{\exists,\forall\}, s_i\neq s_{i+1}, i_j\in\mathbb{N}^+$), then $f(p):=f(s_1^{i_1})*\cdots * f(s_n^{i_n})$.
\end{enumerate} 
  \end{defi}

\begin{lem} \label{prefix}
Let $p$ be a prefix. Then $f(p)=\overline{f(\bar p)}.$
\end{lem}
\begin{proof}
 Assume that $p=s_1^{i_1}*\cdots *s_n^{i_n}$ ($s_i\in\{\exists,\forall\}, s_i\neq s_{i+1}, i_j\in\mathbb{N}^+$). By the definition of the dual of a prefix, $\bar p=\overline{s_1}^{i_1}*\cdots*\overline{s_n}^{i_n}$. Note that $\overline{s_i}\neq \overline{s_{i+1}}$ since $s_i\neq s_{i+1}$. By definition, $f(\bar p)=f(\overline{s_1}^{i_1})*\cdots*f(\overline{s_n}^{i_n})$.  Note that by definition $f(\overline{s_j}^{i_j})=\overline{f(s_j^{i_j})}$. Hence, $f(\bar p)=\overline{f(s_1^{i_1})}*\cdots*\overline{f(s_n^{i_n})}$.   
 By Lemma \ref{conca-duals} \textit{(ii)}, it means that $f(\bar p)=\overline{f(s_1^{i_1})*\cdots * f(s_n^{i_n})}=\overline{f(p)}$. Therefore,  $f(p)=\overline{f(\bar p)}$.
\end{proof}

\begin{lem} \label{roson}
(Rosen, \cite{rosen05prifix}). For every prefix $p\in \Gamma^*$, $f(p)$ is the unique word in $\Gamma_c^*$ such that $\gamma^-(f(p))=\{q\in\Gamma^*\mid p\npreceq q\}$.
\end{lem}

\begin{lem} \label{downward-closure-P}
For any $\mathcal{P}, \mathcal{P}_1, \mathcal{P}_2\subseteq \Gamma^*$, $\mathcal{P}^-=\mathcal{P}_1^-\cup\mathcal{P}_2^-$ if $\mathcal{P}=\mathcal{P}_1\cup\mathcal{P}_2$.
\end{lem}
\begin{proof}
If $p\in \mathcal{P}^-$, then there exists $q\in \mathcal{P}_1\cup\mathcal{P}_2$ such that $p\preceq q$. In other words, either $q\in \mathcal{P}_1$ or $q\in \mathcal{P}_2$ such that $p\preceq q$.  That is, $p\in\mathcal{P}_1^-$ or $p\in\mathcal{P}_2^-$. Hence, $p\in \mathcal{P}_1^-\cup \mathcal{P}_2^-$.

If $p\in \mathcal{P}_1^-\cup \mathcal{P}_2^-$, then either $p\in \mathcal{P}_1^-$ or $p\in \mathcal{P}_2^-$. That is, there exists $q\in \mathcal{P}_1$, or $q\in \mathcal{P}_2$, such that $p\preceq q$. In other words, there exists $q\in \mathcal{P}$ such that $p\preceq q$. Therefore, $p\in \mathcal{P}^-$.
\end{proof}

\subsection{Quantifier classes}
Let $\varphi$ be a $\fo$ formula. Recall that we assume that any formula is in negation normal form. Let $\uplus$ be the  disjoint union and $\displaystyle \biguplus_i S_i$ be the disjoint union of $S_i$. 

If $\Phi$ is a set of formulas, then $\bigwedge \Phi$ ($\bigvee \Phi$ resp.) is the conjunction (disjunction resp.) of all the formulas in $\Phi$. Similarly, we use $\displaystyle\bigwedge_i \theta_i$ ($\displaystyle\bigvee_i \theta_i$ resp.) to represent the conjunction (disjunction resp.) of all $\theta_i$.  
\begin{defi}\label{quantifier-structure-1}
The \textit{quantifier structure} of $\varphi$, denoted $qs(\varphi)$, is a $\Gamma$-labeled forest, which is defined inductively as follows: 
\begin{itemize}
\item If $\varphi$ is a literal, then $qs(\varphi)$ is empty;
\item If $\varphi=\displaystyle \bigwedge_i \theta_i$ or $\displaystyle \bigvee_i \theta_i$, then $qs(\varphi)$ is $\displaystyle\biguplus_i qs(\theta_i)$;
\item If $\varphi=\exists x \theta,$ then $qs(\varphi)$ is composed of an $\mathscr{E}$ node and $qs(\theta)$ where there is an arc from this $\mathscr{E}$ node  to each root of $qs(\theta)$ (note that $qs(\theta)$ is a forest); 

Similarly, If $\varphi=\forall x \theta,$ then $qs(\varphi)$ is composed of an $\mathscr{A}$ node and $qs(\theta)$ where  there is an arc from this $\mathscr{A}$ node to each root of $qs(\theta)$; 

In these two cases, if $qs(\theta)$ is empty, then $qs(\varphi)$ contains a single node.
\end{itemize}
  \end{defi}  
%\cbstart
Note that this definition is  different from Gr\"{a}del and McColm's \cite{GradelM96}, in which $qs(\varphi)$ is defined as a set of strings: 
\begin{itemize}
\item If $\psi$ is a literal, then $qs(\psi)=\{\epsilon\}$ where $\epsilon$ is the empty word;

\item If $\psi:=\bigwedge \Phi$ or $\psi:=\bigvee \Phi$ where $\Phi$ is a set of formulas, then 
\[
   qs(\psi):=\bigcup_{\varphi\in \Phi} qs(\varphi);
\]
\item $\psi:=\exists x_i \varphi$, then $qs(\varphi):=\exists *qs(\varphi)$; likewise, if $\psi:=\forall x_i \varphi$, then $qs(\psi):=\forall *qs(\varphi)$.
\end{itemize}

%\cbend
\begin{defi}\label{forests-embedding-1}
Let $S_1$, $S_2$ be two  $\Gamma$-labeled forests. Define $S_1\preceq_e S_2$ if there is a mapping $\iota$, \textit{not necessarily injective}, from the nodes of $S_1$ to the nodes of $S_2$ such that $v$ and $\iota(v)$ have the same label for any $v$, and there is a nontrivial directed path from $\iota(x)$ to $\iota(y)$ in $S_2$ if there is an arc from node $x$ to node $y$ in $S_1$. 
  \end{defi}

\begin{rem}

Note that the relation $\preceq_e$ is not necessary antisymmetric. That is, there are non-isomorphic $\Gamma$-labeled forests $S_1, S_2$ such that $S_1\preceq_e S_2$ and $S_2\preceq_e S_1$.  
       
\end{rem}
\begin{defi}\label{read-off}
Suppose that we are given a $\Gamma$-labeled forest $S$. For any path $P:=(v_0,\cdots,v_n)$ in the forest, there is a word $(s_0,\cdots,s_n)$ in $\Gamma^*$ associated with it such that the node $v_i$ is labeled with $s_i$. We say that this word, as well as all its subsequences, can be \textit{read off} this $\Gamma$-labeled forest.  Let $\mathscr{W}(\mathcal{S})$ be the set of words that can be read off the forest $\mathcal{S}$.
  \end{defi}

\begin{defi}\label{F(P)}
Suppose that we are given a set $\mathcal{P}\subseteq \Gamma^*$. Let $\mathcal{P}=\mathcal{P}_\exists\cup \mathcal{P}_\forall$, where $\mathcal{P}_\exists=\exists*\mathcal{P}_1$  and $\mathcal{P}_\forall=\forall*\mathcal{P}_2$. These  sets can be empty. We can inductively define a $\Gamma$-labeled forest $\mathscr{F}(\mathcal{P})$ as follows: \begin{enumerate}
\item If $\mathcal{P}$ is empty, then $\mathscr{F}(\mathcal{P})$ is empty, i.e.\ this forest contains no node.
\item Let $S_1$ be a $\Gamma$-labeled forest such that its root is an $\mathscr{E}$ node and there is an arc from this root to each root of $\mathscr{F}({\mathcal{P}_1})$. Likewise, let $S_2$ be a $\Gamma$-labeled forest such that its root is an $\mathscr{A}$ node and there is an arc from this root to each root of $\mathscr{F}(\mathcal{P}_2)$.
\item $\mathscr{F}(\mathcal{P})$ is the disjoint union of $S_1$ and $S_2$.
\end{enumerate}
  \end{defi}

Note that $\mathscr{F}(P)$ is composed of at most two trees, the roots of which have different labels.
\begin{lem}\label{forest2word}
$\mathscr{W}(\mathscr{F}(\mathcal{P}))=\mathcal{P}^-$, for all $\mathcal{P}\subseteq \Gamma^*$.
\end{lem}
\begin{proof}
The base case when $\mathscr{F}(\mathcal{P})$ is empty, i.e.\ when $rk(\mathscr{F}(\mathcal{P}))=0$, is trivial. 

Assume that it holds when $rk(\mathscr{F}(\mathcal{P}))\leq k$ for some $k\geq 0$. 

Assume that \[rk(\mathscr{F}(\mathcal{P}))=k+1 \text{ and }
\mathcal{P}=\mathcal{P}_\exists\cup \mathcal{P}_\forall,\]
where \[\mathcal{P}_\exists=\exists*\mathcal{P}_1  \text{ and }
\mathcal{P}_\forall=\forall*\mathcal{P}_2.\]
Clearly, \[rk(\mathscr{F}(\mathcal{P}_1))\leq k \mbox{ and }
rk(\mathscr{F}(\mathcal{P}_2))\leq k.\]\vspace{-2 pt}

\noindent According to Definition \ref{F(P)}, $\mathscr{F}(\mathcal{P})$ is the disjoint union of $\mathscr{F}(\mathcal{P}_\exists)$ and $\mathscr{F}(\mathcal{P}_\forall)$. In other words, $\mathscr{W}(\mathscr{F}(\mathcal{P}))$ equals $\mathscr{W}(\mathscr{F}(\mathcal{P}_\exists))\cup \mathscr{W}(\mathscr{F}(\mathcal{P}_\forall))$, hence equals \[(\exists*\mathscr{W}(\mathscr{F}(\mathcal{P}_1)))\cup \mathscr{W}(\mathscr{F}(\mathcal{P}_1))\cup (\forall*\mathscr{W}(\mathscr{F}(\mathcal{P}_2)))\cup \mathscr{W}(\mathscr{F}(\mathcal{P}_2)),\] and by assumption  equals \[(\exists*\mathcal{P}_1^-)\cup\mathcal{P}_1^-\cup(\forall*\mathcal{P}_2^-)\cup\mathcal{P}_2^-=\mathcal{P}_\exists^-\cup \mathcal{P}_\forall^-.\] By Lemma \ref{downward-closure-P}, $\mathcal{P}^-=\mathcal{P}_\exists^-\cup \mathcal{P}_\forall^-$.  Therefore, $\mathscr{W}(\mathscr{F}(\mathcal{P}))=\mathcal{P}^-$.
\end{proof}

\begin{rem} This lemma implies that, for any $p\in \mathcal{P}^-\subseteq \Gamma^*$, $p$ can be read off from some path of $\mathscr{F}(\mathcal{P})$.
       
\end{rem}

\begin{lem}\label{words2forest}
For any $\Gamma$-labeled forest $S$ and $\mathcal{P}\subseteq \Gamma^*$, $\hspace{2pt} S\preceq_e \mathscr{F}(\mathcal{P})$ if $\mathscr{W}(S)\subseteq \mathcal{P}^-$.
\end{lem}
\begin{proof}
The base case when $S$ is empty is trivial. 

Assume that it holds when $rk(S)\leq k$ where $k\geq 0$. 

Let $S$ be a $\Gamma$-labeled forest such that  $rk(S)=k+1$ and $\mathscr{W}(S)\subseteq \mathcal{P}^-$. $S$ is a disjoint union of at most two forests $S_\exists$ and $S_\forall$: the roots of $S_\exists$ are all $\mathscr{E}$ nodes and the roots of $S_\forall$ are all $\mathscr{A}$ nodes. Then $\mathscr{W}(S)=\mathscr{W}(S_\exists)\cup \mathscr{W}(S_\forall)$. Note that a substructure of a forest is also a forest. Because $\mathscr{W}(S)\subseteq \mathcal{P}^-$ and Lemma \ref{forest2word}, $\mathscr{W}(S)\subseteq \mathscr{W}(\mathscr{F}(P))$.  It means that there is a forest (substructure) $\mathcal{F}_\exists$ of $\mathscr{F}(P)$ such that all its roots are $\mathscr{E}$ nodes and that $\mathscr{W}(S_\exists)\subseteq\mathscr{W}(\mathcal{F}_\exists)$. Likewise,   there is a forest (substructure) $\mathcal{F}_\forall$ of $\mathscr{F}(P)$ such that all its roots are $\mathscr{A}$ nodes and that $\mathscr{W}(S_\forall)\subseteq\mathscr{W}(\mathcal{F}_\forall)$. Note that $\mathcal{F}_\exists$ and $\mathcal{F}_\forall$ are not necessary disjoint.

Now, if we remove all the $\mathscr{E}$ roots from $S_\exists$, we get a forest called $S_1$. Similarly, if we remove all the $\mathscr{A}$ roots from $S_\forall$, we get another forest called $S_2$. 

Likewise, if we remove all the $\mathscr{E}$ roots from $\mathcal{F}_\exists$, we get a forest called $\mathcal{F}_1$. Similarly,   if we remove all the $\mathscr{A}$ roots from $\mathcal{F}_\forall$, we get another forest called $\mathcal{F}_2$.

Observe that $\mathscr{W}(S_1)\subseteq \mathscr{W}(\mathcal{F}_1)\subseteq \mathscr{W}(\mathcal{F}_1)^-$ and $rk(S_1)\leq k$. By assumption, $S_1\preceq_e\mathcal{F}_1$. Let us denote the map that embeds $S_1$ to $\mathcal{F}_1$ as $\iota_1$. Likewise, $S_2\preceq_e\mathcal{F}_2$ and the embedding map is denoted $\iota_2$. Note that the domains of $\iota_1$ and $\iota_2$ are different.   Therefore, we can merge these two maps easily, i.e.\ let $\iota_0=\iota_1\cup \iota_2$. Note that $S_1$ ($S_2$ resp.) is embeddable to $\mathcal{F}_1$ ($\mathcal{F}_2$ resp.) through $\iota_0$.  Now, we can extend the embedding map $\iota_0$ to $\iota$ such that: (i) the father of any root $r_1$ of $S_1$ is mapped to the father of $\iota_0(r_1)$; (ii) the father of any root $r_2$ of $S_2$ is mapped to the father of $\iota_0(r_2)$. Therefore, $S$ is embeddable to $\mathscr{F}(P)$ through $\iota$, i.e.\ $S\preceq_e \mathscr{F}(P)$. 
\end{proof}
\begin{rem} Lemma \ref{forest2word} and Lemma \ref{words2forest} tell us that $\mathscr{F}(P)$ is the ``maximal'' $\Gamma$-labeled forest (in the sense of embeddings) among all these forests, from which the set of words that can be read off is a subset of $\mathcal{P}^-$. 
       
\end{rem}

Define the \textit{quantifier rank} of $\fo$ formula $\varphi$, denoted  $qr(\varphi)$, to be $rk(qs(\varphi))$. Note that this definition is equivalent to the usual definition of quantifier rank (see for instance Libkin \cite{Libkin04Elements}). Let $\fo[k]:=\{\varphi\in\fo\mid qr(\varphi)\leq k\}$.

 Let $S$ be a $\Gamma$-labeled forest. Define the quantifier class $\fo\{S\}$ to be the set of queries that are definable by the set of first-order sentences $\{\theta\in \fo\hspace{1pt}| \hspace{3pt} qs(\theta))\preceq_e S\}$. 

A first-order formula is in prenex normal form if it is  a single string of quantifiers  followed by a quantifier free formula. Its quantifier prefix, which is obtained from this string of quantifiers by removing the variables in the string, corresponds to a $\Gamma$-labeled degenerate tree.

Given a prefix $p$, we define the prefix class $\fo(p)$ as the set of $\fo$ sentences in prenex normal form such that for any $\psi\in\fo(p)$, its prefix is a subsequence of $p$ (Gr\"adel and McColm, \cite{GradelM96}). Gr\"adel-McColm's conjecture says that the prefix classes form a strict hierarchy: For any prefix $p,q$, $\fo(p)\nsubseteq \fo(q)$ if $p\npreceq q$ over arbitrary structures. Rosen \cite{rosen05prifix} confirmed this conjecture over infinite structures and called it the first-order prefix hierarchy. \label{Grdel-McColm-conjecture} 
Similarly, we can define a hierarchy formed by quantifier classes, which can be called the first-order quantifier structure hierarchy. 
These two hierarchies are independent.

\section{Quantifier Structure Hierarchy: the first observation}\label{games-and-hierarchy}
In this section, we define a variant of Ehrenfeucht-Fra\"iss\' e games that characterizes quantifier classes and prove that  those quantifier classes form a natural and strict hierarchy:
\begin{thm}\label{main1}
Let $S_1$ and $S_2$ be two $\Gamma$-labelled forests. Over the class of all digraphs,
\[
         \mathrm{if}\hspace{3pt} \mathscr{W}(S_1)\nsubseteq \mathscr{W}(S_2), \hspace{2pt} \mathrm{then}\hspace{3pt} \fo\{S_1\}\nsubseteq \fo\{S_2\}.
\]
\end{thm} 
 
\subsection{Games that characterize quantifier classes}\label{games}
Let $S$ be a $\Gamma$-labelled forest. We define an asymmetric variant of the Ehrenfeucht-Fra\"iss\' e games as follows. Let $k\in\mathbb{N}$, $\sigma$ contains $k$ constant symbols. Let $\mathfrak{A}$ and $\mathfrak{B}$ be two $\sigma$-structures. Let the $k$-tuple $\overline{\mathfrak{u}}$ be the interpretation of the constants in $\mathfrak{A}$ and the $k$-tuple $\overline{\mathfrak{v}}$ be the interpretation of the constants in $\mathfrak{B}$.  The game $G_{S}(\mathfrak{A},\mathfrak{B})$ is played by two players, called the spoiler and duplicator, on a game board consisting of $S$, $\mathfrak{A}$ and $\mathfrak{B}$.  At the beginning of the game, the spoiler picks a tree $\mathcal{T}$ in the forest $S$ and puts a token on the root of $\mathcal{T}$. Assume that the depth of $\mathcal{T}$ is $n-1$. Afterwards, for every $i$ where $1\leq i\leq n$, in the $i$-th round the spoiler chooses an element from the structure $\mathfrak{A}$ if the current node, on which the token is put, is an $\mathscr{E}$ node. Otherwise if it is an $\mathscr{A}$ node he picks an element of $\mathfrak{B}$. Then the duplicator has to respond by picking an element from the other structure. Afterwards the spoiler chooses a child of the current node in $S$  and moves the token to it. This completes one round.  

Assume that after $n^{\prime}$ ($n^{\prime}\leq n$) rounds a sequence $\bar c=(c_1,\cdots,c_{n^{\prime}})$ has been picked in $\mathfrak{A}$ and a sequence $\bar d=(d_1,\cdots,d_{n^{\prime}})$ has been picked in $\mathfrak{B}$. The spoiler wins the game if $(\overline{\mathfrak{u}}\bar c,\overline{\mathfrak{v}}\bar d)$ does not define a partial isomorphism between $\mathfrak{A}$ and $\mathfrak{B}$. 

The game ends whenever the spoiler wins or the token arrives at a leaf of $\mathcal{T}$. The duplicator wins if the spoiler fails to win in the end. 

Informally, $\overline{\mathfrak{u}}$ and $\overline{\mathfrak{v}}$ can be regarded as a carry-over of past history of the game played before the beginning, which has to be taken care of. 

A strategy of the duplicator is a scheme by which she knows how to choose an element in each round depending on the history of the play.

 For any tuple $(\mathfrak{a}_1,\cdots,\mathfrak{a}_n)\in (|\mathfrak{A}|\uplus |\mathfrak{B}|)^n$, we associate it with a prefix $(p[1],\cdots,p[n])$ such that $\mathfrak{a}_i\in |\mathfrak{A}|$ iff $p[i]=\exists$ for any $1\leq i\leq n$. A \textit{strategy of the duplicator} in the game $G_{S}(\mathfrak{A},\mathfrak{B})$ is a function\\[4pt]
\indent\hspace{20pt}$D^{S}_{\mathfrak{A},\mathfrak{B}}:\biguplus_{i=1}^{h(S)+1} (|\mathfrak{A}|\uplus |\mathfrak{B}|)^i\rightarrow (|\mathfrak{A}|\uplus |\mathfrak{B}|)$.\\[-4pt]  

If the duplicator has a strategy guiding her choices in the game that ensures her winning in the end no matter how the spoiler plays, we call this strategy  a winning strategy of the duplicator. If there exists such a winning strategy for the duplicator in the game $G_{S}(\mathfrak{A},\mathfrak{B})$ then we write $\mathfrak{A}\leadsto_{S} \mathfrak{B}$. The winning strategy of the spoiler can be defined dually because in our games either the spoiler or the duplicator has a winning strategy. Let $\bar a\in |\mathfrak{A}|^t$ and $\bar b\in |\mathfrak{B}|^t$. We use $(\mathfrak{A},\bar a)\leadsto_{S} (\mathfrak{B},\bar b)$ to denote that the duplicator has a winning strategy, in which $\bar a$ is picked in $\mathfrak{A}$, and $\bar b$ is picked in $\mathfrak{B}$, before the game starts. Equivalently, we say that the spoiler has a winning strategy in the game $G_S((\mathfrak{A},\bar a),(\mathfrak{B},\bar b))$.

Note that the standard Ehrenfeucht-Fra\"iss\' e game $G_n(\mathfrak{A},\mathfrak{B})$ (see \cite{Libkin04Elements}) is exactly the game $G_{S}(\mathfrak{A},\mathfrak{B})$ where $S=\{^*\mathcal{T}^{\exists}_n,^*\!\mathcal{T}^{\forall}_n\}$.

\begin{defi}
Let $n\in \mathbb{N}$ and $\bar c$ be an $n$-tuple of elements from  $|\mathfrak{A}|$. Then for a $\Gamma$-labelled forest $S$,  the $QS$-$S$ $n$-type of $\bar c$ over $\sigma$-structure $\mathfrak{A}$ is defined as:
\[tp^S_n(\mathfrak{A},\bar c)=\{\varphi(\bar c)\in \fo\{S\} \mid \mathfrak{A}\models \varphi(\bar c)\}.\]
    \end{defi} 

The following lemma is well-known, cf. \cite{Libkin04Elements} for a simple explanation.   
\begin{lem}\label{finite-logic-equiv-fixed-rank} 
For fixed $k,n\in\mathbb{N}$, there are only finitely many formulas, in $n$ free variables, in $\fo[k]$ up to logical equivalence.
\end{lem}
\begin{cor} \label{finite-formu-in-type}Let  $n\in\mathbb{N}$ and $S$ be a $\Gamma$-labelled forest, and let $\bar c$ be an $n$-tuple of elements from  $|\mathfrak{A}|$, there are only finitely many formulas in $tp^{S}_n(\mathfrak{A},\bar c)$ up to logical equivalence.
\end{cor}

Let $S$ be a $\Gamma$-labelled forest.
\begin{defi}
Let $[\mathfrak{A},\bar a]$ be an expansion of $\mathfrak{A}$ to $\sigma\cup\{\overline{\mathfrak{c}}\}$ such that $\bar a$ interprets the tuple of constants $\overline{\mathfrak{c}}$ in $[\mathfrak{A},\bar a]$.
    \end{defi}

\begin{lem} \label{m2m-}
Let $S$ and $S^{\prime}$ be two $\Gamma$-labelled forests such that $S^{\prime}\preceq_e S$. Let $\mathfrak{A}$ and $\mathfrak{B}$ be two structures over the same signature. If the duplicator has a winning strategy in the game $G_S(\mathfrak{A},\mathfrak{B})$, then she also has a winning strategy in the game $G_{S^{\prime}}(\mathfrak{A},\mathfrak{B})$.
\end{lem}
\begin{proof}
Note that the duplicator can mimic her winning strategy in the game $G_{S^{\prime}}(\mathfrak{A},\mathfrak{B})$ to play the game 
$G_{S^{\prime}}(\mathfrak{A},\mathfrak{B})$. And a subset of a partial isomorphism is still a partial isomorphism.
\end{proof}

\begin{rem} Lemma \ref{m2m-} tells us that if the duplicator has a winning strategy in $G_S(\mathfrak{A},\mathfrak{B})$, she also has a winning strategy when the players are allowed to skip playing arbitrary rounds of the game. Lemma \ref{m2m-} also tells us that if the spoiler has a winning strategy in $G_{S^{\prime}}(\mathfrak{A},\mathfrak{B})$, he also has a winning strategy in $G_S(\mathfrak{A},\mathfrak{B})$. In other words, quantifiers are logical resources that can be exploited by the spoiler to detect the difference between two structures in the games.
             
\end{rem}

\indent It is obvious that $[\mathfrak{A},\bar a]\leadsto_{S} [\mathfrak{B},\bar b]$ iff $(\mathfrak{A},\bar a)\leadsto_{S} (\mathfrak{B},\bar b)$, because in each round of both the game $G_S([\mathfrak{A},\bar a],[\mathfrak{B},\bar b])$ and $G_S((\mathfrak{A},\bar a),(\mathfrak{B},\bar b))$, if there is a partial isomorphism between two structures in the former, then this partial isomorphism is also a partial isomorphism between two structures in the latter.

In the following, we prove a connection between the games just defined and quantifier classes, which is a variant of the result of Gr\"adel and McColm \cite{GradelM96}.

Recall that $\mathfrak{A}$ has constants that are interpreted by $\overline{\mathfrak{u}}$. And $\mathfrak{B}$ has constants that are interpreted by $\overline{\mathfrak{v}}$. We assume that $\overline{\mathfrak{u}}$ ($\overline{\mathfrak{v}}$ resp.) and $\bar a$ ($\bar b$ resp.) do not share any element.

\begin{thm}\label{qs-game}
For arbitrary finite $\sigma$-structures $\mathfrak{A}$,
$\mathfrak{B}$, two tuples $\bar a\in |\mathfrak{A}|^t$, $\bar b\in
|\mathfrak{B}|^t$, and a $\Gamma$-labelled forest $S$, the following
are equivalent:
\begin{enumerate}[label=(\roman*)]
\item $[\mathfrak{A},\bar a]\leadsto_{S} [\mathfrak{B},\bar b]$;
\item $tp^{S}_t (\mathfrak{A},\bar a)\subseteq tp^{S}_t (\mathfrak{B},\bar b)$.
\end{enumerate}
\end{thm}
\begin{proof}
\noindent{(i)}$\rightarrow${(ii)}:

When $S$ is an empty forest, i.e.\ $rk(S)=0$, $tp^{S}_t (\mathfrak{A},\bar a)\nsubseteq tp^{S}_t (\mathfrak{B},\bar b)$ means there is a quantifier-free formula $\eta(\bar x)$  such that $(\mathfrak{A},\bar a)\models\eta(\bar x)$ but $(\mathfrak{B},\bar b)\not\models\eta(\bar x)$. Hence, the mapping from $\overline{\mathfrak{u}}\bar a$ to $\overline{\mathfrak{v}}\bar b$ does not define a partial isomorphism. In other words, the spoiler wins the game and $[\mathfrak{A},\bar a]\not\leadsto_S [\mathfrak{B},\bar b]$. 

Assume that {(i)}$\rightarrow${(ii)} when $rk(S)\leq k$ for $k\geq 0$.

Assume that $rk(S)=k+1$ and $S$ consists of $m$ trees  $S_1,\cdots,S_m$. Suppose that $(ii)$ is false. Let $\varphi(\bar x)\in \fo\{S\}$ such that $(\mathfrak{A},\bar a)\models\varphi(\bar x)$ but $(\mathfrak{B},\bar b)\not\models\varphi(\bar x)$. Then $\varphi$ is a first-order formula that is a disjunction or conjunction of formulas $\fo\{S_i\}$ ($1\leq i\leq m$). There must exist one disjunct or conjunct $\psi$ such that $(\mathfrak{A},\bar a)\models\psi(\bar x)$ while $(\mathfrak{B},\bar b)\not\models\psi(\bar x)$, where $qs(\psi)\preceq_e S_i$ for some $1\leq i\leq m$. By Lemma \ref{m2m-}, we may assume that the spoiler moves the token from the root of $S_i$. Assume that the root of $S_i$ is an $\mathscr{E}$ node, then $\psi$ has the form $\exists y \psi_1(\bar xy)$. Hence, there is $c\in |\mathfrak{A}|$ s.t.\ $(\mathfrak{A},\bar ac)\models \psi_1(\bar xy)$. Then the spoiler can pick $c$, and no matter which element, say $d$, the duplicator picks, $\psi_1(\bar xy)$ distinguishes the pair $[\mathfrak{A},\bar ac]$ and $[\mathfrak{B},\bar bd]$, where the variables $\bar xy$ are assigned the values $\bar ac$ and $\bar bd$ respectively, because $(\mathfrak{B},\bar b)\not\models \exists y \psi_1(\bar xy)$. By induction assumption the spoiler has a winning strategy over the game $G_{qs(\psi_1)}([\mathfrak{A},\bar ac],[\mathfrak{B},\bar bd])$.
Similarly, if $S_i$ is a tree whose root is an $\mathscr{A}$ node, the spoiler can pick $d\in |\mathfrak{B}|$ such that for any $c\in |\mathfrak{A}|$ picked by the duplicator $\psi_1(\bar xy)$ distinguishes  the pair $[\mathfrak{A},\bar ac]$ and $[\mathfrak{B},\bar bd]$ where the variables $\bar xy$ are assigned the values $\bar ac$ and $\bar bd$ respectively. In other words, the spoiler can show that there is an element $d$ such that it makes $(\mathfrak{B},\bar bd)\not\models \psi_1(\bar xy)$ while $(\mathfrak{A},\bar ac)\models \psi_1(\bar xy)$ is always true.  By induction assumption the spoiler has a  winning strategy over the game $G_{qs(\psi_1)}([\mathfrak{A},\bar ac],[\mathfrak{B},\bar bd])$. Therefore, $[\mathfrak{A},\bar a]\not\leadsto_S[\mathfrak{B},\bar b]$. \\[5pt] 
\noindent {(ii)}$\rightarrow${(i)}:  

According to the definition of the game, when $rk(S)=0$, $[\mathfrak{A},\bar a]\not\leadsto_S [\mathfrak{B},\bar b]$ means that the mapping from $\overline{\mathfrak{u}}\bar a$ to $\overline{\mathfrak{v}}\bar b$ is not a partial isomorphism, which implies that  there exists a quantifier-free $\fo$ formula $\xi(\bar x)$ s.t.\ either ($(\mathfrak{A},\bar a)\models \xi(\bar x)$ but $(\mathfrak{B},\bar b)\not\models \xi(\bar x)$), or ($(\mathfrak{A},\bar a)\not\models \xi(\bar x)$ but $(\mathfrak{B},\bar b)\models \xi(\bar x)$). Let $\lnot \xi(\bar x)$ be the negation of $\xi(\bar x)$. In the former case, it implies that $tp^{S}_t (\mathfrak{A},\bar a)\nsubseteq tp^{S}_t (\mathfrak{B},\bar b)$; in the latter case, $(\mathfrak{A},\bar a)\models \lnot \xi(\bar x)$ but $(\mathfrak{B},\bar b)\not\models \lnot\xi(\bar x)$, which also implies that $tp^{S}_t (\mathfrak{A},\bar a)\nsubseteq tp^{S}_t (\mathfrak{B},\bar b)$. All in all, $[\mathfrak{A},\bar a]\not\leadsto_{S}[\mathfrak{B},\bar b]$ implies $tp^{S}_t (\mathfrak{A},\bar a)\nsubseteq tp^{S}_t (\mathfrak{B},\bar b)$ when $rk(S)=0$.

Assume that {(ii)}$\rightarrow${(i)} when $rk(S)=k$ for $k\geq 0$.

Now assume that $S$ is composed of trees $S_1,\cdots,S_m$ and $rk(S)=k+1$. Assume that $(i)$ is false. Then over one of the trees the spoiler has a winning strategy. Hence, the spoiler can first pick this tree to play. If this tree's root is an $\mathscr{E}$ node $r$, we can regard this tree as a digraph composed of $r$ and a forest $S^{\prime}$ such that there is an arc from $r$ to each root of $S^{\prime}$. In the first round the spoiler can pick an element $c\in |\mathfrak{A}|$ such that no matter which element $d\in |\mathfrak{B}|$ the duplicator picks, $[\mathfrak{A},\bar ac]$ and $[\mathfrak{B},\bar bd]$ form the new game board over which the spoiler will win the game $G_{S^{\prime}}([\mathfrak{A},\bar ac],[\mathfrak{B},\bar bd])$. By Corollary \ref{finite-formu-in-type}, there are only finitely many formulas in $tp^{S^{\prime}}_{t+1}(\mathfrak{A},\bar ac)$ up to logical equivalence. Let $T/E$ be a set of formulas where each equivalent class in $tp^{S^{\prime}}_{t+1}(\mathfrak{A},\bar ac)$ has exactly one formula in $T/E$. Let $\varphi(\bar xy)$ be the conjunction of all the formulas in $T/E$. By the induction hypothesis,  for any $d$ there is a formula $\eta(\bar xy)\in \fo\{S^{\prime}\}$ such that $(\mathfrak{A},\bar ac)\models \eta(\bar xy)$ but  $(\mathfrak{B},\bar bd)\not\models \eta(\bar xy)$. Note that $\eta$ is equivalent to one formula in $T/E$. 
Hence, $(\mathfrak{A},\bar ac)\models \varphi(\bar xy)$ but $(\mathfrak{B},\bar bd)\not\models \varphi(\bar xy)$,  for any $d$. 
In other words, $(\mathfrak{A},\bar a)\models \exists y \varphi(\bar xy)$ but  $(\mathfrak{B},\bar b)\not\models \exists y \varphi(\bar xy)$. 
Note that $\exists y\varphi(\bar xy)\in \fo\{S\}$. Therefore, $tp^{S}_t (\mathfrak{A},\bar a)\nsubseteq tp^{S}_t (\mathfrak{B},\bar b)$. 

The case when the tree picked by the spoiler in the first step  is a tree whose root is an $\mathscr{A}$ node can be proved similarly.
\end{proof}

\begin{cor} \label{q-type-game-2} 
Let $\mathcal{K}$ be a class of finite structures and $S$ be a $\Gamma$-labelled forest. If there is $\mathfrak{A}\in \mathcal{K}$ and $\mathfrak{B}\notin \mathcal{K}$ such that $\mathfrak{A}\leadsto_S \mathfrak{B}$, then there is no first-order sentence $\varphi$ such that $qs(\varphi)\preceq_e S$ and $\mathcal{K}=\mathrm{Mod}(\varphi)$.
\end{cor}

\subsection{Point-expansions}
\begin{defi} \label{point-expansion}
Let $\mathfrak{A}$ be a structure  over signature $\sigma_A$ and $\mathcal{K}$ be a set $\{\mathfrak{C}_i\}_{i\in I}$ of finite structures indexed by a set $I$. Let $\sigma_{i}$ be the signature of $\mathfrak{C}_i$ for each $i\in I$ such that $\sigma_i$ contains a special constant $c_i$ that is called {hook}. Assume that no two signatures share a constant. Let $\gimel_{\mathfrak{A}}: |\mathfrak{A}|\rightarrow \mathcal{K}$ be a total function. Define the {point-expansion} of $\mathfrak{A}$ by $\gimel_{\mathfrak{A}}$ over $\mathcal{K}$, denoted $\mathcal{E}_{\mathcal{K}}^{\gimel_{\mathfrak{A}}}(\mathfrak{A})$, as follows: \begin{enumerate}
\item The signature of $\mathcal{E}_{\mathcal{K}}^{\gimel_\mathfrak{A}}(\mathfrak{A})$, denoted $\sigma_{\mathcal{E}}$, is composed of the union of $\sigma_{i}$ and $\sigma_A$,  except for the hook in $\sigma_i$ for any $i$. 
\item Let $M:=\displaystyle\biguplus_{a\in |\mathfrak{A}|} \gimel_{\mathfrak{A}}(a)$ and $|\mathcal{E}_{\mathcal{K}}^{\gimel_{\mathfrak{A}}}(\mathfrak{A})|:=|M|$.
\item Let $M^{\prime}$ be the set of elements that interpret the hooks in respective disjoint substructures.
\item There is a bijection $g: M^{\prime}\rightarrow |\mathfrak{A}|$ such that for any $k$-tuple $\bar v=(v_1,\cdots, v_k)\in |M^{\prime}|^k$ and any $k$-ary relation $R\in \sigma_A$, 

\[\bar v \in R^{\mathcal{E}_{\mathcal{K}}^{\gimel_{\mathfrak{A}}}(\mathfrak{A})} \hspace{3pt}\mathrm{iff}\hspace{3pt} (g(v_1),\cdots, g(v_k))\in R^{\mathfrak{A}}.\]

In other words, $M^{\prime}$ induces in $\mathcal{E}_{\mathcal{K}}^{\gimel_{\mathfrak{A}}}(\mathfrak{A})$ an isomorphic copy of $\mathfrak{A}$. Let \[R_g:=\{\bar v\in |M^{\prime}|^k\mid (g(v_1),\cdots, g(v_k))\in R^{\mathfrak{A}}\}.\] 
\item Let $\{\mathcal{X}_i^{\hat{\mathcal{R}}}\}$ be the set of structures in $M$ whose signatures contain $\hat{\mathcal{R}}$.
\[
\hat{\mathcal{R}}^{\mathcal{E}_{\mathcal{K}}^{\gimel_{\mathfrak{A}}}(\mathfrak{A})}:=
\begin{cases}
   \displaystyle \bigcup_i \hat{\mathcal{R}}^{\mathcal{X}_i^{\hat{\mathcal{R}}}}, &\hat{\mathcal{R}}\notin\sigma_A;\\
   \displaystyle \bigcup_i \hat{\mathcal{R}}^{\mathcal{X}_i^{\hat{\mathcal{R}}}}\cup \hat{\mathcal{R}}_g, &\hat{\mathcal{R}}\in\sigma_A.\\
\end{cases}   
\]
\item For any constant $c\in \sigma_{\mathcal{E}}$, \[c^{\mathcal{E}_{\mathcal{K}}^{\gimel_{\mathfrak{A}}}(\mathfrak{A})}:=
\begin{cases}
      c^{\mathfrak{C}_i}, & c\in \sigma_i\setminus \sigma_A;\\
      g^{-1}(c^{\mathfrak{A}}), & c\in \sigma_A.\\ 
\end{cases}
\]
\end{enumerate} 
    \end{defi}

\noindent Informally speaking, $\mathcal{E}_{\mathcal{K}}^{\gimel_{\mathfrak{A}}}(\mathfrak{A})$ is a structure that is obtained from $\mathfrak{A}$ by substituting each element $a\in|\mathfrak{A}|$ with $\gimel_{\mathfrak{A}}(a)$, identifying $a$ with the hook in $\gimel_{\mathfrak{A}}(a)$. A point-expansion of the structure $\mathfrak{A}$ can also be regarded as the result of a process that ``glues'' a small substructure $\gimel_{\mathfrak{A}}(a)$ at each element $a$ (also at the hook of $\gimel_{\mathfrak{A}}(a)$) of the ``prototype'' structure $\mathfrak{A}$. And each small substructure shares only one element with the prototype structure, i.e.\ the ``point'' where they are ``glued'' together. 

\begin{lem} \label{composition_of_strategies}
Suppose that we are given a forest $\mathcal F$, two structures $\mathfrak{A}$, $\mathfrak{B}$, a finite set $\mathcal{K}$ of structures and two mappings $\gimel_{\mathfrak{A}}$ and $\gimel_{\mathfrak{B}}$ that expand $\mathfrak{A}$ and $\mathfrak{B}$ over $\mathcal{K}$  respectively. Then the duplicator has a winning strategy in the game $G_{\mathcal F}(\mathcal{E}_{\mathcal{K}}^{\gimel_{\mathfrak{A}}}(\mathfrak{A}),\mathcal{E}_{\mathcal{K}}^{\gimel_{\mathfrak{B}}}(\mathfrak{B}))$
 if the following is true: the duplicator has a winning strategy $D^{\mathcal F}_{\mathfrak{A},\mathfrak{B}}$ in the game $G_{\mathcal F}(\mathfrak{A},\mathfrak{B})$ such that
\begin{enumerate}%[(a)]
\item for any $a\in |\mathfrak{A}|$ and  sequence of elements $\bar s\in (|\mathfrak{A}|\uplus |\mathfrak{B}|)^{|\bar s|}$ ($|\bar s|\leq h(\mathcal F)$),  the duplicator has a winning strategy 
\[D^{\mathcal F}_{\gimel_{\mathfrak{A}}(a), \gimel_{\mathfrak{B}}(D^{\mathcal F}_{\mathfrak{A},\mathfrak{B}}(\bar s a))}\] in the game $G_{\mathcal F}(\gimel_{\mathfrak{A}}(a),\gimel_{\mathfrak{B}}(D^{\mathcal F}_{\mathfrak{A},\mathfrak{B}}(\bar s a)))$. \vspace{3pt}
\item for any $b\in |\mathfrak{B}|$ and  sequence of elements $\bar s\in (|\mathfrak{A}|\uplus |\mathfrak{B}|)^{|\bar s|}$ ($|\bar s|\leq h(\mathcal F)$),  the duplicator  has a winning strategy \[D^{\mathcal F}_{\gimel_{\mathfrak{A}}(D^{\mathcal F}_{\mathfrak{A},\mathfrak{B}}(\bar s b)),\gimel_{\mathfrak{B}}(b)}\] in the game $G_{\mathcal F}(\gimel_{\mathfrak{A}}(D^{\mathcal F}_{\mathfrak{A},\mathfrak{B}}(\bar s b)),\gimel_{\mathfrak{B}}(b))$.
\end{enumerate} 
\end{lem}
\begin{proof}
One winning strategy of the duplicator in the game $G_{\mathcal F}(\mathcal{E}_{\mathcal{K}}^{\gimel_{\mathfrak{A}}}(\mathfrak{A}),\mathcal{E}_{\mathcal{K}}^{\gimel_{\mathfrak{B}}}(\mathfrak{B}))$ is  the composition of her winning strategies in $G_{\mathcal F}(\mathfrak{A},\mathfrak{B})$ and $G_{\mathcal F}(\mathfrak{C}_i^A,\mathfrak{C}_j^B)$ where $\mathfrak{C}_i^A, \mathfrak{C}_j^B\in \mathcal{K}$. 

Let $\mathfrak{A}_{\gimel_{\mathfrak{A}}}$ be the isomorphic copy of $\mathfrak{A}$ in $\mathcal{E}_{\mathcal{K}}^{\gimel_{\mathfrak{A}}}(\mathfrak{A})$, whose elements interpret the hooks of disjoint substructures in $\mathcal{K}$. And let $\mathfrak{B}_{\gimel_{\mathfrak{B}}}$ be the isomorphic copy of $\mathfrak{B}$ in $\mathcal{E}_{\mathcal{K}}^{\gimel_{\mathfrak{B}}}(\mathfrak{B})$, whose elements interpret the hooks of disjoint substructures in $\mathcal{K}$. Assume that the spoiler has already picked a sequence $\bar s$ of elements in the game $G_{\mathcal F}(\mathfrak{A}_{\gimel_{\mathfrak{A}}},\mathfrak{B}_{\gimel_{\mathfrak{B}}})$:

\begin{enumerate}[label=(\roman*i)]
\item When the spoiler picks an element in $\mathfrak{A}_{\gimel_{\mathfrak{A}}}$ or $\mathfrak{B}_{\gimel_{\mathfrak{B}}}$, the duplicator uses the strategy $D^{\mathcal F}_{\mathfrak{A},\mathfrak{B}}$; otherwise:
\item For any $a\in |\mathfrak{A}_{\gimel_{\mathfrak{A}}}|$, if the spoiler picks an element in $\gimel_{\mathfrak{A}}(a)$, the duplicator uses the strategy $D^{\mathcal F}_{\gimel_{\mathfrak{A}}(a), \gimel_{\mathfrak{B}}(D^{\mathcal F}_{\mathfrak{A},\mathfrak{B}}(\bar s a))}$;
\item For any $b\in |\mathfrak{B}_{\gimel_{\mathfrak{B}}}|$, if the spoiler picks an element in $\gimel_{\mathfrak{B}}(b)$, the duplicator uses the strategy $D^{\mathcal F}_{\gimel_{\mathfrak{A}}(D^{\mathcal F}_{\mathfrak{A},\mathfrak{B}}(\bar s b)),\gimel_{\mathfrak{B}}(b)}$.
\end{enumerate}

\noindent Clearly, partial isomorphisms are preserved under such compositions, which provides the duplicator a winning strategy in the game $G_{\mathcal F}(\mathcal{E}_{\mathcal{K}}^{\gimel_{
\mathfrak{A}}}(\mathfrak{A}),\mathcal{E}_{\mathcal{K}}^{\gimel_{\mathfrak{B}}}(\mathfrak{B}))$, which can be regarded as a main game together with a series of subgames. 
\end{proof}

\subsection{Strictness of the FO quantifier structure hierarchy}

\begin{defi}
Let $\mathcal{I}$ be a structure over signature $\langle c_I\rangle$ where $c_I$ is the hook constant and  $\mathcal{I}$ has only one element, which is used to interpret $c_I$.
    \end{defi}

\begin{defi}\label{structures-over-TAU}
Let $m\in \mathbb{N}^+$ and $p\in \Gamma^*$. Let $\tau^+ =\langle U,R,B,r\rangle$ where $R,B$  are  binary relation symbols, $U$ is a unary relation symbol, and $r$ is a constant symbol. To make it vivid, we say that an element $x$ is \textit{black} if $x\in U$. All the elements in the structures are white unless explicitly labeled black. Likewise, an arc $(x,y)$ is red if $(x,y)\in R$;  an arc is blue if $(x,y)\in B$. Let $\tau^+_0=\tau^+\setminus\{U\}$. 

We define $\widetilde{\mathfrak{A}}^p_m$ and $\widetilde{\mathfrak{B}}^p_m$ to be $\tau^+$-structures as follows: \begin{itemize}
\item $\widetilde{\mathfrak{A}}^p_m$ and $\widetilde{\mathfrak{B}}^p_m$ are trees, whose edges are  either red or blue. The constant $r$ is interpreted by the root of the respective trees.
\item When $|p|=1$, all edges in $\widetilde{\mathfrak{A}}^p_m$ and $\widetilde{\mathfrak{B}}^p_m$ are  red.
\item If $p=\exists$, 

{(1)} $\widetilde{\mathfrak{A}}^p_m$ is a depth 1 tree that has $2m+1$ leaves. One of its leaves is black. 

{(2)} $\widetilde{\mathfrak{B}}^p_m$ is a depth 1 tree that has $2m$ leaves. \textit{None} of them is black.

$\mathcal{T}^{\mathfrak{A}}_{\exists,m}:=\widetilde{\mathfrak{A}}^\exists_m|\tau^+_0$; $\mathcal{T}^{\mathfrak{B}}_{\exists,m}:=\widetilde{\mathfrak{B}}^\exists_m|\tau^+_0$. 
\item If $p=\forall$, 

{(1)} $\widetilde{\mathfrak{A}}^p_m$ is a depth 1  tree that has $2m$ leaves. All of them are black. 

{(2)} $\widetilde{\mathfrak{B}}^p_m$ is a depth 1  tree that has
$2m+1$ leaves. All are black except one. 

$\mathcal{T}^{\mathfrak{A}}_{\forall,m}:=\widetilde{\mathfrak{A}}^\forall_m|\tau^+_0$; $\mathcal{T}^{\mathfrak{B}}_{\forall,m}:=\widetilde{\mathfrak{B}}^\forall_m|\tau^+_0$. 

Note that $\mathcal{T}^{\mathfrak{A}}_{\exists,m}=\mathcal{T}^{\mathfrak{B}}_{\forall,m}$ and $\mathcal{T}^{\mathfrak{B}}_{\exists,m}=\mathcal{T}^{\mathfrak{A}}_{\forall,m}$.

\item When $|p|>1:$\\[3pt]
\indent\hspace{6pt} Let $\widetilde{\mathfrak{A}}^{-q}_m$ be the same as $\widetilde{\mathfrak{A}}^{q}_m$ except that the colours of all the edges are exchanged, i.e.\ red is interchanged with blue. Let $\mathfrak{D}^{A,A}_{q,m}$, a $(\tau^+\cup\{\mathfrak{e}\}\setminus\{r\})$-structure, be built from a copy of $\widetilde{\mathfrak{A}}^q_m$ and a copy of $\widetilde{\mathfrak{A}}^{-q}_m$ where they are joined together at their roots. Call their shared root a \textit{junction point}, which interprets the hook constant $\mathfrak{e}$.  Similarly, define $\mathfrak{D}^{A,B}_{q,m}$ to be the join of $\widetilde{\mathfrak{A}}^q_m$ and $\widetilde{\mathfrak{B}}^{-q}_m$, $\mathfrak{D}^{B,A}_{q,m}$ be the join of $\widetilde{\mathfrak{B}}^q_m$ and $\widetilde{\mathfrak{A}}^{-q}_m$;  $\mathfrak{D}^{B,B}_{q,m}$ be the join of $\widetilde{\mathfrak{B}}^q_m$ and $\widetilde{\mathfrak{B}}^{-q}_m$. By  ``copies'' we mean disjoint copies. Let $\mathcal{K}:=\{\mathcal{I},\mathfrak{D}^{A,A}_{q,m},\mathfrak{D}^{A,B}_{q,m},\mathfrak{D}^{B,A}_{q,m}\}$.

\begin{itemize}
\item If $p=\exists q$, 

{({1}$^\exists$)} $\widetilde{\mathfrak{A}}^p_m$ is a point-expansion of $\mathcal{T}^{\mathfrak{A}}_{\exists,m}$ over $\mathcal{K}$ as follows: The root of $\mathcal{T}^{\mathfrak{A}}_{\exists,m}$ is expanded by $\mathcal{I}$. It is also called the ``\textit{root}'' of $\mathcal{T}^{\mathfrak{A}}_{\exists,m}$ that interprets $r$. One of its leaves is expanded by a copy of $\mathfrak{D}^{A,A}_{q,m}$. $m$ leaves are expanded by a copy of $\mathfrak{D}^{A,B}_{q,m}$. Another $m$ leaves are expanded by a copy of $\mathfrak{D}^{B,A}_{q,m}$.  

{({2}$^\exists$)} $\widetilde{\mathfrak{B}}^p_m$ is a point-expansion of $\mathcal{T}^{\mathfrak{B}}_{\exists,m}$ over $\mathcal{K}$ as follows: The root of $\mathcal{T}^{\mathfrak{B}}_{\exists,m}$ is expanded by $\mathcal{I}$. It is also called the ``\textit{root}'' of $\widetilde{\mathfrak{B}}^p_m$ that interprets $r$. $m$ leaves are expanded by a copy of $\mathfrak{D}^{A,B}_{q,m}$. The other $m$ leaves are expanded by a copy of $\mathfrak{D}^{B,A}_{q,m}$.
\item If $p=\forall q$,

{({1}$^\forall$)} The \textit{lem} of  $\widetilde{\mathfrak{A}}^p_m$ is the same as the above \textit{lem} ({2}$^\exists$) of  $\widetilde{\mathfrak{B}}^p_m$.

{({2}$^\forall$)} The \textit{lem} of  $\widetilde{\mathfrak{B}}^p_m$ is the same as the \textit{lem}  ({1}$^\exists$) of $\widetilde{\mathfrak{A}}^p_m$ except that $\mathfrak{D}^{A,A}_{q,m}$ is replaced by a copy of $\mathfrak{D}^{B,B}_{q,m}$.

\end{itemize}
\end{itemize}
    \end{defi}
\begin{rem} $\widetilde{\mathfrak{A}}^p_m$ is the same as $\widetilde{\mathfrak{B}}^{\bar p}_m$ except that the ``colours'' of leaves are flipped: black is interchanged with \textit{not} black.
             
\end{rem}

The structures $\widetilde{\mathfrak{A}}^p_m$ and $\widetilde{\mathfrak{B}}^p_m$ are trees with  coloured edges and nodes. Note that only a leaf could be black by this definition because the root of a tree is not black by default. See Figure \ref{EE-tree} for example, where $p=\exists\exists$ and $m=1$. Here, we use a solid line to represent a red edge and a dashed line to represent a blue edge.

\begin{figure}[h]
\centering
\includegraphics[trim = 15mm 35mm 15mm 40mm, clip, width=7cm]{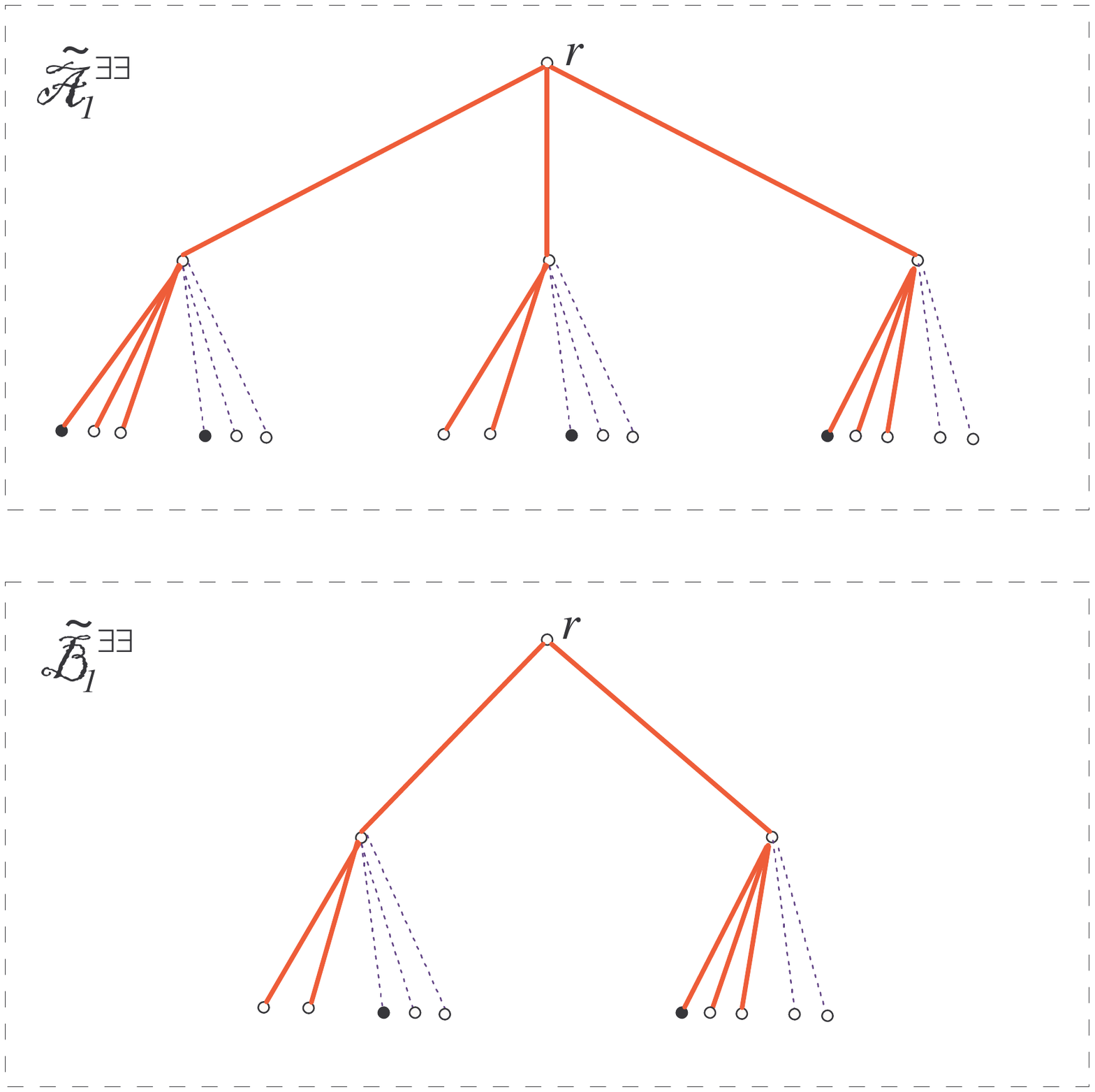}
%\scalebox{10}{}
\caption{The structures $\widetilde{\mathfrak{A}}^{\exists\exists}_1$ and $\widetilde{\mathfrak{B}}^{\exists\exists}_1$.}
\label{EE-tree}
\end{figure}

We are going to define a formula $\widetilde{\varphi}_p$ for each string $p\in \Gamma^*$.
\begin{defi} \label{property-over-TAU}
Let $q\in \Gamma^*$ and assume that $|q|=d\geq 0$.

\begin{enumerate}
\item 
$\widetilde{\psi}_{\epsilon}(x)=\widetilde{\psi}_{-\epsilon}(x):=U(x)$;
\vspace{6pt}
\item $\widetilde{\psi}_{\exists q}(y):=\exists x_{d+1}(Ryx_{d+1}\land \widetilde{\psi}_{q}(x_{d+1})\land \widetilde{\psi}_{-q}(x_{d+1}))$;\\[4pt]
$\widetilde{\psi}_{\forall q}(y):=\forall x_{d+1}( Ryx_{d+1}\rightarrow\widetilde{\psi}_{q}(x_{d+1})\lor \widetilde{\psi}_{-q}(x_{d+1}))$;
\vspace{6pt}
\item $\widetilde{\psi}_{-\exists q}(y):=\exists x_{d+1}(Byx_{d+1}\land \widetilde{\psi}_{q}(x_{d+1})\land \widetilde{\psi}_{-q}(x_{d+1}))$;\\[4pt]
$\widetilde{\psi}_{-\forall q}(y):=\forall x_{d+1}(Byx_{d+1}\rightarrow \widetilde{\psi}_{q}(x_{d+1})\lor \widetilde{\psi}_{-q}(x_{d+1}))$;
\end{enumerate}
\vspace{6pt}
Now we define a sentence over the signature $\tau^+$:
\begin{gather}\label{separate-sentence-0}
\widetilde{\varphi}_{p}:=\widetilde{\psi}_{p}(r); \\
\widetilde{\varphi}_{-p}:=\widetilde{\psi}_{-p}(r).
\end{gather}
    \end{defi}

\noindent From now on, we assume that $m\in \mathbb{N}^+$ is an arbitrary natural number except where defined explicitly in context (see Theorem \ref{main} for example). 

\begin{lem} \label{Dual-Hintikka}
$\widetilde{\mathfrak{A}}^p_m\models\widetilde{\varphi}_p\hspace{3pt}$ iff $\hspace{3pt}\widetilde{\mathfrak{A}}^{-p}_m\models\widetilde{\varphi}_{-p}$.
\end{lem}
\begin{proof}
Even though $\widetilde{\mathfrak{A}}^{-p}_m$ is defined as a structure which is obtained from $\widetilde{\mathfrak{A}}^p_m$ by exchanging all the colours of the edges, we will prove that $\widetilde{\mathfrak{A}}^p_m$ can be turned into $\widetilde{\mathfrak{A}}^{-p}_m$ if we only exchange the colours of the edges that connect to root (from red to blue and from blue to red) in $\widetilde{\mathfrak{A}}^p_m$, and vice versa: \begin{enumerate}
\item The base cases when $|p|=1$ are trivial;
\item Assume that it holds when $|p|=k$;
\item Assume that $p=\exists q$ where $|q|=k$. According to the definition,  the structure obtained from $\mathfrak{D}^{A,B}_{q,m}$ by exchanging the edge colours is  $\mathfrak{D}^{B,A}_{q,m}$. Likewise,  the structure obtained from $\mathfrak{D}^{B,A}_{q,m}$ by exchanging the edge colours is  $\mathfrak{D}^{A,B}_{q,m}$.  For the reason of symmetry, the structure obtained from $\mathfrak{D}^{A,A}_{q,m}$ by exchanging the edge colours is the same as $\mathfrak{D}^{A,A}_{q,m}$. Therefore, the conclusion holds according to the definition of $\widetilde{\mathfrak{A}}^{\exists q}_m$.
\item Similarly, we can prove it holds when $p=\forall q$ where $|q|=k$.
Therefore, it holds when $|p|=k+1$.
\end{enumerate}

\noindent Also note that $\widetilde{\varphi}_p$ is the same as $\widetilde{\varphi}_{-p}$ except that $\widetilde{\varphi}_p$ claims that the edges that connect to root are red and $\widetilde{\varphi}_{-p}$ claims that the edges that connect to root are blue. Therefore, $\widetilde{\mathfrak{A}}^p_m\models\widetilde{\varphi}_p\hspace{3pt}$ iff $\hspace{3pt}\widetilde{\mathfrak{A}}^{-p}_m\models\widetilde{\varphi}_{-p}$.
\end{proof}

\begin{lem} \label{claim1}
$\widetilde{\mathfrak{A}}^p_m \models \widetilde{\varphi}_p$ and $\widetilde{\mathfrak{B}}^p_m \not\models \widetilde{\varphi}_p$. 
\end{lem}
\begin{proof} It is obvious when $p=\exists$ or $\forall$.
Assume that it holds when $|p|=k$.  

If $p=\exists q$, $\widetilde{\mathfrak{A}}^p_m$ is composed of a node $a_r$, a copy of $\mathfrak{D}^{A,A}_{q,m}$, $m$ copies of $\mathfrak{D}^{A,B}_{q,m}$ and $m$ copies of $\mathfrak{D}^{B,A}_{q,m}$; $\widetilde{\varphi}_p=\exists x_{k+1}(Rrx_{k+1}\land\widetilde{\psi}_{q}(x_{k+1})\land \widetilde{\psi}_{-q}(x_{k+1}))$. Let $a_r$ interpret $r$ in $\widetilde{\mathfrak{A}}^p_m$, and the junction point $b_r$ of $\mathfrak{D}^{A,A}_{q,m}$ witness $\exists x_{k+1}$ in $\widetilde{\varphi}_p$. Note that $b_r$ divides $\mathfrak{D}^{A,A}_{q,m}$ into two parts: one is an isomorphic copy of $\widetilde{\mathfrak{A}}^q_m$ and the other is an isomorphic copy of $\widetilde{\mathfrak{A}}^{-q}_m$. For  convenience, we still use $\widetilde{\mathfrak{A}}^q_m$ and $\widetilde{\mathfrak{A}}^{-q}_m$ to denote these copies. Therefore, when $r$ is interpreted as $b_r$, $\hspace{3pt}\widetilde{\mathfrak{A}}^{q}_m\models\widetilde{\psi}_{q}(r)$ and  $\hspace{3pt}\widetilde{\mathfrak{A}}^{-q}_m\models\widetilde{\psi}_{-q}(r)$ (according to the induction hypothesis and Lemma \ref{Dual-Hintikka}).  

Moreover, all the quantifiers in $\widetilde{\psi}_q$ are relativized by relations either $Ryx$ or $Byx$, where $x$ is the quantified variable. And $\widetilde{\psi}_q$ expresses some property that has nothing to do with the elements outside the tree substructure $\widetilde{\mathfrak{A}}^q_m$. More precisely, in Definition \ref{property-over-TAU}, the variable ``$y$'' does not occur free in the formulas $\widetilde{\psi}_{q}(x_{d+1})$ and $\widetilde{\psi}_{-q}(x_{d+1})$.  
As a consequence, $\widetilde{\mathfrak{A}}^{q}_m\models\widetilde{\psi}_{q}(r)$ implies $\widetilde{\mathfrak{A}}^{p}_m\models\widetilde{\psi}_{q}(b_r)$.\\[4pt]
\indent By the same argument,  $\widetilde{\mathfrak{A}}^{-q}_m\models\widetilde{\psi}_{-q}(r)$ implies $\widetilde{\mathfrak{A}}^{p}_m\models\widetilde{\psi}_{-q}(b_r)$. Therefore, $\hspace{3pt}\widetilde{\mathfrak{A}}^{p}_m\models\widetilde{\varphi}_{p}$. \\[4pt]
\indent Let $c_r$ be a child of $r^{\widetilde{\mathfrak{B}}^{p}_m}$ in $\widetilde{\mathfrak{B}}^{p}_m$. By assumption,  $\widetilde{\mathfrak{B}}^q_m\not\models\widetilde{\psi}_{q}(r)$ ($\widetilde{\mathfrak{B}}^{-q}_m\not\models\widetilde{\psi}_{-q}(r)$ resp.) where $c_r$ interprets $r$ in $\widetilde{\mathfrak{B}}^q_m$ ($\widetilde{\mathfrak{B}}^{-q}_m$ resp.). As explained before, this means $\widetilde{\mathfrak{B}}^p_m\not\models\widetilde{\psi}_{q}(c_r)$ ($\widetilde{\mathfrak{B}}^p_m\not\models\widetilde{\psi}_{-q}(c_r)$ resp.). So, $c_r$ cannot be a witness of $\exists x_{k+1}$. Therefore,
$\hspace{3pt}\widetilde{\mathfrak{B}}^{p}_m\not\models\widetilde{\varphi}_{p}$.

Similarly, we can prove that it also holds when $p=\forall q$. Hence, it holds when $|p|=k+1$. 
\end{proof}
\vspace{4pt}

Let $S$ and $S^{\prime}$ be two finite $\Gamma$-labelled forests. We collect some simple facts below.

\begin{lem}\label{p2-p}
$\widetilde{\mathfrak{A}}^p_m\leadsto_{S} \widetilde{\mathfrak{B}}^p_m\hspace{3pt}$ iff $\hspace{3pt}\widetilde{\mathfrak{A}}^{-p}_m\leadsto_{S} \widetilde{\mathfrak{B}}^{-p}_m$.
\end{lem}
\begin{proof}
According to the definitions, $\widetilde{\mathfrak{A}}^{-p}_m$ ( $\widetilde{\mathfrak{B}}^{-p}_m$ resp.) is similar to $\widetilde{\mathfrak{A}}^p_m$ ($\widetilde{\mathfrak{B}}^p_m$ resp.) except that the colours of their edges are exchanged. Therefore, the duplicator can mimic her winning strategy in one game when she plays in the other. 
\end{proof}

\begin{defi}
Let $\mathfrak{A}\bigoplus_a\mathfrak{B}$, a structure over signature $\sigma$, be the join of two disjoint $\sigma$-substructures $\mathfrak{A}$ and $\mathfrak{B}$ at its element $a$. That is, \begin{enumerate}

\item $|\mathfrak{A}|\cap|\mathfrak{B}|=\{a\}$;
\item $|\mathfrak{\mathfrak{A}\bigoplus_a\mathfrak{B}}|=|\mathfrak{A}|\cup|\mathfrak{B}|$.
\item For any $R\in \sigma$, $R^{\mathfrak{A}\bigoplus_a\mathfrak{B}}:=R^{\mathfrak{A}}\cup R^{\mathfrak{B}}$.
\end{enumerate} 
    \end{defi}

\begin{lem} \label{pick-in-isom}
 Let $S$ be a $\Gamma$-labelled forest. $\mathfrak{A}\bigoplus_a\mathfrak{B}\leadsto_S \mathfrak{A}\bigoplus_b\mathfrak{B}^{\prime}$ if there is an automorphism $h$ of $\mathfrak{A}$ s.t.\ $b=h(a)$ and $\mathfrak{B}\leadsto_S\mathfrak{B}^{\prime}$.
\end{lem}
\begin{proof}
Assume that both $\mathfrak{B}$ and $\mathfrak{B}^{\prime}$ are $\sigma$-structures.  Let $\sigma^+:=\sigma\cup\{c_B\}$ where $c_B$ is called a hook. Let   $\check{\mathfrak{B}}$ ($\check{\mathfrak{B}}^{\prime}$ resp.) be an expansion of $\mathfrak{B}$ ($\mathfrak{B}^{\prime}$ resp.) to $\sigma^+$.  Let $\mathcal{K}=\{\mathcal{I},\check{\mathfrak{B}}\}$ and  $\mathcal{E}_{\mathcal{K}}^{\gimel}(\mathfrak{A})$ be a point-expansion of $\mathfrak{A}$ over $\mathcal{K}$ defined by $\gimel$ such that the element $a$ of $\mathfrak{A}$ is expanded by  $\check{\mathfrak{B}}$ and all the other elements are expanded by $\mathcal{I}$. Similarly, let $\mathcal{K}^{\prime}=\{\mathcal{I},\check{\mathfrak{B}}^{\prime}\}$ and  $\mathcal{E}_{\mathcal{K}^{\prime}}^{\gimel^{\prime}}(\mathfrak{A})$ be a point-expansion of $\mathfrak{A}$ over $\mathcal{K}^{\prime}$ defined by $\gimel^{\prime}$ such that the element $b$ of $\mathfrak{A}$ is expanded by  $\check{\mathfrak{B}}^{\prime}$ and all the other elements are expanded by $\mathcal{I}$. If there is an automorphism $h$ of $\mathfrak{A}$ s.t.\ $b=h(a)$, then by Lemma \ref{composition_of_strategies} the following holds: \[\mathcal{E}_{\mathcal{K}}^{\gimel}(\mathfrak{A})\leadsto_S \mathcal{E}_{\mathcal{K}^{\prime}}^{\gimel^{\prime}}(\mathfrak{A}) \hspace{4pt}\mathrm{if}\hspace{4pt} \mathfrak{B}\leadsto_S \mathfrak{B}^{\prime}.\]
Observe that $\mathcal{E}_{\mathcal{K}}^{\gimel}(\mathfrak{A})$ is exactly $\mathfrak{A}\bigoplus_a\mathfrak{B}$ and $\mathcal{E}_{\mathcal{K}^{\prime}}^{\gimel^{\prime}}(\mathfrak{A})$ is exactly $\mathfrak{A}\bigoplus_b\mathfrak{B}^{\prime}$. Hence, the lemma holds.
\end{proof}

\begin{figure}[h]
\centering
\includegraphics[trim = 45mm 100mm 45mm 100mm, clip, scale=0.52]{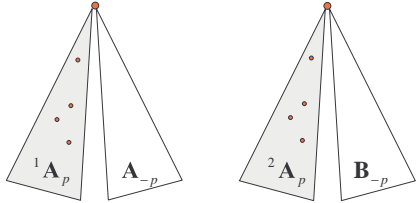}
%\scalebox{10}{}
\caption{$(\mathfrak{D}^{A,A}_{p,m},\bar a)$ and $(\mathfrak{D}^{A,B}_{p,m},\bar b)$ where $\bar a$ includes $r^{^1\widetilde{\mathfrak{A}}^p_m}$ and $\bar b$ includes $r^{^2\widetilde{\mathfrak{A}}^p_m}$. }
\label{Mini}
\end{figure}

We use $^1\widetilde{\mathfrak{A}}^p_m$ and $^2\widetilde{\mathfrak{A}}^p_m$ to denote two isomorphic copies of $\widetilde{\mathfrak{A}}^p_m$ in $\mathfrak{D}^{A,A}_{p,m}$ and $\mathfrak{D}^{A,B}_{p,m}$ respectively (see Figure \ref{Mini}. Note that  in the picture we use ``$^1$A$_{p}$'' to denote $^1\widetilde{\mathfrak{A}}^{p}_m$.). The superscripts ``1'' and ``2'' are used to distinguish these two copies. Let $\bar a\in |^1\widetilde{\mathfrak{A}}^p_m|^k$ and $\bar b\in |^2\widetilde{\mathfrak{A}}^p_m|^k$. The following lemma is a special case of Lemma \ref{pick-in-isom}.
\begin{lem} \label{constants-in-isom}
Assume that $(^1\widetilde{\mathfrak{A}}^p_m,\bar a)\cong(^2\widetilde{\mathfrak{A}}^p_m,\bar b)$. Then, \[(\mathfrak{D}^{A,A}_{p,m},\bar a)\leadsto_S(\mathfrak{D}^{A,B}_{p,m},\bar b) \hspace{4pt}\mathrm{if}\hspace{4pt} \widetilde{\mathfrak{A}}^{-p}_m\leadsto_S \widetilde{\mathfrak{B}}^{-p}_m\]
\end{lem}\vspace{2 pt}

\noindent Clearly, for two $\Gamma$-labelled forests $S_1$ and $S_2$, $\mathscr{W}(S_1)\subseteq \mathscr{W}(S_2)$ if $S_1\preceq_e S_2$.
Let $f^p_m:=\{q\in \gamma^-(f(p))|\hspace{3pt}|q|\leq m\}$.
\begin{lem} \label{EF-Game}
Let $p\in \Gamma^*$ and $m\in\mathbb{N}^+$, then the following holds:
\begin{enumerate}
\item The duplicator has a winning strategy in $G_{\mathscr{F}(f^p_m)}(\widetilde{\mathfrak{A}}^p_m,\widetilde{\mathfrak{B}}^p_m)$;
\item  For any $\psi$ such that $\mathscr{W}(qs(\psi))\subseteq f^p_m$, we have that 
\[
\widetilde{\mathfrak{A}}^p_m\models \psi\Rightarrow\widetilde{\mathfrak{B}}^p_m\models \psi.
\]
\end{enumerate}
\end{lem}
\begin{proof}
By Lemma \ref{words2forest}, $qs(\psi)\preceq_e \mathscr{F}(f^p_m)$. Hence by Theorem \ref{qs-game} and Lemma \ref{m2m-}, {(2)} is implied by {(1)} and we need only prove {(1)}.
\vspace{3 pt}

Note that $\widetilde{\varphi}_p$ is obtained from $\lnot \widetilde{\varphi}_{\bar p}$ by adding a ``$\lnot$'' before all the occurrences of the unary predicate $U$. And  $\widetilde{\varphi}_p$ is equivalent to a sentence in $\fo\{f(p)\}\hspace{1pt}$ iff $\lnot \widetilde{\varphi}_{\bar p}$ is equivalent to a sentence in $\fo\{f(p)\}\hspace{1pt}$ iff $\hspace{1pt}\widetilde{\varphi}_{\bar p}$ is equivalent to a sentence in $\fo\{f(\bar p)\}$ (the second ``iff'' is due to Lemma \ref{prefix}). That is,  the duplicator has a winning strategy in $G_{\mathscr{F}(f^p_m)}(\widetilde{\mathfrak{A}}^p_m,\widetilde{\mathfrak{B}}^p_m)$ iff she has a winning strategy in $G_{\mathscr{F}(f^{\bar p}_m)}(\widetilde{\mathfrak{A}}^{\bar p}_m,\widetilde{\mathfrak{B}}^{\bar p}_m)$.
 Therefore, we need only consider the case when $p[1]=\exists$ since the lemma holds when $p[1]=\exists$ iff it holds when $p[1]=\forall$.  Soon we shall see that the case when  $p[2]=\forall$ is different from the case when $p[2]=\exists$, provided that $p[1]=\exists$. Therefore, we discuss them separately.   

Let $|p|=d$. When $d=1$ i.e.\ $p=\exists$, $f(p)=\forall^*$. In the game, the spoiler can only pick at most $m$ distinct elements in $\widetilde{\mathfrak{B}}^p_m$. And $\widetilde{\mathfrak{A}}^p_m$ has $2m$ distinct elements that are not black. Hence, the duplicator is able to mimic the spoiler's picking as follows. She picks the root of $\widetilde{\mathfrak{A}}^p_m$ if the spoiler picks the root of $\widetilde{\mathfrak{B}}^p_m$; she picks a leaf that is not black in $\widetilde{\mathfrak{A}}^p_m$ if the spoiler picks a leaf of $\widetilde{\mathfrak{B}}^p_m$. That is, the duplicator has a winning strategy in this game. Similarly, when $p=\forall$, the duplicator also win the game. 

Assume that it holds when $d\leq k$ for some $k\geq 1$. That is, the duplicator has a winning strategy in the game $G_{\mathscr{F}(f^p_m)}(\widetilde{\mathfrak{A}}^p_m,\widetilde{\mathfrak{B}}^p_m)$ for any $|p|\leq k$ and any $m\in \mathbb{N}^+$.
 
Assume that $p=\exists\forall q$ where $|q|=k-1$. Then $f(p)=\forall^* *f(\forall q)$. This case is relatively easy to explain.

 The strategy of the duplicator in the game $G_{\mathscr{F}(f^{\forall q}_m)}(\mathcal{T}^{\mathfrak{A}}_{\exists,m},\mathcal{T}^{\mathfrak{B}}_{\exists,m})$ is very simple: \begin{enumerate}[label=(\Roman*)] 
\item If the spoiler picks the junction point of $\mathfrak{D}^{A,A}_{\forall q,m}$ that is a leaf of $\mathcal{T}^{\mathfrak{A}}_{\exists,m}$, the duplicator replies with a junction point of one copy of $\mathfrak{D}^{A,B}_{\forall q,m}$, called $\mathfrak{D}_A$, that is a leaf of  $\mathcal{T}^{\mathfrak{B}}_{\exists,m}$.
\item If the spoiler picks a junction point of $\mathfrak{D}^{A,B}_{\forall q,m}$ or $\mathfrak{D}^{B,A}_{\forall q,m}$ that is a leaf of one tree and that is not $\mathfrak{D}_A$, the duplicator replies with a junction point of an isomorphic copy that is a leaf of the other tree.
\item If the spoiler picks an element, say $a$, which has been picked before, the duplicator picks $b$, which was picked in the same round when $a$ was picked.

\end{enumerate}

\noindent We can regard the words that can be read off $\mathscr{F}(f^p_m)$ as the ``resource'' that the spoiler can use to detect the difference between the structures. Note that the first universal quantifier block in the words that can possibly be read off $\mathscr{F}(f^p_m)$ is useless for the spoiler: no matter how he picks in $\widetilde{\mathfrak{B}}^p_m$, the duplicator always picks an isomorphic substructure and mimics the spoiler's picks in the isomorphic substructure. By Lemma \ref{pick-in-isom}, if the spoiler can win in the end, he can also win if the players do not pick in these isomorphic substructures.

Let $Q:=\{s\in f^{\forall q}_j\mid 0\leq j\leq m\}$. Clearly, $Q\preceq f^{\forall q}_m$. 

Observe that the strategy described before is a winning strategy for the duplicator in the game $G_{\mathscr{F}(f^{\forall q}_m)}(\mathcal{T}^{\mathfrak{A}}_{\exists,m},\mathcal{T}^{\mathfrak{B}}_{\exists,m})$. By induction hypothesis, she also has a winning strategy in the game $G_{\mathscr{F}(f^{\forall q}_m)}(\widetilde{\mathfrak{A}}^{\forall q}_m,\widetilde{\mathfrak{B}}^{\forall q}_m)$. Hence, by Lemma \ref{constants-in-isom} and Lemma \ref{p2-p}, $\mathfrak{D}^{A,A}_{\forall q,m}\leadsto_{f^{\forall q}_m} \mathfrak{D}^{A,B}_{\forall q,m}$. Moreover, the winning strategies of the duplicator in the games $G_{\mathscr{F}(f^{\forall q}_m)}(\mathcal{T}^{\mathfrak{A}}_{\exists,m},\mathcal{T}^{\mathfrak{B}}_{\exists,m})$ and $G_{\mathscr{F}(f^{\forall q}_m)}(\mathfrak{D}^{A,A}_{\forall q,m},\mathfrak{D}^{A,B}_{\forall q,m})$ can be combined together: if the spoiler picks a leaf of $\mathcal{T}^{\mathfrak{B}}_{\exists,m}$ ($\mathcal{T}^{\mathfrak{A}}_{\exists,m}$ resp.) that is the junction point of $\mathfrak{D}^{A,B}_{\forall q,m}$ or $\mathfrak{D}^{B,A}_{\forall q,m}$ (not $\mathfrak{D}_A$), the duplicator picks a leaf of $\mathcal{T}^{\mathfrak{A}}_{\exists,m}$ ($\mathcal{T}^{\mathfrak{B}}_{\exists,m}$ resp.) that is the junction point of an isomorphic structure; if the spoiler picks a leaf of $\mathcal{T}^{\mathfrak{A}}_{\exists,m}$ that is the junction point of $\mathfrak{D}^{A,A}_{\forall q,m}$, the duplicator picks  a leaf of $\mathcal{T}^{\mathfrak{B}}_{\exists,m}$ that is the junction point of $\mathfrak{D}_A$. Therefore, she has a combined winning strategy in the game $G_{\mathscr{F}(f^{\forall q}_m)}(\widetilde{\mathfrak{A}}^{p}_m,\widetilde{\mathfrak{B}}^{p}_m)$, by  Lemma \ref{composition_of_strategies}. 
By  Lemma \ref{m2m-}, she also has a winning strategy in the game $G_{\mathscr{F}(Q)}(\widetilde{\mathfrak{A}}^{p}_m,\widetilde{\mathfrak{B}}^{p}_m)$. Note that $f^p_m=\{s\in\Gamma^m\mid s=\forall^i*s^{\prime}$ where $i\leq m$ and $s^{\prime}\in Q\}$. In other words, if we remove the path initiating from a root in $\mathscr{F}(f^p_m)$ where the word that can be read off the path is $\forall^m$, we can turn $\mathscr{F}(f^p_m)$ into $\mathscr{F}(Q)$.  Recall that the first universal quantifier block in the words that can be read off $\mathscr{F}(f^p_m)$ is useless for the spoiler. Therefore, the duplicator has a winning strategy in the game $G_{\mathscr{F}(f^p_m)}(\widetilde{\mathfrak{A}}^p_m,\widetilde{\mathfrak{B}}^p_m)$.

\begin{figure}[h]
\centering
\includegraphics[trim = 15mm 20mm 15mm 35mm, clip, scale=0.52]{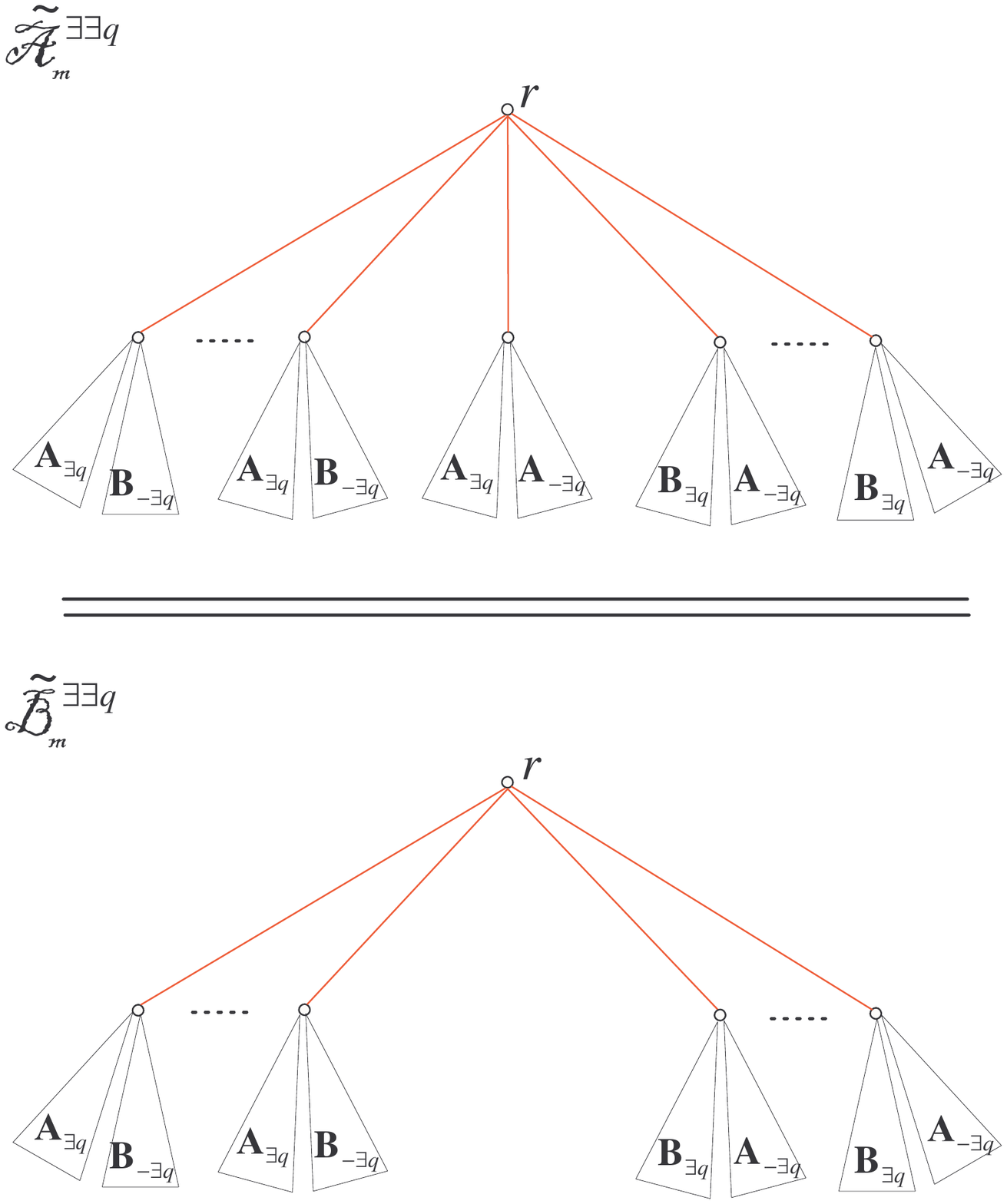}
%\scalebox{10}{}
\caption{The structures $\widetilde{\mathfrak{A}}^{\exists\exists q}_m$ and $\widetilde{\mathfrak{B}}^{\exists\exists q}_m$.}
\label{EEq}
\end{figure}

Assume that $p=\exists\exists q$ where $|q|=k-1$ (see Figure \ref{EEq}). Then $f(p)=\forall^*\exists*f(\exists q)$. As we have explained before, the first universal quantifier block in the words that can possibly be read off $\mathscr{F}(f^p_m)$ is useless for the spoiler. 

The strategy of the duplicator in the game $G_{\mathscr{F}(f^{\exists q}_m)}(\mathcal{T}^{\mathfrak{A}}_{\exists,m},\mathcal{T}^{\mathfrak{B}}_{\exists,m})$ is the same as her strategy in the game $G_{\mathscr{F}(f^{\forall q}_m)}(\mathcal{T}^{\mathfrak{A}}_{\exists,m},\mathcal{T}^{\mathfrak{B}}_{\exists,m})$.\\[6pt]
\indent For the first existential quantifier that can possibly be read off $\mathscr{F}(f^p_m)$, there are several choices for the spoiler:\begin{itemize} 
\item \textit{Picking the root} (or \textit{picking a junction point} resp.). 

The game is reduced to a composition of the main game $G_{\mathscr{F}(f^{\exists q}_m)}(\mathcal{T}^{\mathfrak{A}}_{\exists,m},\mathcal{T}^{\mathfrak{B}}_{\exists,m})$ and the subgames in which the duplicator has a winning strategy, that is, in the subgames 
\[
G_{\mathscr{F}(f^{\exists q}_m)}(\mathfrak{D}^{A,A}_{\exists q,m},\mathfrak{D}^{A,B}_{\exists q,m}) \hspace{5pt} \mathrm{or} \hspace{5pt} G_{\mathscr{F}(f^{\exists q}_m)}((\mathfrak{D}^{A,A}_{\exists q,m},r^{^1\widetilde{\mathfrak{A}}^p_m}),(\mathfrak{D}^{A,B}_{\exists q,m},r^{^2\widetilde{\mathfrak{A}}^p_m})),\] or subgames between isomorphic structures, according to Lemma \ref{p2-p} and Lemma \ref{constants-in-isom}. And by Lemma \ref{composition_of_strategies}, the duplicator has a winning strategy.

\item \textit{Picking inside the structure}, either $\mathfrak{D}^{A,A}_{\exists q,m}$ or $\mathfrak{D}^{A,B}_{\exists q,m}$ or $\mathfrak{D}^{B,A}_{\exists q,m}$, except their junction points. When the spoiler picks inside $\mathfrak{D}^{A,B}_{\exists q,m}$ or $\mathfrak{D}^{B,A}_{\exists q,m}$, the duplicator can mimic it in an isomorphic copy. Hence, only one new case need  be taken care of: the spoiler picks the element inside $\mathfrak{D}^{A,A}_{\exists q,m}$. In this case, the duplicator has a strategy as follows:

\begin{enumerate}[label=(\roman*)]
\item If the spoiler picks the element inside $\widetilde{\mathfrak{A}}^{\exists q}_m$, the duplicator mimics his picking in the isomorphic copy that is a part of $\mathfrak{D}^{A,B}_{\exists q,m}$. 
\item If the spoiler picks the element inside $\widetilde{\mathfrak{A}}^{-\exists q}_m$, the duplicator mimics his picking in the isomorphic copy that is a part of $\mathfrak{D}^{B,A}_{\exists q,m}$.
\end{enumerate} 

\noindent According to  Lemma \ref{pick-in-isom}, picking in these isomorphic substructures doesn't influence the outcome. In either case, the game is reduced to a composition of the main game $G_{\mathscr{F}(f^{\exists q}_m)}(\mathcal{T}^{\mathfrak{A}}_{\exists,m},\mathcal{T}^{\mathfrak{B}}_{\exists,m})$ and the subgame $G_{\mathscr{F}(f^{\exists q}_m)}(\widetilde{\mathfrak{A}}^{\exists q}_m,\widetilde{\mathfrak{B}}^{\exists q}_m)$ or $G_{\mathscr{F}(f^{\exists q}_m)}(\widetilde{\mathfrak{A}}^{-\exists q}_m,\widetilde{\mathfrak{B}}^{-\exists q}_m)$ or a subgame between two isomorphic structures. By the assumption we know that the duplicator has a winning strategy in $G_{\mathscr{F}(f^{\exists q}_m)}(\widetilde{\mathfrak{A}}^{\exists q}_m,\widetilde{\mathfrak{B}}^{\exists q}_m)$. Furthermore, by Lemma \ref{p2-p},  the duplicator has a winning strategy in $G_{\mathscr{F}(f^{\exists q}_m)}(\widetilde{\mathfrak{A}}^{-\exists q}_m,\widetilde{\mathfrak{B}}^{-\exists q}_m)$. Finally, by Lemma \ref{constants-in-isom}, Lemma \ref{m2m-} and Lemma \ref{composition_of_strategies}, the duplicator has a winning strategy in $G_{\mathscr{F}(f^p_m)}(\widetilde{\mathfrak{A}}^p_m,\widetilde{\mathfrak{B}}^p_m)$.
\end{itemize}
\vspace{3pt}

\noindent When $p[1]=\exists$, we have proved that $\widetilde{\mathfrak{A}}^p_m\leadsto_{f^p_m} \widetilde{\mathfrak{B}}^p_m$. Therefore, this theorem follows when $p[1]=\exists$. As a consequence it also holds when $p[1]=\forall$ by previous analysis.
\end{proof}

\begin{thm} \label{main}
Let $S_1$ and $S_2$ be two finite $\Gamma$-labelled forests. Over the class of all finite $\tau^+$-structures, 
\[
         \mathrm{if}\hspace{3pt} \mathscr{W}(S_1)\nsubseteq \mathscr{W}(S_2),\hspace{2pt} \mathrm{then}\hspace{3pt} \fo\{S_1\}\nsubseteq \fo\{S_2\}.
\]
\end{thm}
\begin{proof}
Let $p\in\mathscr{W}(S_1)\setminus\mathscr{W}(S_2)$.  According to  Lemma \ref{roson}, $\gamma^-(f(p))$ is the set of \textit{all} prefixes that $p$ is not a subsequence of.  Hence, $\mathscr{W}(S_2)\subseteq \gamma^-(f(p))$.  Clearly, $\widetilde{\varphi}_p\in \fo\{S_1\}$.  By Lemma \ref{claim1}, for any positive natural number $n$, $\widetilde{\mathfrak{A}}^p_n\models \widetilde{\varphi}_p$ but $\widetilde{\mathfrak{B}}^p_n\not\models \widetilde{\varphi}_p$. Assume for the purpose of a contradiction that there is a formula $\psi$ such that $qs(\psi)\preceq_e S_2$, $qr(\psi)=m$ and $\psi$ defines the same property as $\widetilde{\varphi}_p$.  Clearly, $\mathscr{W}(qs(\psi))\preceq f^p_m$. By Lemma \ref{EF-Game}, $\widetilde{\mathfrak{A}}^p_m\models \psi\Rightarrow\widetilde{\mathfrak{B}}^p_m \models \psi$. Together with  Lemma \ref{claim1}, and Corollary \ref{q-type-game-2}, the property defined by $\widetilde{\varphi}_p$ is not definable by any first-order sentence whose quantifier structure is $\mathcal{F}$ s.t.\ $\mathscr{W}(\mathcal{F})\subseteq f^p_m$, which is in contradiction with the assumption that $\psi$ defines the same property as $\widetilde{\varphi}_p$ does.
\end{proof}

If we want to prove something similar to Theorem \ref{main}, but over a restricted signature $\langle r, E\rangle$, then we need to adapt the constructions and formulas a bit. Note that we can use forward arrows and backward arrows to replace the red edges and blue edges. And we can use bi-directional edges to indicate where the black leaves are. More precisely, the new structure ${\mathfrak{A}_m^p}^\prime$ (${\mathfrak{B}_m^p}^\prime$ resp.) is obtained from $\widetilde{\mathfrak{A}}^p_m (\widetilde{\mathfrak{B}}^p_m$ resp.) by the following process:
\begin{enumerate}

\item Use an arc from $b$ to $c$ to represent that the edge between the vertices $b$ and $c$ is red;\vspace{1pt}
\item Use an arc from $c$  to $b$ to represent that the edge between $b$ and $c$ is blue;\vspace{1pt}
\item  Add an edge from every black leaf to the junction point in the same connected component and from the junction point to every black leaf in the same connected component;\vspace{1pt}
\item When $|p|=2$, add a self loop to every leaf which is an endpoint of a red edge.
\end{enumerate}

Correspondingly, the new formula, called $\varphi_p^\prime$, is  obtained from $\widetilde{\varphi}_p$ by the following process:
\begin{enumerate}[label=(\arabic*{*})] 
\item $Rxy$ is replaced by $Exy$;\vspace{1pt}
\item $Bxy$ is replaced by $Eyx$;\vspace{1pt}
\item When $|p|\neq 2$, $Uy$ is replaced by $Ex_1y\land Eyx_1$;\vspace{1pt}
\item When $|p|=2$, $Ux_1$ in $\widetilde{\psi}_\exists(y)$ and $\widetilde{\psi}_\forall(y)$ is replaced by $Ex_1y\land Eyx_1\land Ex_1x_1$; $Ux_1$ in $\widetilde{\psi}_{-\exists}(y)$ and $\widetilde{\psi}_{-\forall}(y)$ is replaced by $Ex_1y\land Eyx_1\land \lnot Ex_1x_1$.
\end{enumerate}

More precisely, we inductively define $\varphi_p^\prime$ as follows. 
Let $q\in \Gamma^*$ and assume $|q|=d\geq 0$.
\begin{enumerate}
\item $\psi_{\epsilon}^\prime(x,y)=\psi_{-\epsilon}^\prime(x,y):=Exy\land Eyx$; 

\item $\psi_{\exists q}^\prime(x,y):=\exists x_{d+1}(Eyx_{d+1}\land x_{d+1}\neq x\land \psi_{q}^\prime(y,x_{d+1})\land \psi_{-q}^\prime(y,x_{d+1}))$;\\
$\psi_{\forall q}^\prime(x,y):=\forall x_{d+1}(Eyx_{d+1}\land x_{d+1}\neq x\rightarrow \psi_{q}^\prime(y,x_{d+1})\lor \psi_{-q}^\prime(y,x_{d+1}))$;
\item $\psi_{-\exists q}^\prime(x,y):=\exists x_{d+1}(Ex_{d+1}y\land x_{d+1}\neq x\land \psi_{q}^\prime(y,x_{d+1})\land \psi_{-q}^\prime(y,x_{d+1}))$;\\
$\psi_{-\forall q}^\prime(x,y):=\forall x_{d+1}(Ex_{d+1}y\land x_{d+1}\neq x\rightarrow \psi_{q}^\prime(y,x_{d+1})\lor \psi_{-q}^\prime(y,x_{d+1}))$;

\item Now we define a sentence $\varphi_p^\prime$ over the signature $\tau$:
\begin{itemize}
\item When $|p|\neq 2$, 

$\indent\hspace{58pt}\varphi_{p}^\prime:=\psi_{p}^\prime(r,r)$.

(note that $\varphi_{-p}^\prime:=\psi_{-p}^\prime(r,r)$).
\item When $|p|=2$,

$\varphi_{\exists\exists}^\prime\!:=\!\exists x_2(Ex_2x_2\land \exists x_1(Ex_2x_1\!\land\! Ex_1x_2\!\land\! Ex_1x_1)
\land \exists x_1(Ex_1x_2\land Ex_2x_1\land \lnot Ex_1x_1))$  

$\varphi_{\exists\forall}^\prime\!:=\!\exists x_2(Ex_2x_2\land\! \forall x_1(Ex_2x_1\rightarrow Ex_1x_2\!\land\! Ex_1x_1)\\\indent\hspace{194pt} 
\land \forall x_1(Ex_1x_2\rightarrow Ex_2x_1\!\land\! \lnot Ex_1x_1))$  

$\varphi_{\forall\exists}^\prime\!:=\!\forall x_2(Ex_2x_2\!\rightarrow\! \exists x_1(Ex_2x_1\!\land\! Ex_1x_2\!\land\! Ex_1x_1) 
\!\lor\! \exists x_1(Ex_1x_2\!\land\! Ex_2x_1\!\land\! \lnot Ex_1x_1))$ 

 $\varphi_{\forall\forall}^\prime:=\forall x_2(Ex_2x_2\rightarrow \forall x_1(Ex_2x_1\rightarrow Ex_1x_2\land Ex_1x_1)\\\indent\hspace{194pt} 
\lor \forall x_1(Ex_1x_2\rightarrow Ex_2x_1\!\land\! \lnot Ex_1x_1))$.
\end{itemize}
\end{enumerate}

\noindent Let $\tau:=\langle E\rangle$ where $E$ is a binary relation symbol. Could we prove something similar to Theorem \ref{main}, but over $\tau$? To achieve it, in addition to the adaption we have introduced above, we need to find a way to get rid of the root, which is obvious. The $\tau$-structure $\mathfrak{A}_m^p$ is obtained from ${\mathfrak{A}_m^p}^\prime$ by the following process: removing  $r^{{\mathfrak{A}^p_m}^\prime}$ and for any junction point $v$, using a self loop at node $v$ to replace  the edge from $r^{{\mathfrak{A}^p_m}^\prime}$  to $v$.

$\mathfrak{B}_m^p$ is obtained from ${\mathfrak{B}_m^p}^\prime$ in exactly the same way. 
  
Correspondingly, $\varphi_p$ is obtained from $\varphi_p^\prime$ by substituting $Rrx$ with $Exx$ and removing the atoms $x\neq r$. 

\begin{rem}
It is a special case when $|p|=2$  because using bi-directional edges as a scheme of colouring does not work in this case. See Figures \ref{prefix-EE} for example, when $p=\exists\exists$ and $m=1$. 
\end{rem}

\begin{figure}[h]
\centering
\includegraphics[trim = 15mm 50mm 15mm 40mm, clip, width=7cm]{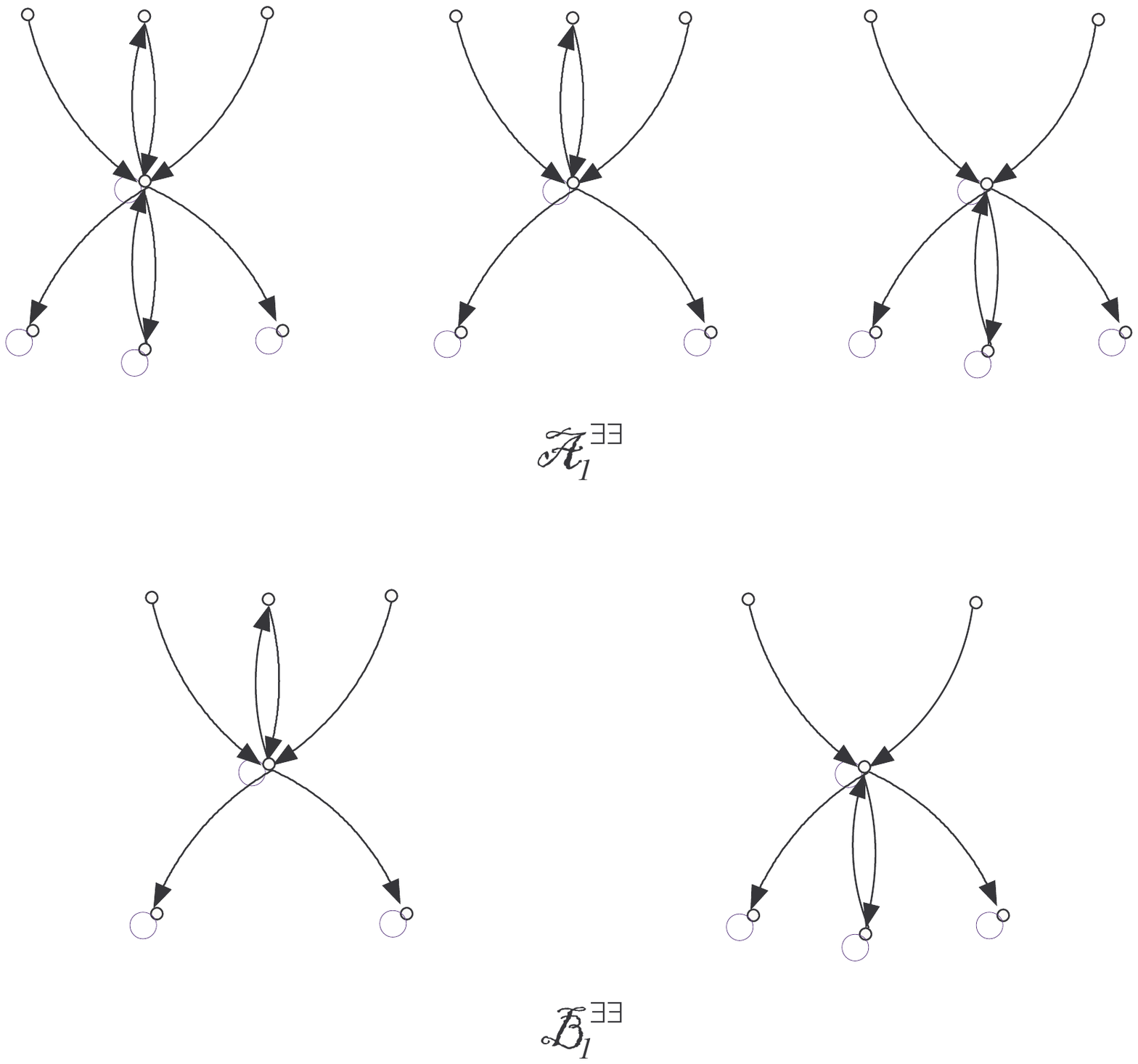}
%\scalebox{10}{}
\caption{The structures $\mathfrak{A}^{\exists\exists}_1$ and $\mathfrak{B}^{\exists\exists}_1$.}
\label{prefix-EE}
\end{figure}

\label{reductions-from-tauplus-to-tau}
\noindent Call the above reductions between structures and formulas ``\underline{reductions from $\tau^+$ to $\tau$}''. 

In a similar way to Lemma \ref{claim1}, we can prove the following lemma:
\begin{lem} \label{hintikka}
$\mathfrak{A}^p_m \models \varphi_p$ and $\mathfrak{B}^p_m \not\models \varphi_p$. 
\end{lem}

\begin{lem} \label{translation}
For any first-order sentence $\zeta$ over $\tau$, there is a first-order sentence $\xi$ over $\tau^+$, with the same quantifier structure, such that the following hold:

\vspace{7pt}
\begin{enumerate}
\item $\mathfrak{A}^p_m\models \zeta\hspace{3pt}$ iff  $\hspace{3pt}\widetilde{\mathfrak{A}}^p_m\models \xi$;\smallskip
\item $\mathfrak{B}^p_m\models \zeta\hspace{3pt}$ iff  $\hspace{3pt}\widetilde{\mathfrak{B}}^p_m\models \xi$.
\end{enumerate}   
\end{lem}
\begin{proof}
Let $\xi$ be obtained from $\zeta$ by \begin{enumerate}[label=\({\alph*}] 
\item relativising all quantifiers in $\zeta$ by $x\neq r$;
\item replacing all occurrences of $Exy$ by 
\[Rxy\lor Byx\lor (Rrx\land x=y)\lor (Rrx\land U(y))\lor (Rry\land U(x)).
\]
\end{enumerate}
Note that the quantifier structure of $\xi$ is the same as that of $\zeta$.

Because $\mathfrak{A}^p_m$ is obtained from $\widetilde{\mathfrak{A}}^p_m$ by: \begin{itemize}
\item deleting the root (corresponding to relativising quantifiers);
\item substituting red edges with forward arcs (corresponding to the disjunct $Rxy$ in the replacing of $Exy$);\vspace{1pt}
\item substituting blue edges with backward arcs (corresponding to the disjunct $Bxy$ in the replacing of $Exy$);\vspace{1pt}
\item adding self-loops at the junction points that are connected to the root(corresponding to the disjunct $Rrx\land x=y$ in the replacing of $Exy$);\vspace{1pt}
\item adding bi-directional edges between the black leaves and junction points that are in the same connected component (corresponding to the disjuncts $(Rrx\land U(y))$ and $(Rry\land U(x))$ in the replacing of $Exy$),
\end{itemize} 
it gives a reduction from the property defined by $\xi$ to the property defined by $\zeta$. In other words,  $\mathfrak{A}^p_m\models \zeta\hspace{3pt}$ iff  $\hspace{3pt}\widetilde{\mathfrak{A}}^p_m\models \xi$. For the same reason, {(2)} also holds.   
\end{proof}

Now, we prove the main result in this section, i.e.\ the theorem \ref{main1}.
\begin{proof}
Let $p\in\mathscr{W}(S_1)\setminus\mathscr{W}(S_2)$. Clearly, the property definable by $\varphi_p$ is in $\fo\{S_1\}$. We try to prove that this property is not in $\fo\{S_2\}$. Assume on the contrary that there is a formula $\psi$ such that $qs(\psi)\preceq_e S_2$ and $\psi$ defines the same property as $\varphi_p$ does. And let $qr(\psi)$ be $m$. According to  Lemma \ref{hintikka}, $\mathfrak{A}^p_m\models\psi$ and $\mathfrak{B}^p_m\not\models\psi$. According to  Lemma \ref{translation}, there is $\xi$, with the same quantifier structure as $\psi$, such that $\widetilde{\mathfrak{A}}^p_m\models\xi$ and $\widetilde{\mathfrak{B}}^p_m\not\models\xi$. Note that  $\mathscr{W}(qs(\xi))\subseteq \gamma^-(f(p))$ since $qs(\xi)\preceq_e S_2$ and $\mathscr{W}(S_2)\subseteq \gamma^-(f(p))$. But this is in contradiction to Lemma \ref{EF-Game}.
\end{proof}
Theorem \ref{main1} tells us that the distinctive collections of quantifier classes form a strict hierarchy, which we call  quantifier structure hierarchy.

\section{Strictness of quantifier hierarchy over ordered finite structures} \label{sec-order-hierarchy}

Up to now, using logics to characterize complexity classes inside NP requires the structures to be ordered, i.e.\ there is a linear order over the universe of the structures. Therefore, it is interesting to extend the main result in the last section to ordered structures: the first-order quantifier structure hierarchy is strict over ordered finite structures. However, separating the expressive power of logics over ordered structures is often difficult, because the spoiler may detect the difference between the structures using a given linear order. But we will see in this section that  the structures will be constructed in such a way that the power of linear order that the spoiler can use is quite limited: it is equivalent to the power that the spoiler can use in a game over a pair of linear orders, and a well-known result tells us that the duplicator has a winning strategy over a game between two linear orders that are sufficiently long. 
In this section we sketch the main ideas that conquer the order problem  and omit most details that resemble those in the formal proof of Theorem \ref{main1}. 

\subsection{The constructions and separating property}

\begin{defi}\label{ordered-structures}
Recall that $\tau^{+\text{ORD}}$ is $\langle R,B,U,r,\leq\rangle$ where $\leq$ is interpreted as a linear order over the universe. Let $\tau_o^{\prime}:=\tau^{+\text{ORD}}\setminus \{U\}$.

Let the structure $\overrightarrow{\mathfrak{A}}^p_m$ be a $\tau^{+\text{ORD}}$-structure defined as follows: \begin{itemize}
\item $\overrightarrow{\mathfrak{A}}^p_m | \tau^+$ and $\overrightarrow{\mathfrak{B}}^p_m | \tau^+$ are trees, whose  roots interpret $r$, and whose edges are coloured either red or blue. 
\item If $p=\exists$, \begin{enumerate}
\item $\overrightarrow{\mathfrak{A}}^p_m | \tau^+$ is a depth 1  tree that has $2^{m+2}+1$ leaves. To construct $\overrightarrow{\mathfrak{A}}^p_m$ from $\overrightarrow{\mathfrak{A}}^p_m | \tau^+$, give these leaves some order such that the $(2^{m+1}+1)$-th leaf is black. All the other leaves are not black.
\item $\overrightarrow{\mathfrak{B}}^p_m | \tau^+$ is a depth 1  tree that has $2^{m+1}$ leaves. \textit{None} of them is black. Give these leaves an arbitrary order to construct $\overrightarrow{\mathfrak{B}}^p_m$ from $\overrightarrow{\mathfrak{B}}^p_m | \tau^+$.

Let $\overrightarrow{\mathcal{T}}^{\mathfrak{A}}_{\exists,m}:=\overrightarrow{\mathfrak{A}}^\exists_m|\tau_o^{\prime}; \overrightarrow{\mathcal{T}}^{\mathfrak{B}}_{\exists,m}:=\overrightarrow{\mathfrak{B}}^\exists_m|\tau_o^{\prime}.$
\end{enumerate}
\item If $p=\forall$, \begin{enumerate}
\item $\overrightarrow{\mathfrak{A}}^p_m | \tau^+$ is a depth 1  tree that has $2^{m+1}$ leaves. All of them are black. 
\item $\overrightarrow{\mathfrak{B}}^p_m | \tau^+$ is a depth 1  tree that has $(2^{m+2}+1)$ leaves. All are black except one. Give these leaves some order such that the $(2^{m+1}+1)$-th leaf is not black. 
\end{enumerate}
In the above definition, we let a node be earlier in the linear order $\leq$ than its children. Moreover, we define all the colours of the edges to be red when $|p|=1$.

\item When $|p|>1$:

\hspace{6pt} Let $\overrightarrow{\mathfrak{A}}^{-q}_m$ be the same as $\overrightarrow{\mathfrak{A}}^{q}_m$ except that the colours of all the edges are exchanged, i.e.\ red is changed to blue and vice versa. Let $\overrightarrow{\mathfrak{D}}^{A,B}_{q,m}$, a $(\tau^{+\text{ORD}}\cup\{\mathfrak{e}\}\setminus\{r\})$-structure, be built from a copy of $\overrightarrow{\mathfrak{A}}^q_m$ and a copy of $\overrightarrow{\mathfrak{B}}^{-q}_m$ where they are joined together at their roots. Their shared root is called junction point, which is used to  interpretes the constant $\mathfrak{e}$. Similarly define $\overrightarrow{\mathfrak{D}}^{A,A}_{q,m}$ as the join of $\overrightarrow{\mathfrak{A}}^q_m$ and $\overrightarrow{\mathfrak{A}}^{-q}_m$, define  $\overrightarrow{\mathfrak{D}}^{B,A}_{q,m}$ as the join of $\overrightarrow{\mathfrak{B}}^q_m$ and $\overrightarrow{\mathfrak{A}}^{-q}_m$, 
and define $\overrightarrow{\mathfrak{D}}^{B,B}_{q,m}$ as the join of $\overrightarrow{\mathfrak{B}}^q_m$ and $\overrightarrow{\mathfrak{B}}^{-q}_m$. In $\overrightarrow{\mathfrak{D}}^{A,B}_{q,m}$, we let the elements in $\overrightarrow{\mathfrak{A}}^q_m$ be later in the linear order than the elements in  $\overrightarrow{\mathfrak{B}}^{-q}_m$. Similarly, in $\overrightarrow{\mathfrak{D}}^{A,A}_{q,m}$, we let the elements in $\overrightarrow{\mathfrak{A}}^q_m$ be later in the order than the elements in  $\overrightarrow{\mathfrak{A}}^{-q}_m$; in $\overrightarrow{\mathfrak{D}}^{B,A}_{q,m}$,  the elements in $\overrightarrow{\mathfrak{B}}^q_m$ are later in the order than the elements in  $\overrightarrow{\mathfrak{A}}^{-q}_m$; in $\overrightarrow{\mathfrak{D}}^{B,B}_{q,m}$,  the elements in $\overrightarrow{\mathfrak{B}}^q_m$ are later in the order than the elements in  $\overrightarrow{\mathfrak{B}}^{-q}_m$. 

In the following, we assume that any node in $\overrightarrow{\mathfrak{D}}^{X_1,Y_1}_{q,m}$ is later in the order than all nodes in  $\overrightarrow{\mathfrak{D}}^{X_2,Y_2}_{q,m}$ if the junction point of $\overrightarrow{\mathfrak{D}}^{X_1,Y_1}_{q,m}$ is later in the order than the junction point of $\overrightarrow{\mathfrak{D}}^{X_2,Y_2}_{q,m}$.
\begin{itemize}
\item If $p=\exists q$, \begin{enumerate}[label=(\arabic*$^\exists$)]
\item  $\overrightarrow{\mathfrak{A}}^p_m$ is a point-expansion of $\overrightarrow{\mathcal{T}}^{\mathfrak{A}}_{\exists,m}$ over $\{\mathcal{I},\overrightarrow{\mathfrak{D}}^{A,A}_{q,m},\overrightarrow{\mathfrak{D}}^{A,B}_{q,m},\overrightarrow{\mathfrak{D}}^{B,A}_{q,m}\}$. Its root is expanded by $\mathcal{I}$. Recall that $\mathcal{I}$ is a structure over the signature $\langle c_I \rangle$, whose universe has exactly one element, and this element is used to interpret the constant $c_I$. When $i\leq 2^{m+1}$, the $i$-th leaf is expanded by a copy of $\overrightarrow{\mathfrak{D}}^{A,B}_{q,m}$ if $i$ is odd, by a copy of $\overrightarrow{\mathfrak{D}}^{B,A}_{q,m}$ if $i$ is  even; the $i$-th leaf is expanded by a copy of $\overrightarrow{\mathfrak{D}}^{A,A}_{q,m}$ if $i=2^{m+1}+1$; when  $i>2^{m+1}+1$, the $i$-th leaf is expanded by a copy of $\overrightarrow{\mathfrak{D}}^{A,B}_{q,m}$ if $i-(2^{m+1}+1)$ is odd, by a copy of $\overrightarrow{\mathfrak{D}}^{B,A}_{q,m}$ if $i-(2^{m+1}+1)$ is even. 
\item $\overrightarrow{\mathfrak{B}}^p_m$ is a point-expansion of $\overrightarrow{\mathcal{T}}^{\mathfrak{B}}_{\exists,m}$ over $\{\mathcal{I},\overrightarrow{\mathfrak{D}}^{A,B}_{q,m},\overrightarrow{\mathfrak{D}}^{B,A}_{q,m}\}$.  Its root is expanded by $\mathcal{I}$. The $i$-th leaf is expanded by a copy of $\overrightarrow{\mathfrak{D}}^{A,B}_{q,m}$ if $i$ is odd, by a copy of $\overrightarrow{\mathfrak{D}}^{B,A}_{q,m}$ if $i$ is  even.
\end{enumerate}

\item If $p=\forall q$, \begin{enumerate}[label=\arabic*1$^\forall$)]
\item  The \textit{lem} of  $\overrightarrow{\mathfrak{A}}^p_m$ is the same as the above \textit{lem} (2$^\exists$) of  $\overrightarrow{\mathfrak{B}}^{\exists q}_m$.
\item  The \textit{lem} of $\overrightarrow{\mathfrak{B}}^p_m$ is the same as the \textit{lem} (1$^\exists$) of $\overrightarrow{\mathfrak{A}}^{\exists q}_m$ except that $\overrightarrow{\mathfrak{D}}^{A,A}_{q,m}$ is replaced by a copy of $\overrightarrow{\mathfrak{D}}^{B,B}_{q,m}$.
\end{enumerate}
\end{itemize}

\end{itemize}

                       \end{defi}

\noindent Note that  every structure under consideration is just an ordered coloured tree and there is a path from its root to any node. We can read off the string of colours of edges along this path without a skip. Define $lab(x)$ to be a string in $\{\mathrm{red},\mathrm{blue}\}^*$ associated with the path. $lab(r):=\epsilon$, where $\epsilon$ is the empty string. We let ``red'' be later in the order than ``blue'' in the lexicographic order. Let $x,y$ be two nodes. We use $x^{\mathfrak{f}}$ ($y^{\mathfrak{f}}$ resp.) to denote the father of $x$ ($y$ resp.) in this section.

In the above, we have defined the structures that will be used in the games. And the sentence that we use to define the separating property is the same as Definition \ref{property-over-TAU}. That is, we actually use the same separating property to achieve the goal.

In the last section, we use point-expansions to realize strategy compositions. In some special cases, such compositions can be simplified: some of the substructures  collapse to ``points'', i.e.\ the details of the substructures are omitted, and we use ``colours'' to distinguish different substructures, which are now \textit{regarded} temporarily as elements. We call such a method a kind of ``structural abstraction'', which is used to omit unrelated details of structures and simplify game arguments, and the games played on the simplified structures are images of the original games.

\subsection{The duplicator's winning strategy}
Note that, even in the unordered case, when the spoiler picks $x$, the duplicator's strategy in the games is always picking $y$ s.t.\ $lab(x)=lab(y)$ and picking a child of an element which is picked in some previous round, say the $i$-th round, when the spoiler picks a child of the other element which is picked in the $i$-th round. This lays the crucial basis for the previous inductive proof of Theorem \ref{main} to extend to classes of linearly ordered finite structures, because now we can use something similar to the following well-known result \cite{Libkin04Elements}: 

Let $k\geq 1$, and let $L_1,L_2$ be linear orders of length no less than $2^k$, then \[L_1\equiv_k L_2. \hspace{1in} (*)\]   

Here, $L_1\equiv_k L_2$ means that, for any $\varphi\in \fo$ with $qr(\varphi)\leq k$, $L_1\models \varphi$ iff $L_2\models \varphi$. 

 Let $\widetilde{\varphi}_p$ be given by the definition \ref{property-over-TAU}. The proof of  lemma \ref{claim1} also shows that $\overrightarrow{\mathfrak{A}}^p_m \models \widetilde{\varphi}_p$ and $\overrightarrow{\mathfrak{B}}^p_m \not\models \widetilde{\varphi}_p$, since ``$\leq$'' does not appear in $\widetilde{\varphi}_p$.

To prove a version of  Theorem \ref{main} over ordered structures, the main idea is almost the same. Here, we just need to take care of the linear order.  As we have explained before, the players will always pick a pair
of elements that have the same label, and if the spoiler picks more than one
child of a node in one structure, so does the duplicator in the other structure. Hence, we can use structural abstraction to conceal the details of the substructures like $\overrightarrow{\mathfrak{D}}^{A,B}_{q,m}$, and regard the problem to be a game over two linear orders with three ``colours", which represent three ``colours''  $\overrightarrow{\mathfrak{D}}^{A,B}_{q,m}$, $\overrightarrow{\mathfrak{D}}^{B,A}_{q,m}$, and $\overrightarrow{\mathfrak{D}}^{A,A}_{q,m}$ respectively. See Figure \ref{prefix-Eq-color} for example. Here, $\overrightarrow{\mathfrak{D}}^{A,B}_{q,m}$ is identified with light red; $\overrightarrow{\mathfrak{D}}^{B,A}_{q,m}$ is identified with light blue; $\overrightarrow{\mathfrak{D}}^{A,A}_{q,m}$ is identified with green. 

\begin{figure}[h]
\centering
\includegraphics[trim = 15mm 35mm 15mm 40mm, clip, width=10cm]{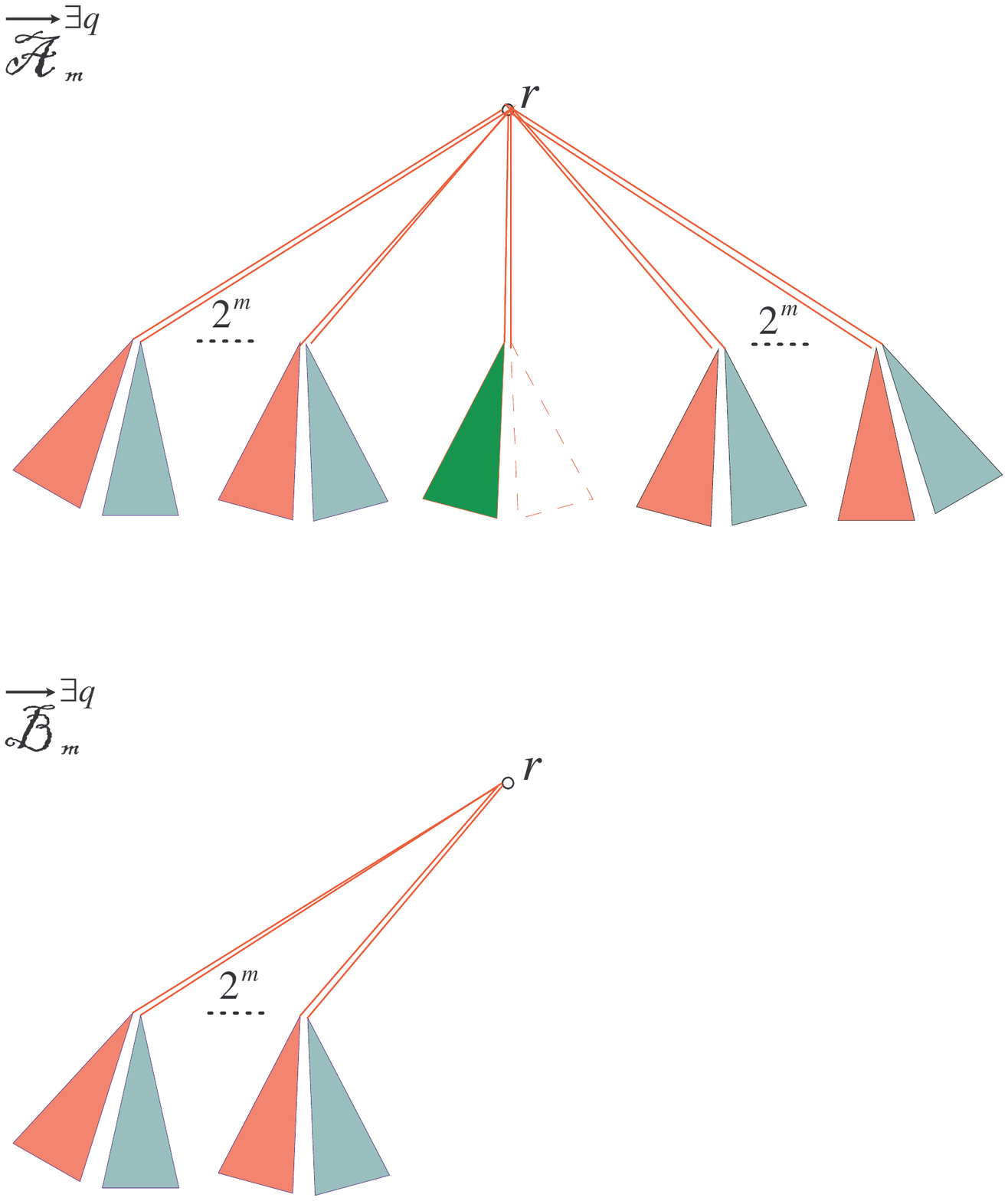}
\caption{The structures where subgraphs are identified with ``colours''.}
\label{prefix-Eq-color}
\end{figure}

As explained before, we may safely assume that $p[1]=\exists$.

\begin{figure}[h]
\centering
\includegraphics[trim = 15mm 35mm 15mm 40mm, clip, width=10cm]{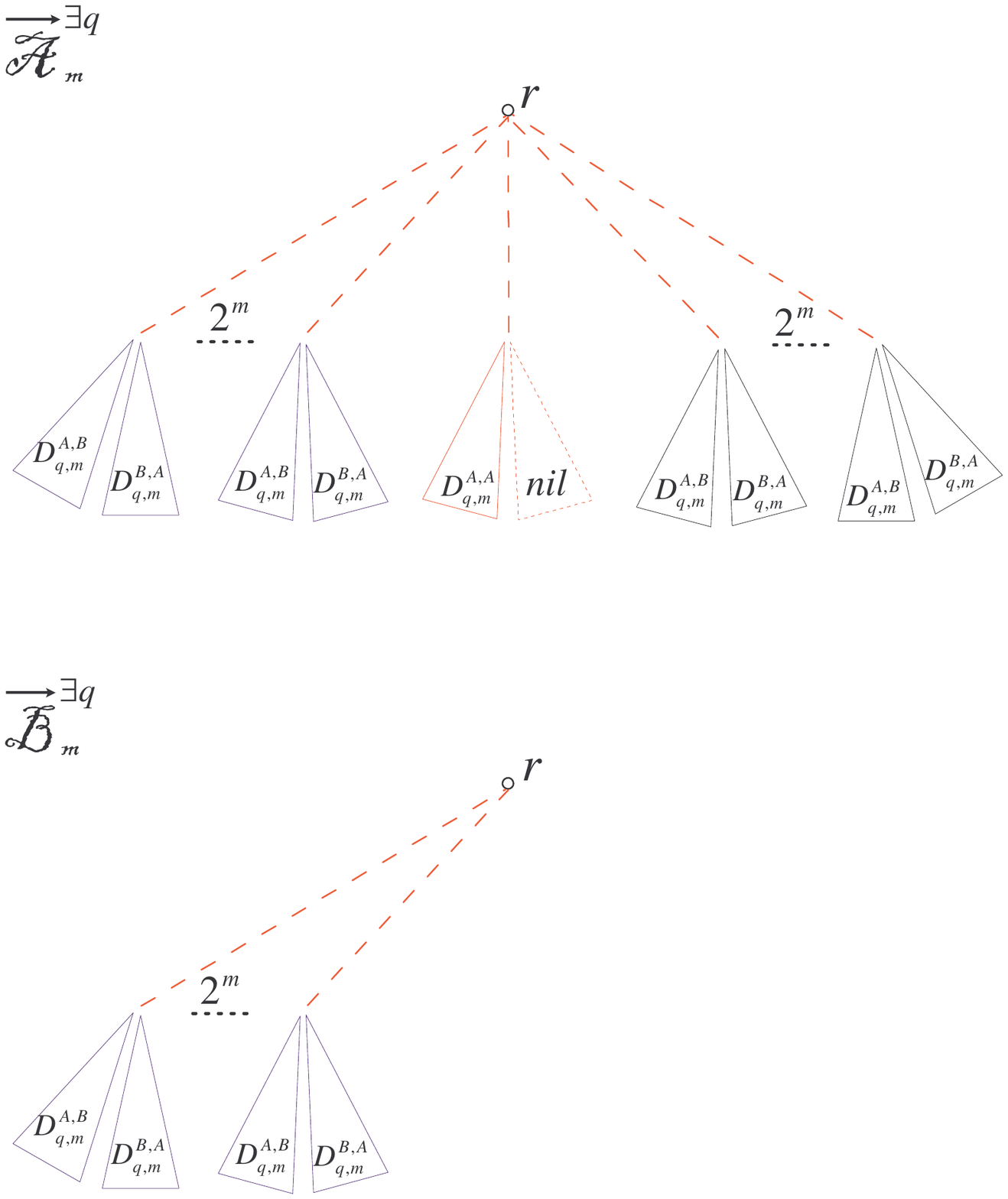}
\caption{The structures and 2-tuples.}
\label{prefix-Eq-order}
\end{figure}
Assume that $p=\exists q$. 
We can pair the children of $r^{\overrightarrow{\mathfrak{A}}^p_m}$ as:
 \begin{equation}\label{2-tuples-times}
(\overrightarrow{\mathfrak{D}}^{A,B}_{q,m},\overrightarrow{\mathfrak{D}}^{B,A}_{q,m}),\cdots,(\overrightarrow{\mathfrak{D}}^{A,A}_{q,m},\times),\cdots,(\overrightarrow{\mathfrak{D}}^{A,B}_{q,m},\overrightarrow{\mathfrak{D}}^{B,A}_{q,m}).
\end{equation}

We can regard two $(\tau^{+\text{ORD}}\cup\{\mathfrak{e}\}\setminus\{r\})$-structures in the brackets as a whole. More precisely, an  \underline{$s$-2-tuple} is a pair of structures from $\{\overrightarrow{\mathfrak{D}}^{A,B}_{s,m}$,$\overrightarrow{\mathfrak{D}}^{B,A}_{s,m}$,$\overrightarrow{\mathfrak{D}}^{A,A}_{s,m}$,$\overrightarrow{\mathfrak{D}}^{B,B}_{s,m}\}$ ($s\preceq p$), which is regarded as a single ``super-element''. From now on, we omit ``$s$'' in  ``$s$-2-tuple'' when it will not cause confusion from the context. In (\ref{2-tuples-times}), ``$\times$'' in the 2-tuple  $(\mathfrak{D}^{A,A}_{q,m},\times)$ represents an empty, or imaginary, structure, which is just used to make up a 2-tuple. See Figure \ref{prefix-Eq-order}. But, since we are regarding a 2-tuple as one single object now, we use one dashed line to represent two arrows.  Note that, by definition, the elements of $\overrightarrow{\mathfrak{D}}^{A,B}_{s,m}$ are earlier than those of $\overrightarrow{\mathfrak{D}}^{B,A}_{s,m}$, in a 2-tuple $(\overrightarrow{\mathfrak{D}}^{A,B}_{s,m},\overrightarrow{\mathfrak{D}}^{B,A}_{s,m})$.   
These 2-tuples have a natural linear order $\leq_2$ inherited  from $\leq^{\overrightarrow{\mathfrak{A}}^p_m}$: for any two 2-tuples $(\mathcal{X}_1,\mathcal{Y}_1)$ and $(\mathcal{X}_2,\mathcal{Y}_2)$, $(\mathcal{X}_1,\mathcal{Y}_1)\leq_2 (\mathcal{X}_2,\mathcal{Y}_2)$ iff $r^{\mathcal{Y}_1}\leq^{\overrightarrow{\mathfrak{A}}^p_m} r^{\mathcal{X}_2}$ where $r^{\mathcal{Y}_1}$ and $r^{\mathcal{X}_2}$ are the junction points of $\mathcal{Y}_1$ and $\mathcal{X}_2$ respectively. Similarly, we can pair the children of $r^{\overrightarrow{\mathfrak{B}}^p_m}$. Note that there are $(2^{m+1}+1)$  2-tuples in $\overrightarrow{\mathfrak{A}}^p_m$, which form a linear order $L_A$, and $2^{m}$ 2-tuples in $\overrightarrow{\mathfrak{B}}^p_m$, which form another linear order $L_B$. \label{LA-and-LB} See Figure \ref{prefix-Eq-color-2}: a yellow node represents a 2-tuple, which corresponds to a pair of light red and light blue structures in Figure \ref{prefix-Eq-color}; the green node represents the 2-tuple that is green in Figure  \ref{prefix-Eq-color}.

\begin{figure}[h]
\centering
\includegraphics[trim = 0mm 65mm 0mm 70mm, clip, width=10cm]{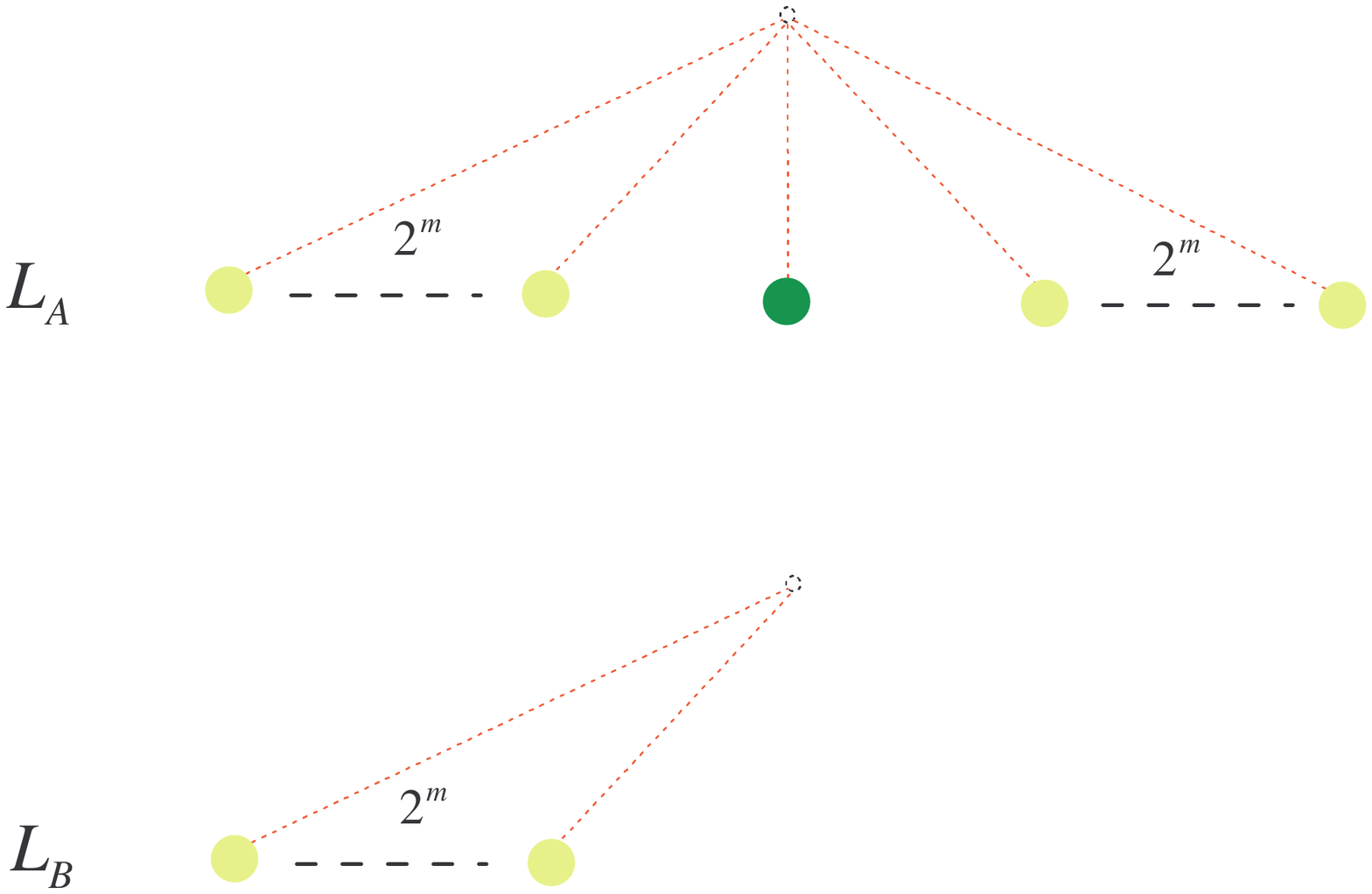}
\caption{$L_A$ and $L_B$.}
\label{prefix-Eq-color-2}
\end{figure}

Let $S$ be a $\Gamma$-labelled forest and $rk(S)=m$.   In the game $G_S(\overrightarrow{\mathfrak{A}}^p_m,\overrightarrow{\mathfrak{B}}^p_m)$, when the spoiler picks an element $x$, which is not the root of the structure, in some round, $x$ determines a $q$-2-tuple, where $p=\exists q$, s.t.\ $x$ is in the universe of this $q$-2-tuple. And we may also say that this  2-tuple is picked in this round. Hence,  $G_S(\overrightarrow{\mathfrak{A}}^p_m,\overrightarrow{\mathfrak{B}}^p_m)$ induces a new game in which the players pick 2-tuples instead of elements in the universe. Moreover, these 2-tuples form a pair of linear orders $L_A$ and $L_B$, as shown in Figure \ref{prefix-Eq-color-2}. Call this new game ``coloured linear order game''  $G_S(L_A,L_B)$. 

From now on, we use a natural number to denote a 2-tuple, in order to omit the details of 2-tuples  that are not related to our concern but at the same time retain the order relation between 2-tuples. Therefore, we can subtract one 2-tuple from another 2-tuple in this context. Note that in this viewpoint linear orders can also be regarded as intervals. Moreover, every set of elements that have the same labels form an interval.

The duplicator's strategy in $G_S(L_A,L_B)$ is as follows. Assume that $\mathcal{B}_i,\mathcal{B}_j$ are already picked. Let $\mathcal{A}_i$ ($\mathcal{A}_j$ resp.) be the element picked in $L_A$ in the same round as $\mathcal{B}_i$ ($\mathcal{B}_j$ resp.) was picked.  Recall that these 2-tuples can be compared by the induced order. 
\begin{itemize}
\item If in the current round the spoiler picks a 2-tuple $\mathcal{B}$ in the interval $[\mathcal{B}_i,\mathcal{B}_j]$ and  $\mathcal{B}-\mathcal{B}_i\leq \mathcal{B}_j-\mathcal{B}$ then the duplicator picks $\mathcal{A}$ s.t.\ $\mathcal{A}-\mathcal{A}_i=\mathcal{B}-\mathcal{B}_i$; otherwise she picks $\mathcal{A}$ s.t.\ $\mathcal{A}_j-\mathcal{A}=\mathcal{B}_j-\mathcal{B}$.
\item If in the current round the spoiler picks a 2-tuple $\mathcal{A}$ in the interval  $[\mathcal{A}_i,\mathcal{A}_j]$, the duplicator's strategy is as follows:
\begin{enumerate}
\item if $\mathcal{A}-\mathcal{A}_i\leq \frac{1}{2}(\mathcal{B}_j-\mathcal{B}_i)$ then the duplicator picks $\mathcal{B}$ in $[\mathcal{B}_i,\mathcal{B}_j]$ s.t.\ $\mathcal{B}-\mathcal{B}_i=\mathcal{A}-\mathcal{A}_i$; otherwise 
\item if $\mathcal{A}_j-\mathcal{A}\leq  \frac{1}{2}(\mathcal{B}_j-\mathcal{B}_i)$ then the duplicator picks   $\mathcal{B}$ in $[\mathcal{B}_i,\mathcal{B}_j]$ s.t.\ $\mathcal{B}_j-\mathcal{B}=\mathcal{A}_j-\mathcal{A}$; otherwise
\item the duplicator picks the middle element in $[\mathcal{B}_i,\mathcal{B}_j]$: if the special 2-tuple ($\overrightarrow{\mathfrak{D}}^{A,A}_{q,m},\times)$ (the green node in Figure \ref{prefix-Eq-color-2}) is in the region $[\mathcal{A},\mathcal{A}_j]$ then  $\mathcal{B}-\mathcal{B}_i=\lfloor\frac{1}{2}(\mathcal{B}_j-\mathcal{B}_i)\rfloor$, otherwise $\mathcal{B}_j-\mathcal{B}=\lfloor\frac{1}{2}(\mathcal{B}_j-\mathcal{B}_i)\rfloor$.
\end{enumerate}
\end{itemize}

\noindent Assume that in the first $k$ rounds the spoiler is restricted to pick in $\overrightarrow{\mathfrak{B}}^p_m$.

Because a 2-tuple may contain smaller 2-tuples when looking inside it, this strategy can be applied recursively until the duplicator finds an element to pick in the game  $G_S(\overrightarrow{\mathfrak{A}}^p_m,\overrightarrow{\mathfrak{B}}^p_m)$: \textit{she always picks an element that has the same label as the element picked by the spoiler in the same round}.  
 Using this strategy together with the strategy that is used in the game between two unordered structures, the duplicator can ensure that $(\overrightarrow{\mathfrak{D}}^{A,A}_{q,m},\times)$ will never be picked in the first $k$ rounds, in which the spoiler picks 2-tuples in $L_B$, and order itself will not cause a problem throughout the game  $G_S(\overrightarrow{\mathfrak{A}}^p_m,\overrightarrow{\mathfrak{B}}^p_m)$. More precisely, the following two lemmas are true:
 (Recall that $S$ is a $\Gamma$-labelled forest, and $m=rk(S)$)
\begin{lem}\label{color-order-game}
Assume that $k\in [0,m]$, and in the first $k$ rounds of the game $G_S(L_A,L_B)$, the spoiler picks 2-tuples in $L_B$.
\begin{enumerate}[label=(\roman*)]
\item At the end of the $i$-th round of $G_S(L_A,L_B)$, for each $i\leq k$,  there is only one interval in $L_A$  that is not isomorphic to the corresponding interval in $L_B$,  both of which are no shorter than $2^{m-i}$, and $(\overrightarrow{\mathfrak{D}}^{A,A}_{q,m},\times)$ always lies in this interval of $L_A$.
\item In the game $G_S(L_A,L_B)$, for any $i,j\leq k$, \[\mathcal{A}_i\leq_2 \mathcal{A}_j \hspace{4pt}\mathrm{iff}\hspace{4pt} \mathcal{B}_i\leq_2 \mathcal{B}_j.\] 
\end{enumerate}
\end{lem} 
\begin{proof} 
Before the game $G_S(L_A,L_B)$, there is only one interval in $L_A$ and $L_B$. Namely $L_A$ and $L_B$ themselves. Let $\mathcal{A}_f, \mathcal{A}_l$ be the fist and the last 2-tuples of $L_A$ respectively. And $\mathcal{B}_f$ and $\mathcal{B}_l$ be the first and the last 2-tuples of $L_B$ respectively.
\begin{enumerate}%[(i)]
\item  In the first round, there are two possibilities:
\begin{enumerate}
	\item The players pick the first 2-tuple or the last 2-tuple in the respective intervals and the intervals remain unchanged;
	\item The picked 2-tuples split $L_A$ and $L_B$ into two intervals. In this case, according to the duplicator's strategy, one interval in $L_A$ and $L_B$ has the same length, which is no larger than $2^{m-1}$, and the other interval in $L_A$ includes $(\overrightarrow{\mathfrak{D}}^{A,A}_{q,m},\times)$ because the length between the first 2-tuple (the last 2-tuple resp.) and $(\overrightarrow{\mathfrak{D}}^{A,A}_{q,m},\times)$ is the length of $L_B$ plus one which means that this length is larger than any interval in $L_B$. It also means that the interval which includes $(\overrightarrow{\mathfrak{D}}^{A,A}_{q,m},\times)$ is longer than the other one in $L_A$. Hence, it is larger than $2^{m-1}$, like the corresponding interval in $L_B$. 
\end{enumerate} 

\noindent Therefore, {(i)} holds when the game is at the end of the first round.
Assume that {(i)} holds when  the game is at the end of the $s$-th round, where $1\leq s\leq k$.

 When  the game is in the $(s+1)$-th round, there are only two cases: \begin{enumerate}%[1)] 
\item  The spoiler picks a 2-tuple that is in the interval, say $I_b$ in $L_B$.  Let the corresponding interval in $L_A$ be $I_a$, which includes the 2-tuple $(\overrightarrow{\mathfrak{D}}^{A,A}_{q,m},\times)$. By assumption, both of them are longer than $2^{m-s}$. In this round, the unique pair of intervals $I_a$ and $I_b$ are split into two pairs: one is isomorphic; the other one is not and includes $(\overrightarrow{\mathfrak{D}}^{A,A}_{q,m},\times)$ because $(\overrightarrow{\mathfrak{D}}^{A,A}_{q,m},\times)$ splits $I_a$ into two pieces and both of them  are longer than $I_b$. Note that the pair of isomorphic intervals are no longer than $2^{m-s-1}$. Hence, the other pair of intervals are no shorter than $2^{m-s-1}$, by the duplicator's strategy.     
\item The spoiler picks a 2-tuple in other intervals. Note that all the other
pairs of intervals have the same length. According to the duplicator's
strategy, splitting such a pair of intervals only produces pairs of intervals
of the same length. And the pair of intervals which are not isomorphic
is unchanged. By the inductive assumption, {(i)} still holds.
\end{enumerate}

\item  We prove an equivalent  conclusion, i.e.\ {(ii)} holds if we add two rounds before the game in which $\mathcal{A}_f, \mathcal{A}_l, \mathcal{B}_f$ and $\mathcal{B}_l$ are picked.
Clearly, $\mathcal{A}_f\leq_2\mathcal{A}_l$ iff $\mathcal{B}_f\leq_2\mathcal{B}_l$. In other words, it holds when these two elements are picked.

Suppose {(ii)} holds when the game is at the end of the $s$-th round, where $s\geq 0$.

Now assume that the game is in the $(s+1)$-th round of the game. If the spoiler picks a 2-tuple that was picked before, then it still holds. If the spoiler picks a 2-tuple $\mathcal{B}$ that splits some interval in $L_B$, say $[\mathcal{B}_h,\mathcal{B}_t]$, the duplicator also picks a 2-tuple $\mathcal{A}$ that splits the corresponding interval $[\mathcal{A}_h,\mathcal{A}_t]$ by her strategy. If $[\mathcal{B}_h,\mathcal{B}_t]$ and $[\mathcal{A}_h,\mathcal{A}_t]$ are isomorphic, then obviously {(ii)} holds. If it is not the case, 
due to {(i)}, we know that  $[\mathcal{A}_h,\mathcal{A}_t]$ is sufficiently long such that it allows such splitting. Then for any 2-tuple $\mathcal{B}_x$ that is picked before, $\mathcal{B}_x\leq_2 \mathcal{B}_h$ if  $\mathcal{B}_x\leq_2 \mathcal{B}$, which,  by assumption, implies that $\mathcal{A}_x\leq_2 \mathcal{A}_h$. Hence, $\mathcal{A}_x\leq_2\mathcal{A}$. Likewise, $\mathcal{A}_x\geq_2 \mathcal{A}$ if $\mathcal{B}_x\geq_2 \mathcal{B}$. Therefore, for any $i,j$, $\mathcal{A}_i\leq_2 \mathcal{A}_j \hspace{4pt}\underline{\mathrm{if}}\hspace{4pt} \mathcal{B}_i\leq_2 \mathcal{B}_j.$
  The ``\underline{only if}'' part can be proven similarly.\qedhere 
\end{enumerate}
\end{proof}

\noindent Part {(i)} of Lemma \ref{color-order-game} tells us that the spoiler cannot use the linear order to force the duplicator to violate her winning strategy in $G_S(\mathfrak{A}^p_m,\mathfrak{B}^p_m)$, no matter how he picks in $L_B$. That is, if the spoiler picks a yellow node in $L_B$ (see Figure \ref{prefix-Eq-color-2}), the duplicator can also pick a yellow node.
 A similar thing can be proved when $p[1]=\forall$. Because the collection of 2-tuples whose roots are children (with the same label) of an element (node) in
$\overrightarrow{\mathfrak{A}}^p_m$ form an interval, we can generalize {(ii)} of Lemma \ref{color-order-game} such that it applies to any pair of intervals that are split in the same round. Recall that a 2-tuple is composed of two trees. A 2-tuple is at depth $i$ if the roots of trees of this 2-tuple are at depth $i$. Hence, we can call these intervals \textit{intervals at depth $i$} if the elements (2-tuples) of these intervals are at depth $i$ of the structures (trees). 

\begin{lem} \label{lower-level-game}
Let $S$ be a $\Gamma$-labelled forest and $m=rk(S)$. Let $p$ be a prefix. In the game $G_S(\overrightarrow{\mathfrak{A}}^p_m,\overrightarrow{\mathfrak{B}}^p_m)$ where $p\notin \mathscr{W}(S)$, the duplicator can play in such a way that: 
\begin{enumerate}[label=(\roman*)]
\item she follows her winning strategy in $G_S(\mathfrak{A}^p_m,\mathfrak{B}^p_m)$ that has been described in Lemma \ref{EF-Game}; 
\item for any  $a_i,a_j\in |\overrightarrow{\mathfrak{A}}^p_m|$,  $b_i,b_j\in |\overrightarrow{\mathfrak{B}}^p_m|$, which are picked by the players in $i$-th and $j$-th rounds ($i,j\leq m$), \[a_i\leq^{\overrightarrow{\mathfrak{A}}^p_m} a_j \hspace{4pt}\mathrm{iff}\hspace{4pt}   b_i\leq^{\overrightarrow{\mathfrak{B}}^p_m} b_j.\]
\end{enumerate}
\end{lem}
\begin{proof}
We assume that, before the first round of the game, both the first child 
and last child of any inner node are already picked. It means that we are trying to
prove an equivalent result. 

In each round, when the spoiler picks an element $a$ in a structure, say $\overrightarrow{\mathfrak{A}}^p_m$, there is a path $P$ from the root of the tree $\overrightarrow{\mathfrak{A}}^p_m$ to the node $a$. For any node $v$ in the path $P$ there is a 2-tuple $(\mathcal{X},\mathcal{Y})$ such that $(\mathcal{X},\mathcal{Y})$ is the 2-tuple that include $a$ and $v$ is the root of either $\mathcal{X}$ or $\mathcal{Y}$. Let $a$ be at the depth $r$ of the tree $\overrightarrow{\mathfrak{A}}^p_m$ and let the path $P$ be $(a_0,\cdots,a_r)$ where $a_0$ is the root of $\overrightarrow{\mathfrak{A}}^p_m$ and $a_r=a$. Assume that the intervals of 2-tuples that are split by the path is $(L_1^A,\cdots,L_r^A)$. That is, $L_i^A$ is an interval at the depth $i$ of the tree $\overrightarrow{\mathfrak{A}}^p_m$, which includes $a_i$. Let $(L_1^B,\cdots,L_r^B)$ be the collection of intervals of 2-tuples where $L_i^B$ is at depth $i$ and is formed in the same rounds as $L_i^A$. The duplicator first picks all the nodes in the path $(a_0,\cdots,a_r)$ (Looking at it in another way, the spoiler implicitly picks all the nodes in the path). Then she uses
her strategy in coloured linear order games recursively as follows: she first plays the one round coloured linear order game at the depth 1, i.e.\ over the pair of intervals $(L_1^A,L_1^B)$, picking a 2-tuple at the depth 1
 of $\overrightarrow{\mathfrak{B}}^p_m$, which splits the interval $L_1^B$ as a consequence; then she picks the root of a tree in the 2-tuple (recall that a 2-tuple is composed of two trees) which respects her  wining strategy in  $G_S(\mathfrak{A}^p_m,\mathfrak{B}^p_m)$ (recall Definition \ref{structures-over-TAU} for the definition of $\mathfrak{A}^p_m$ and $\mathfrak{B}^p_m$): whenever possible she tries to choose  the tree, whose root is $b_1$, that is isomorphic to the one implicitly picked by the spoiler in $L_1^A$ whose root is $a_1$, and pick $b_1$. 
 Then she goes on to play the game over the pair of intervals $(L_2^A,L_2^B)$. For $1\leq i\leq r-1$, once she picked a 2-tuple at depth $i$, she will pick the root, say $c$, of a tree from the 2-tuple, afterwards she picks a 2-tuple at depth $i+1$, whose trees are the children of $c$. At last, she picks a 2-tuple in the interval $L_r^B$ and picks the root of a tree from this 2-tuple that respects her winning strategy in $G_S(\mathfrak{A}^p_m,\mathfrak{B}^p_m)$.  
 
 {(i)} will be violated only when the lengths of the pair of non-isomorphic intervals are not long enough such that the spoiler can force the duplicator to pick a different type of 2-tuple when he picks repeatedly in the shorter interval. Note that any pair of isomorphic intervals have the same type of 2-tuples. Hence, to prove {(i)}, we need only show that {(i')} at the end of the $i$-th round, if a pair of intervals is not isomorphic, then the length of them are no less than $2^{m-i}$.  However, {(i')} is obvious because the strategy of the duplicator ensures that after the  $i$-th round ($i\geq 1$), the lengths of any pair of non-isomorphic intervals reduce at most $2^{m-i}$, while they are at least $2^{m-i+1}$ before the $i$-th round, which can be proved inductively as in  Lemma \ref{color-order-game}.  In other words, this strategy is able to incorporate her winning strategy in $G_S(\mathfrak{A}^p_m,\mathfrak{B}^p_m)$ - it can guide her to win the game if order is not taken into account. It remains to show that this is a strategy for the duplicator to avoid the order problem, i.e.\ to show {(ii)}.

For the sake of convenience, we call the stage before the players play the game as the 0-th round. In the first round, the order will not be a problem since other than those nodes at the ends of intervals, there is no other node that can violate {(ii)}.

Assume that {(ii)} holds when it is at the end of the $s$-th round.

Now suppose that the game is in the $(s+1)$-th round. Let $a_i$ and $a_j$ be two elements picked in the $i$-th and $j$-th rounds ($i,j\leq s$) in $\overrightarrow{\mathfrak{A}}^p_m$. We may further assume that either $i$ or $j$ equals $s+1$. When $r, a_i, a_j$ are in a path, then $a_i<a_j$ implies $b_i<b_j$ because the duplicator's strategy ensures that $r, b_i, b_j$ are also in a path. 

Now assume that $r, a_i, a_j$ are not in a path. 

If $a_i$ and $a_j$ have the same father and the same label, which means they are in the same interval, then we can apply the same argument of  Lemma \ref{color-order-game}, simply by regarding an element as a 2-tuple. If $a_i$ and $a_j$ have the same father  but have different labels, then by definition  their order is determined by their labels. So are $b_i$ and $b_j$. Note that the label of $a_i$ and $b_i$ ($a_j$ and $b_j$ resp.) are the same, by the duplicator's strategy. Therefore, {(ii)} holds. 

Assume that $a_i$ and $a_j$ have different fathers. Note that $a_i, a_j$ always share at least one ancestor, i.e.\ the root $r$. Let $c$ be such an shared ancestor, and for all other shared ancestors, $c$ is later in the order. Let $a_i^{\prime}$ be the ancestor of $a_i$ (or $a_i$ itself) and the child of $c$. Let $a_j^{\prime}$ be the ancestor of $a_j$ (or $a_j$ itself) and the child of $c$. Let $b_i^{\prime}, b_j^{\prime}$ be defined in a similar way in $\overrightarrow{\mathfrak{B}}^p_m$. Then by definition the order between $a_i$ and $a_j$ is determined by the order between $a_i^{\prime}$ and $a_j^{\prime}$. And the duplicator can ensure that $a_i^{\prime}<a_j^{\prime}$ iff $b_i^{\prime}<b_j^{\prime}$, according to the same argument as   Lemma \ref{color-order-game}. Therefore, {(ii)} holds.
\end{proof}

Lemma \ref{lower-level-game} tell us that linear order
does not cause a problem to the duplicator, and together with the arguments in
Theorem \ref{main1} the following holds.

\begin{thm} \label{main-order-0}
Let $S_1$ and $S_2$ be two finite $\Gamma$-labelled forests. Over the class of all finite $\tau^{+\text{ORD}}$-structures,
\[
    \mbox{if   } \mathscr{W}(S_1)\nsubseteq \mathscr{W}(S_2), \mbox{   then   } \fo\{S_1\}\nsubseteq \fo\{S_2\}.
\]
\end{thm}\smallskip

\noindent Using similar arguments as in the last section, in particular the same reduction as 
in Lemma \ref{translation}, we can prove the following Theorem.

\begin{thm} \label{main-order}
Let $S_1$ and $S_2$ be two finite $\Gamma$-labelled forests. Over the class of all  ordered finite digraphs,
\[
    \mbox{if   } \mathscr{W}(S_1)\nsubseteq \mathscr{W}(S_2), \mbox{   then   } \fo\{S_1\}\nsubseteq \fo\{S_2\}.
\]
\end{thm}\smallskip

\noindent Here we call $\tau\cup\{\leq\}$-structures (linearly) ordered digraphs.

The following corollary is a special case of Theorem \ref{main-order}, where $S_1$ and $S_2$ are degenerate trees (or directed paths).  

\begin{cor} \label{two-prefixes}
Let $p,q\in \Gamma^*$. Over the class of all  ordered finite digraphs,
\[
    \mbox{if   } p\npreceq q, \mbox{   then   } \fo\{p\}\nsubseteq \fo\{q\}.
\]
\end{cor}\smallskip

\noindent Note that it is different from Gr\"adel and McColm's conjecture \cite{GradelM96}.

 A natural question is whether something similar to Theorem \ref{main1} holds, but over finite digraphs with built-in $\mathrm{BIT}$. Here, $\mathrm{BIT}$ is the binary relation for the bit operator: $\mathrm{BIT}(x,y)=1$ if the $y$-th bit of the binary representation of $x$ is $1$. The operator $\mathrm{BIT}$ seems very powerful. It is known that first-order logic equipped with $\mathrm{BIT}$ can define arbitrary algorithmic operators, including $\leq, \times, +, \mathrm{Exp}$,  and $\mathrm{Squares}$ (Schweikardt, \cite{schweikardt_arithmetic_2005}). Supprisingly, Schweikardt and Schwentick \cite{SchweikardtLinearOrder} showed that  $\mathrm{BIT}$ is similar to linear orders in terms of expressive power in first-order logic. Based on their constructions, it is not difficult to show that the quantifier structure hierarchy is strict in $\fo$, even in the presence of $\mathrm{BIT}$.

\section{A refined quantifier structure hierarchy}\label{refined-hierarchy}

\subsection{The structures and separating property.}
It is possible that two $\Gamma$-labeled forests cannot embed to each other but the set of words that can be read off them are the same. 
It is natural to conjecture that they represent different logical resources. However, the
hierarchy we defined in the last section cannot tell us about it. In the following, we are
going to show a refined strict hierarchy, which confirms this intuition: 
\begin{thm}\label{main-main}
Let $S_1$ and $S_2$ be two $\Gamma$-labelled forests. Over the class of all digraphs,
\[
         \mathrm{if}\hspace{3pt} S_1\npreceq_e S_2, \hspace{2pt} \mathrm{then}\hspace{3pt} \fo\{S_1\}\nsubseteq \fo\{S_2\}.
\]
\end{thm}
\vspace{7pt}
 
\noindent As in the last section, we let $\tau^+:=\langle R,B,r,U \rangle$, and let $\tau^+_0:=\tau^+\setminus \{U\}$.

\begin{defi}
A $\Gamma$-labelled tree $\mathcal{T}$ is an irreducible tree if for any inner node $a$ the following holds: 

Let $b_1,\cdots,b_k$ be the children of $a$ and $\mathcal{T}_{1},\cdots,\mathcal{T}_{k}$ be the maximal subtrees of $\mathcal{T}$ that are rooted at $b_1,\cdots,b_k$ respectively, then $\mathcal{T}_{i}$ cannot embed in  $\mathcal{T}_{j}$ for any $i,j\in [1,k]$ where $i\neq j$.                   
\end{defi}

\begin{defi} \label{Refined-Struc}
 Let $\mathcal{T}$ be a $\Gamma$-labelled irreducible $k$-ary tree.
\begin{itemize}
\item $\widetilde{\mathfrak{A}}^{T}_m$ and $\widetilde{\mathfrak{B}}^{T}_m$ are coloured trees.\footnote{Here, the ``$T$'' in $\widetilde{\mathfrak{A}}^{T}_m$ and $\widetilde{\mathfrak{B}}^{T}_m$ refers to the tree $\mathcal{T}$. Also cf. Figure \ref{E(ET1)(AT2)}.} The constant $r$ is interpreted as the root of the respective trees. As in Definition \ref{structures-over-TAU}, we say an element $a$ is black if $a\in U$.
\item{(1)}  $\widetilde{\mathfrak{A}}^\exists_{m,l}$ is a depth 1 tree that has $lm+1$ leaves. One of the leaves is black and all the other leaves are not black. 

\noindent{(2)} $\widetilde{\mathfrak{B}}^\exists_{m,l}$ is a depth 1 tree that has $lm$ leaves. \textit{None} of them is black.

All edges in $\widetilde{\mathfrak{A}}^{\exists}_{m,l}$ and $\widetilde{\mathfrak{B}}^{\exists}_{m,l}$ are red. Recall that we say an edge $(a,b)$ is red if, and only if, $(a,b)\in R$.

$\mathcal{T}^{\mathfrak{A},\exists}_{m,l}:=\widetilde{\mathfrak{A}}^\exists_{m,l}|\tau^+_0$; $\mathcal{T}^{\mathfrak{B},\exists}_{m,l}:=\widetilde{\mathfrak{B}}^\exists_{m,l}|\tau^+_0$. 

\item{(1)} $\widetilde{\mathfrak{A}}^\forall_{m,l}$ is a depth 1  tree that has $lm$ leaves. All of the leaves are black. 

\noindent{(2)} $\widetilde{\mathfrak{B}}^\forall_{m,l}$ is a depth 1  tree that has $lm+1$ leaves. One of them is not black and all the other leaves are black. 

All edges in $\widetilde{\mathfrak{A}}^{\forall}_{m,l}$ and $\widetilde{\mathfrak{B}}^{\forall}_{m,l}$ are red.

$\mathcal{T}^{\mathfrak{A},\forall}_{m,l}:=\widetilde{\mathfrak{A}}^\forall_{m,l}|\tau^+_0$; $\mathcal{T}^{\mathfrak{B},\forall}_{m,l}:=\widetilde{\mathfrak{B}}^\forall_{m,l}|\tau^+_0$. 

\item If $\mathcal{T}$ contains a single $\mathscr{E}$ node,\\
 \indent $\hspace{12pt}\widetilde{\mathfrak{A}}^T_m$ is $\widetilde{\mathfrak{A}}^\exists_{m,1}$; $\widetilde{\mathfrak{B}}^T_m$ is $\widetilde{\mathfrak{B}}^\exists_{m,1}$.

\item If $\mathcal{T}$ contains a single $\mathscr{A}$ node, \\
 \indent $\hspace{12pt}\widetilde{\mathfrak{A}}^T_m$ is $\widetilde{\mathfrak{A}}^\forall_{m,1}$; $\widetilde{\mathfrak{B}}^T_m$ is $\widetilde{\mathfrak{B}}^\forall_{m,1}$.

\item When $|rk(\mathcal{T})|>1:$\\[3pt]
\indent\hspace{6pt} Assume that the root $r$ of $\mathcal{T}$ has $k$ children, and the maximal subtrees that are rooted at these  children are $\mathcal{T}_1,\cdots,\mathcal{T}_k$ respectively. Recall that $\mathcal{T}$ is an irreducible tree, which implies  $\mathcal{T}_i$ and $\mathcal{T}_j$ are not isomorphic if $i\neq j$.

\indent\hspace{6pt} For \textit{any} $\Gamma$-labelled tree $\mathcal{T}^{\prime}$, let $\widetilde{\mathfrak{A}}^{-T^{\prime}}_m$ be the same as $\widetilde{\mathfrak{A}}^{T^{\prime}}_m$ except that the colours of all the edges are exchanged, i.e.\ red is exchanged with blue. Let $P_i$, which is a member of the set $\{(A,A),(A,B),(B,A),(B,B),A,B\}$, and $\mathfrak{D}^{P_1,\cdots,P_k}_{T_1,\cdots,T_k,m}$, which is a $(\tau^+\cup\{\mathfrak{e}\}\setminus\{r\})$-structure where $\mathfrak{e}$ is a hook constant, be built by the following process:
\begin{center}
\begin{enumerate}[label=step \arabic*:]
\item Let $M:=\emptyset$;
\item For $i=1$ to $k$ do the following:
\begin{itemize}
\item if  $P_i=(A,A)$ then $M:=M\cup\{\widetilde{\mathfrak{A}}^{T_i}_m,\widetilde{\mathfrak{A}}^{-T_i}_m\}$; 
\item if  $P_i=(A,B)$ then $M:=M\cup\{\widetilde{\mathfrak{A}}^{T_i}_m,\widetilde{\mathfrak{B}}^{-T_i}_m\}$;
\item if  $P_i=(B,A)$ then $M:=M\cup\{\widetilde{\mathfrak{B}}^{T_i}_m,\widetilde{\mathfrak{A}}^{-T_i}_m\}$;
\item if  $P_i=(B,B)$ then $M:=M\cup\{\widetilde{\mathfrak{B}}^{T_i}_m,\widetilde{\mathfrak{B}}^{-T_i}_m\}$;
\item if  $P_i=A$ then $M:=M\cup\{\widetilde{\mathfrak{A}}^{T_i}_m\}$;
\item if  $P_i=B$ then $M:=M\cup\{\widetilde{\mathfrak{B}}^{T_i}_m\}$;
\end{itemize}
\item join all the trees in $M$ at their roots. Call their shared root a \textit{junction point}, which interprets the hook constant $\mathfrak{e}$.
\end{enumerate}
%}
\end{center}
We use $O_i$ to denote the label of the root of the tree $\mathcal{T}_i$.  Let $H:=\{i\in [1,k]\mid \text{ the root of }\mathcal{T}_i$  has the same label as that of the root of  $\mathcal{T}\}$.

In the following, we define some substructure for the constructions:
\begin{itemize}%[$\diamond$]
\item For any $j\in H$, let $\mathfrak{C}_{\exists,j,m}^{A,B,T}:=\mathfrak{D}^{P_1,\cdots,P_k}_{T_1,\cdots,T_k,m}$  where  $P_j=(A,B)$  and  [ if  $i\neq j$  then  $(P_i=(A,A)$ if  $O_i=\exists)$  and $(P_i=A$  if  $O_i=\forall)]$; 
\item For any $j\in H$, let  $\mathfrak{C}_{\exists,j,m}^{B,A,T}:=\mathfrak{D}^{P_1,\cdots,P_k}_{T_1,\cdots,T_k,m}$  where  $P_j=(B,A)$ and  [ if  $i\neq j$  then  $(P_i=(A,A)$ if  $O_i=\exists)$  and $(P_i=A$ if  $O_i=\forall)]$; 
\item For any $j\in [1,k]$ and $j\notin H$, let  $\mathfrak{C}_{\exists,j,m}^{B,B,T}:=\mathfrak{D}^{P_1,\cdots,P_k}_{T_1,\cdots,T_k,m} \text{ where } P_j=B \text{ and } [\text{ if } i\neq j \text { then } (P_i=(A,A)\text{ if } O_i=\exists)$ and $(P_i=A \text{ if } O_i=\forall)]$; 
\item We use $\mathfrak{C}_{\exists,0,m}^{A,A,T}$ to denote the structure  $\mathfrak{D}^{P_1,\cdots,P_k}_{T_1,\cdots,T_k,m}$  where $P_i=(A,A)$  if  $O_i=\exists; P_i=A$  if  $O_i=\forall$.
\end{itemize}

 Dually, we can define $\mathfrak{C}_{\forall,j,m}^{B,A,T}, \mathfrak{C}_{\forall,j,m}^{A,B,T}, \mathfrak{C}_{\forall,j,m}^{A,A,T}, \mathfrak{C}_{\forall,0,m}^{B,B,T}$ as follows: 
\begin{itemize}%[*] 
\item For any $j\in H$, let $\mathfrak{C}_{\forall,j,m}^{A,B,T}:=\mathfrak{D}^{P_1,\cdots,P_k}_{T_1,\cdots,T_k,m}$ where  $P_j=(A,B)$  and  [ if  $i\neq j$  then  $(P_i=(B,B)$ if  $O_i=\forall)$  and $(P_i=B$  if  $O_i=\exists)]$; 
\item For any $j\in H$, let  $\mathfrak{C}_{\forall,j,m}^{B,A,T}:=\mathfrak{D}^{P_1,\cdots,P_k}_{T_1,\cdots,T_k,m}$  where  $P_j=(B,A)$  and  [ if  $i\neq j$  then  $(P_i=(B,B)$ if  $O_i=\forall)$  and $(P_i=B$  if  $O_i=\exists)]$; 
\item For $j\notin H$, let  $\mathfrak{C}_{\forall,j,m}^{A,A,T}:=\mathfrak{D}^{P_1,\cdots,P_k}_{T_1,\cdots,T_k,m}$  where  $P_j=A$  and  [ if  $i\neq j$  then  $(P_i=(B,B)$ if  $O_i=\forall)$  and $(P_i=B$  if  $O_i=\exists)]$; 
\item We use $\mathfrak{C}_{\forall,0,m}^{B,B,T}$ to denote the structure  $\mathfrak{D}^{P_1,\cdots,P_k}_{T_1,\cdots,T_k,m}$  where $P_i=(B,B)$  if  $O_i=\forall; P_i=B$  if  $O_i=\exists$.
\end{itemize}

\noindent  As usual, we define the main structures based on point-expansions over some sets of structures. Now we define such sets.

\begin{itemize}%[$\diamond$]
\item $\mathcal{K}_0^{\exists,T}:=\{\mathfrak{C}^{A,B,T}_{\exists,i,m}\mid i\in H\}\cup\{\mathfrak{C}^{B,A,T}_{\exists,i,m}\mid i\in H\}\cup\{\mathfrak{C}^{B,B,T}_{\exists,i,m}\mid i\in [1,k]$ and  $i\notin H\}$;
\item $\mathcal{K}_A^{\exists,T}:=\{\mathfrak{C}^{A,A,T}_{\exists,0,m}\}\cup\{\mathcal{I}\}\cup \mathcal{K}_0^{\exists,T}$;
\item $\mathcal{K}_B^{\exists,T}:=\{\mathcal{I}\}\cup \mathcal{K}_0^{\exists,T}$.   
\end{itemize}

\begin{itemize}%[*]
\item $\mathcal{K}_0^{\forall,T}:=\{\mathfrak{C}^{A,B,T}_{\forall,i,m}\mid i\in H\}\cup\{\mathfrak{C}^{B,A,T}_{\forall,i,m}\mid i\in H\}\cup\{\mathfrak{C}^{A,A,T}_{\forall,i,m}\mid i\in [1,k] \text { and } i\notin H\}$;
\item $\mathcal{K}_A^{\forall,T}:=\{\mathcal{I}\}\cup \mathcal{K}_0^{\forall,T}$;
\item $\mathcal{K}_B^{\forall,T}:=\{\mathfrak{C}^{B,B,T}_{\forall,0,m}\}\cup\{\mathcal{I}\}\cup \mathcal{K}_0^{\forall,T}$.  
\end{itemize}

 Now we define $\widetilde{\mathfrak{A}}^T_m$ and $\widetilde{\mathfrak{B}}^T_m$ as follows:
\begin{itemize}
\item If the root of $\mathcal{T}$ is an $\mathscr{E}$ node, which is connected to the roots of $k$ trees $\mathcal{T}_1, \cdots,\mathcal{T}_k$, and assume that the number of $\exists$ in the tuple $(O_1,\cdots,O_k)$ is $\jmath$, then $\widetilde{\mathfrak{A}}^T_m$  and $\widetilde{\mathfrak{B}}^T_m$ are defined as follows:
 \begin{itemize}%[$\diamond$]
\item $\widetilde{\mathfrak{A}}^T_m$ is a point-expansion of $\mathcal{T}^{\mathfrak{A},\exists}_{m,\jmath+k}$ over $\mathcal{K}_A^{\exists,T}$: The root of $\mathcal{T}^{\mathfrak{A},\exists}_{m,\jmath+k}$ is expanded by $\mathcal{I}$. It is also called the ``\textit{root}'' of $\widetilde{\mathfrak{A}}^T_m$ that interprets $r$. One leaf is expanded by a copy of $\mathfrak{C}^{A,A,T}_{\exists,0,m}$. For  each element of $\mathcal{K}_0^{\exists,T}$ there are exactly $m$ distinct  leaves which are expanded by it.
     
\item $\widetilde{\mathfrak{B}}^T_m$ is a point-expansion of $\mathcal{T}^{\mathfrak{B},\exists}_{m,\jmath+k}$ over $\mathcal{K}_B^{\exists,T}$, similar to $\widetilde{\mathfrak{A}}^T_m$: The root of $\mathcal{T}^{\mathfrak{B},\exists}_{m,\jmath+k}$ is expanded by $\mathcal{I}$. For  each element of $\mathcal{K}_0^{\exists,T}$ there are exactly $m$ distinct leaves which are expanded by it. 

\end{itemize}

\item If the root of $\mathcal{T}$ is an $\mathscr{A}$ node, which is connected to the roots of $k$ trees $\mathcal{T}_1, \cdots,\mathcal{T}_k$, and assume that the number of $\forall$ in the tuple $(O_1,\cdots,O_k)$ is $\jmath$, then $\widetilde{\mathfrak{A}}^T_m$  and $\widetilde{\mathfrak{B}}^T_m$ are defined as follows:
\begin{itemize}%[*]
\item $\widetilde{\mathfrak{A}}^T_m$ is a point-expansion of $\mathcal{T}^{\mathfrak{A},\forall}_{m,\jmath+k}$ over $\mathcal{K}_A^{\forall,T}$: The root of $\mathcal{T}^{\mathfrak{A},\forall}_{m,\jmath+k}$ is expanded by $\mathcal{I}$.  For  each element of $\mathcal{K}_0^{\forall,T}$ there are exactly $m$ distinct leaves which are expanded by it. 
\item $\widetilde{\mathfrak{B}}^T_m$ is a point-expansion of $\mathcal{T}^{\mathfrak{B},\forall}_{m,\jmath+k}$ over $\mathcal{K}_B^{\forall,T}$: The root of $\mathcal{T}^{\mathfrak{B},\forall}_{m,\jmath+k}$ is expanded by $\mathcal{I}$. One leaf is expanded by a copy of $\mathfrak{C}^{B,B,T}_{\forall,0,m}$. For  each element of $\mathcal{K}_0^{\forall,T}$ there are exactly $m$ distinct leaves which are expanded by it.

\end{itemize}

\end{itemize}
\end{itemize} 

                       \end{defi}

\begin{exa} Let $\mathcal{T}$ be a $\Gamma$-labelled irreducible binary tree.  Assume that its root is an $\mathscr{E}$ node and is connected to two subtrees $\mathcal{T}_1$ and $\mathcal{T}_2$, where the root of $\mathcal{T}_1$ is an $\mathscr{E}$ node and the root of $\mathcal{T}_2$ is an $\mathscr{A}$ node. See Figure \ref{E(ET1)(AT2)} for the illustration of the structures  $\widetilde{\mathfrak{A}}^T_m$ and  $\widetilde{\mathfrak{B}}^T_m$. A ``*'' at the root of a subtree $\mathcal{T}^{\prime}$ means that we have $m$ disjoint isomorphic copies of this tree $\mathcal{T}^{\prime}$ and for each copy we add an edge between the root, $r$, of the whole structure and the root of this copy. 
\end{exa}

\begin{figure}[h]
\centering
\includegraphics[trim = 0mm 35mm 0mm 40mm, clip, width=10cm]{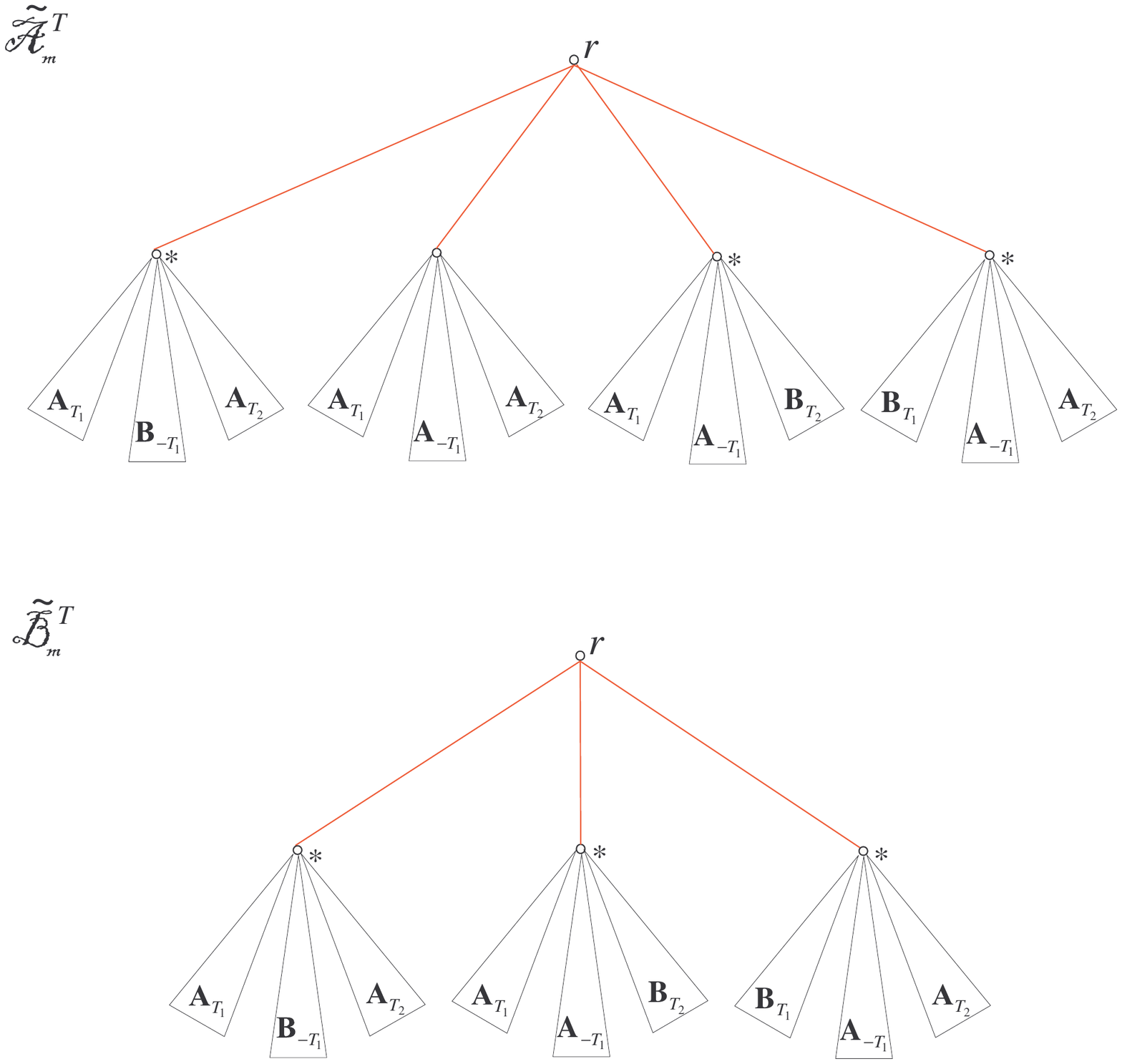}
\caption{The structures $\widetilde{\mathfrak{A}}^T_m$ and $\widetilde{\mathfrak{B}}^T_m$, where $\mathcal{T}$ is an irreducible binary tree, which is connected to $\mathcal{T}_1$ and $\mathcal{T}_2$. The root of $\mathcal{T}$ is an $\mathscr{E}$ node; The root of $\mathcal{T}_1$ is an $\mathscr{E}$ node; and the root of $\mathcal{T}_2$ is an $\mathscr{A}$ node. In the figure, $A_{T_i}$ and $B_{T_i}$ represent $\widetilde{\mathfrak{A}}^{T_i}_m$ and $\widetilde{\mathfrak{B}}^{T_i}_m$ respectively.}
\label{E(ET1)(AT2)}
\end{figure}

Now we define a sentence $\widetilde{\phi}_{\mathcal{T}}$ over the signature $\tau^+$ such that, for any $m$,  $\widetilde{\mathfrak{A}}^T_m\models \widetilde{\phi}_{\mathcal{T}}$ but $\widetilde{\mathfrak{B}}^T_m\not\models \widetilde{\phi}_{\mathcal{T}}$.

\begin{defi}\label{sentence-for-property-sigma_F}
Let $\mathcal{T}$ be a $\Gamma$-labelled irreducible $k$-ary  tree, which is connected to the roots of $k$ trees $\mathcal{T}_1, \cdots,\mathcal{T}_k$.  Assume that $rk(\mathcal{T})=d$. Recall that $H=\{i\in [k]\mid \text{ the root of }\mathcal{T}_i$  has the same label as that of the root of  $\mathcal{T}\}$. We define $\tau^+$-sentences  $\widetilde{\phi}_{\mathcal{T}}$ and $\widetilde{\phi}_{-\mathcal{T}}$ based on the tree $\mathcal{T}$ as follows: 
\begin{enumerate}
\item If $\mathcal{T}$ is empty, then 
\[\widetilde{\xi}_{\mathcal{T}}(x)=\widetilde{\xi}_{-\mathcal{T}}(x):=U(x);\]  
\item If the root of $\mathcal{T}$ is an $\mathscr{E}$ node, then 
\[\widetilde{\xi}_{\mathcal{T}}(y):=\exists  x_{d}(Ryx_{d}\land \displaystyle\bigwedge_{i\in [k]} \widetilde{\xi}_{\mathcal{T}_i}(x_{d})\land \displaystyle\bigwedge_{j\in H} \widetilde{\xi}_{-\mathcal{T}_j}(x_{d}));\]

\[\widetilde{\xi}_{-\mathcal{T}}(y):=\exists  x_{d}(Byx_{d}\land \displaystyle\bigwedge_{i\in [k]} \widetilde{\xi}_{-\mathcal{T}_i}(x_{d})\land \displaystyle\bigwedge_{j\in H} \widetilde{\xi}_{\mathcal{T}_j}(x_{d}));\]
\item If the root of $\mathcal{T}$ is an $\mathscr{A}$ node, then 
\[\widetilde{\xi}_{\mathcal{T}}(y):=\forall x_{d}(Ryx_{d}\rightarrow \displaystyle\bigvee_{i\in [k]} \widetilde{\xi}_{\mathcal{T}_i}(x_{d})\lor\displaystyle\bigvee_{j\in H}\widetilde{\xi}_{-\mathcal{T}_j}(x_{d}));\]

\[\widetilde{\xi}_{-\mathcal{T}}(y):=\forall x_{d}(Byx_{d}\rightarrow \displaystyle\bigvee_{i\in [k]} \widetilde{\xi}_{-\mathcal{T}_i}(x_{d})\lor\displaystyle\bigvee_{j\in H}\widetilde{\xi}_{\mathcal{T}_j}(x_{d}));\]
\end{enumerate}

Now, $\widetilde{\phi}_{\mathcal{T}}$ and $\widetilde{\phi}_{-\mathcal{T}}$ are defined as:
\begin{gather}\label{separate-sentence-general}
\widetilde{\phi}_{\mathcal{T}}:=\widetilde{\xi}_{\mathcal{T}}(r);\\ \widetilde{\phi}_{-\mathcal{T}}:=\widetilde{\xi}_{-\mathcal{T}}(r).
\end{gather}
                       \end{defi}

\subsection{The duplicator's winning strategy.} In the following, we show that a refined strict quantifier hierarchy exists by proving that the duplicator has a winning strategy in the games that we introduced before. Since the proof resembles that in Section \ref{games-and-hierarchy}, we sketch the main ideas and omit similar details. 

\begin{lem} \label{EF-Game-refined}
Let $S_1$ and $S_2$ be two finite $\Gamma$-labelled forests. If $S_1\npreceq_e S_2$,  then there is an irreducible subtree $\mathcal{T}$ in $S_1$ such that the duplicator has a winning strategy in $G_{S_2}(\widetilde{\mathfrak{A}}^{T}_m,\widetilde{\mathfrak{B}}^{T}_m)$ for any $m\geq rk(S_2)$.
  
\end{lem}
\begin{proof}
Observe that when $rk(S_1)=1$, i.e.\ the rank of $S_1$ is 1, $S_1\npreceq_e S_2$ implies $\mathscr{W}(S_1)\nsubseteq \mathscr{W}(S_2)$. According to Lemma \ref{EF-Game}, the duplicator has a winning strategy.  

Assume that the lemma holds when $rk(S_1)\leq h$. 

Now assume that $rk(S_1)=h+1$.

Because $S_1\npreceq_e S_2$,  for some $k\in \mathbb{N}^+$ there exists a $k$-ary  subtree $\mathcal{T}$ in $S_1$ such that $\mathcal{T}\npreceq_e \mathcal{T}^{\prime}$ for any subtree $\mathcal{T}^{\prime}$ in $S_2$. Moreover, we assume that $\mathcal{T}$ is the minimal subtree in $S_1$ that is not embeddable in $S_2$, i.e.\  any subtree of $\mathcal{T}$, which is not $\mathcal{T}$ itself, is embeddable in $S_2$. Note that such a tree is an irreducible tree, whose rank is no larger than $h+1$. Here we use $r(\mathcal{T})$ to denote the root of $\mathcal{T}$. And we assume that $r(\mathcal{T})$ is an $\mathscr{E}$ node. We assume that there are $1\leq k^{\prime}\leq k$ subtrees, $\mathcal{T}_1,\cdots,\mathcal{T}_{k^{\prime}}$, whose roots are connected to $r(\mathcal{T})$. Let $\text{Mg}$ be the set of $\mathscr{E}$ nodes in $S_2$ such that for any $a\in \text{Mg}$ no other $\mathscr{E}$ node appears in the path from the root of $S_2$ to $a$. 
By assumption $\mathcal{T}$ is not embeddable in any subtree of $S_2$. For any $a\in \text{Mg}$, let  $\mathcal{T}^{\prime}$ be a tree rooted at $a$ and $\mathcal{F}$ be the forest obtained from $\mathcal{T}^{\prime}$ by removing the root of $\mathcal{T}^{\prime}$. Hence, at least one of the trees $\mathcal{T}_1,\cdots,\mathcal{T}_{k^{\prime}}$ cannot be embedded in $\mathcal{F}$. Observe that $rk(\mathcal{T}_i)\leq h$ for any $i\in [k^{\prime}]$.

Note that, in any play of the game, the moving track of the token in $S_2$ is a directed path. If the first place where the token lies is an $\mathscr{A}$ node, then the path is initiated with a block of universal quantifiers. However, in the rounds based on this  first block of universal quantifiers the spoiler has to pick in $\widetilde{\mathfrak{B}}^{T}_m$ and no matter how he picks the duplicator can mimic his picking in the isomorphic subtrees. By Lemma \ref{pick-in-isom}, if the spoiler can win the game by picking these elements, he can also win the game without picking these elements. 

When the token is on a node of $\text{Mg}$, say $a$, we assume that the spoiler picks in $\mathfrak{C}^{A,A,T}_{\exists,0,m}$ (i.e.\ in a copy of  $\widetilde{\mathfrak{A}}^{T}_m$) because if he picks in other places the duplicator can mimic his picking in an isomorphic tree substructure, hence by Lemma \ref{pick-in-isom} the spoiler can win the game without picking in these places if he has at least one winning strategy. In this round, the duplicator's strategy is as follows (i.e.\ picking in a copy of $\widetilde{\mathfrak{B}}^{T}_m$):
\begin{itemize}
\item If $\mathcal{T}_i$ is not embeddable in the forest $\mathcal{F}$ and $r(\mathcal{T}_i)$ is an $\mathscr{E}$ node, then 
     \begin{itemize}
     \item if the spoiler picks inside   $\widetilde{\mathfrak{A}}^{T_i}_m$, which is a part of $\mathfrak{C}^{A,A,T}_{\exists,0,m}$, then the duplicator mimics it in an isomorphic copy of $\widetilde{\mathfrak{A}}^{T_i}_m$, which is a part of  $\mathfrak{C}^{A,B,T}_{\exists,i,m}$.
      \item if the spoiler picks inside   $\widetilde{\mathfrak{A}}^{-T_i}_m$, which is a part of $\mathfrak{C}^{A,A,T}_{\exists,0,m}$, then the duplicator mimics it in an isomorphic copy of $\widetilde{\mathfrak{A}}^{-T_i}_m$, which is a part of  $\mathfrak{C}^{B,A,T}_{\exists,i,m}$.
     \item if the spoiler picks inside  $\widetilde{\mathfrak{A}}^{T_j}_m$ or $\widetilde{\mathfrak{A}}^{-T_j}_m$ where $j\neq i$, then the duplicator mimics it in an isomorphic copy of   $\widetilde{\mathfrak{A}}^{T_j}_m$ or $\widetilde{\mathfrak{A}}^{-T_j}_m$, which is a part of $\mathfrak{C}^{A,B,T}_{\exists,i,m}$.
     \item if the spoiler picks the junction point, i.e.\ the root, of $\mathfrak{C}^{A,A,T}_{\exists,0,m}$, then the duplicator picks the junction point of $\mathfrak{C}^{A,B,T}_{\exists,i,m}$.
     \end{itemize}
     
\item If $\mathcal{T}_i$ is not embeddable in $\mathcal{F}$ and $r(\mathcal{T}_i)$ is an $\mathscr{A}$ node, then 
\begin{itemize}
\item if the spoiler picks inside $\widetilde{\mathfrak{A}}^{T_i}_m$, which is a part of $\mathfrak{C}^{A,A,T}_{\exists,0,m}$, 
the duplicator is able to mimic the spoiler's picks inside $\widetilde{\mathfrak{B}}^{T_i}_m$, which is a part of $\mathfrak{C}^{B,B,T}_{\exists,i,m}$. It is because the root of $\mathcal{T}_i$ is an $\mathscr{A}$ node, which means that the number of different types of subtrees, whose roots are connected to the root of $\widetilde{\mathfrak{A}}^{T_i}_m$ (in $\widetilde{\mathfrak{A}}^T_m$), is less than the number of different types of subtrees, whose roots are connected to the root of $\widetilde{\mathfrak{B}}^{T_i}_m$ (in $\widetilde{\mathfrak{B}}^T_m$). 
 \item if the spoiler picks inside  $\widetilde{\mathfrak{A}}^{T_j}_m$ or $\widetilde{\mathfrak{A}}^{-T_j}_m$ where $j\neq i$, then the duplicator mimics it in an isomorphic copy of   $\widetilde{\mathfrak{A}}^{T_j}_m$ or $\widetilde{\mathfrak{A}}^{-T_j}_m$, which is a part of $\mathfrak{C}^{B,B,T}_{\exists,i,m}$.
    \item if the spoiler picks the junction point, i.e.\ the root, of $\mathfrak{C}^{A,A,T}_{\exists,0,m}$, then the duplicator picks the junction point of $\mathfrak{C}^{B,B,T}_{\exists,i,m}$.
\end{itemize}
\end{itemize}

\noindent Note that such a strategy has exploited the feature of the structures on the game board: (recall that the root of $\mathcal{T}$ is an $\mathscr{E}$ node)  the structures are constructed in such a way that if the spoiler does not pick the junction point of $\mathfrak{C}^{A,A,T}_{\exists,0,m}$  in $\widetilde{\mathfrak{A}}^{T}_m$ then the duplicator is able to mimic his picking in an isomorphic tree substructure of $\widetilde{\mathfrak{B}}^{T}_m$ - such isomorphic subtrees  always exist. And by Lemma \ref{pick-in-isom} the spoiler can also win without such picking if he can win in any way. If the spoiler pick the junction point of $\mathfrak{C}^{A,A,T}_{\exists,0,m}$,  by her strategy, the duplicator picks the junction point of $\mathfrak{C}^{A,B,T}_{\exists,i,m}$  where $\mathcal{T}_i\npreceq_e \mathcal{F}$. Observe that, the difference between $\mathfrak{C}^{A,A,T}_{\exists,0,m}$  and $\mathfrak{C}^{A,B,T}_{\exists,i,m}$ is the difference between $\widetilde{\mathfrak{A}}^{-T_i}_m$ and $\widetilde{\mathfrak{B}}^{-T_i}_m$ (or equivalently the difference between $\widetilde{\mathfrak{A}}^{T_i}_m$ and $\widetilde{\mathfrak{B}}^{T_i}_m$). Therefore, using this strategy, no matter how the spoiler picks, the duplicator can reply properly such that in the end the spoiler can win the game $G_{\mathcal{F}}(\widetilde{\mathfrak{A}}^{-T_i}_m,\widetilde{\mathfrak{B}}^{-T_i}_m)$ (or $G_{\mathcal{F}}(\widetilde{\mathfrak{A}}^{T_i}_m,\widetilde{\mathfrak{B}}^{T_i}_m)$), if the spoiler can win  in any way. Recall that $rk(\mathcal{T}_i)\leq h$. By induction assumption, the duplicator can win the game using this strategy. Therefore, the spoiler cannot win the game if the duplicator plays according to this strategy. 
 In other words, this strategy is a winning strategy for the duplicator in the game $G_{S_2}(\widetilde{\mathfrak{A}}^{T}_m,\widetilde{\mathfrak{B}}^{T}_m)$. 

 When $r(\mathcal{T})$ is an $\mathscr{A}$ node, the analysis is similar.  
\end{proof}

%\begin{rem} 
\noindent Note that the duplicator's winning strategy in Lemma \ref{EF-Game-refined} depends on how the spoiler moves the token. That is, she has to keep an eye on the token track before she make the choices. 
          
%\end{rem}

The following corollary is directly due to Lemma \ref{m2m-} and Theorem \ref{qs-game}.
\begin{cor}
Let $S_1$ and $S_2$ be two finite $\Gamma$-labelled forests. If $S_1\npreceq_e S_2$,  then for any $\psi$ such that $qs(\psi)\preceq_e S_2$, and any $m\geq rk(S_2)$, we have: 
\[
\widetilde{\mathfrak{A}}^{S_1}_m\models \psi\Rightarrow\widetilde{\mathfrak{B}}^{S_1}_m\models \psi.
\]
\end{cor}

\noindent The following lemma is similar to Lemma \ref{claim1}. Actually, the proof resembles that of Lemma \ref{claim1}. 
\begin{lem}\label{claim1-refined}
Let $\mathcal{T}$ be $\Gamma$-labelled irreducible tree. Then 
\[
   \widetilde{\mathfrak{A}}^T_m\models\widetilde{\phi}_{\mathcal{T}} \text{  but  } \widetilde{\mathfrak{B}}^T_m\not\models\widetilde{\phi}_{\mathcal{T}}.   
\]  
\end{lem}
\begin{proof}
It is obvious when $\mathcal{T}$ is a single node. Assume that it holds when $rk(\mathcal{T})=h$.

Now assume that $\mathcal{T}$ is a $k$-ary $\Gamma$-labelled tree and its rank is $h+1$. Suppose that the root of $\mathcal{T}$ is an $\mathscr{E}$ node, and its subtrees, whose roots are children of $r(\mathcal{T})$, are $\mathcal{T}_1,\cdots,\mathcal{T}_{k}$. Recall that $H=\{i\in [k]\mid \text{ the root of }\mathcal{T}_i$  has the same label as that of the root of  $\mathcal{T}\}$.

Let $a_r$ be the node that interprets $r$ in $\widetilde{\mathfrak{A}}^T_m$. And let $b_r$ be the node that interprets the hook constant in $\mathfrak{C}^{A,A,T}_{\exists,0,m}$, i.e.\ $b_r$ is the junction point of $\mathfrak{C}^{A,A,T}_{\exists,0,m}$. By assumption and an observation similar to Lemma \ref{Dual-Hintikka},  $\widetilde{\mathcal{A}}^{T_i}_m\models \widetilde{\xi}_{\mathcal{T}_i}(r)$ ($i\in [k]$), and $\widetilde{\mathcal{A}}^{-T_j}_m\models \widetilde{\xi}_{-\mathcal{T}_j}(r)$ ($j\in H$).

Moreover, all the quantifiers in $\widetilde{\xi}_{\mathcal{T}_i}(r)$ are relativized by relations either $Ryx$ or $Byx$, where $x$ is the quantified variable. It means that $\widetilde{\xi}_{\mathcal{T}_i}(r)$ expresses some property that is nothing to do with the elements of $\widetilde{\mathfrak{A}}^{T}_m$ outside the tree substructure $\widetilde{\mathfrak{A}}^{T_i}_m$. As a consequence, $\widetilde{\mathfrak{A}}^{T_i}_m\models\widetilde{\xi}_{\mathcal{T}_i}(r)$ implies $\widetilde{\mathfrak{A}}^{T}_m\models\widetilde{\xi}_{\mathcal{T}_i}(b_r)$.\\[4pt]
\indent By the same argument, $\widetilde{\mathfrak{A}}^{-T_i}_m\models\widetilde{\xi}_{-\mathcal{T}_i}(r)$ implies $\widetilde{\mathfrak{A}}^{T}_m\models\widetilde{\xi}_{-\mathcal{T}_i}(b_r)$.\\[4pt]
\indent Therefore, in $\widetilde{\mathfrak{A}}^{T}_m$, $b_r$ is the witness of the quantifier $\exists x_{d}$ in the formula $\widetilde{\xi}_{\mathcal{T}}(r)$. That is, $\widetilde{\mathfrak{A}}^{T}_m\models\widetilde{\xi}_{\mathcal{T}}(r)$, or $\widetilde{\mathfrak{A}}^{T}_m\models\widetilde{\phi}_{\mathcal{T}}$. 

Likewise, by assumption and similar analysis,  
in $\widetilde{\mathfrak{B}}^T_m$ the subtrees $\mathfrak{C}^{A,B,T}_{\exists,i,m}$, or $\mathfrak{C}^{B,A,T}_{\exists,i,m}$, or $\mathfrak{C}^{B,B,T}_{\exists,i,m}$, has exactly one subtree $\widetilde{\mathfrak{B}}^{T_i}_m$, or $\widetilde{\mathfrak{B}}^{-T_i}_m$, which does not satisfy  some subformula of $\widetilde{\phi}_{\mathcal{T}}$ (i.e.\ one conjunct of $\displaystyle\bigwedge_{i\in [k]} \widetilde{\xi}_{\mathcal{T}_i}(x_{d})\land \displaystyle\bigwedge_{j\in H} \widetilde{\xi}_{-\mathcal{T}_j}(x_{d})$), when its root interprets $x_{d}$. In other words, $b_r$ cannot be a witness of the quantifer $\exists x_{d}$. Therefore, $\widetilde{\mathfrak{B}}^T_m\not\models \widetilde{\phi}_{\mathcal{T}}$.

When the root of $\mathcal{T}$ is an $\mathscr{A}$ node, the analysis is similar.
\end{proof}
  
\begin{thm}\label{main-main-tau-plus}
Let $S_1$ and $S_2$ be two $\Gamma$-labelled forests. Over the class of all finite $\tau^+$-structures,
\[
         \mathrm{if}\hspace{3pt} S_1\npreceq_e S_2, \hspace{2pt} \mathrm{then}\hspace{3pt} \fo\{S_1\}\nsubseteq \fo\{S_2\}.
\]\end{thm}
\begin{proof}
By Lemma \ref{EF-Game-refined}, we know that the duplicator has a winning strategy in the game  $G_{S_2}(\widetilde{\mathfrak{A}}^{S_1}_m,\widetilde{\mathfrak{B}}^{S_1}_m).$ In other words, the property defined by $\widetilde{\phi}_{\mathcal{T}}$ is not expressible in $\fo\{S_2\}$, by Lemma \ref{q-type-game-2}. Observe that   $\widetilde{\phi}_{\mathcal{T}}\in\fo\{S_1\}$. Therefore, $\fo\{S_1\}\nsubseteq\fo\{S_2\}$.  
\end{proof}

Recall that $\tau=\langle E\rangle$. Based on the same transformations as the ``reductions from $\tau^+$ to $\tau$'' (p. \pageref{reductions-from-tauplus-to-tau}), the diagraph $\mathfrak{A}^T_m$ ($\mathfrak{B}^T_m$ resp.) is obtained from $\widetilde{\mathfrak{A}}^T_m$ ($\widetilde{\mathfrak{B}}^T_m$ resp.), as $\mathfrak{A}^p_m$ ($\mathfrak{B}^p_m$ resp.) is obtained from $\widetilde{\mathfrak{A}}^p_m$ ($\widetilde{\mathfrak{B}}^p_m$ resp.); $\phi_{\mathcal{T}}$ is obtained from $\widetilde{\phi}_{\mathcal{T}}$ as $\varphi_p$ is obtained from $\widetilde{\varphi}_p$. 
  Recall that the reductions are mainly doing three things: 
 \begin{itemize}
 \item change red edges to forward edges; change blue edges to backward edges;
 \item use self-loops to indicate the positions of junction points;
 \item use bi-directional edges to indicate the positions of   black leaves.
 \end{itemize}
 
\noindent We can use the same reductions, prove a new version of Lemma \ref{translation} as the following:
 
For any first-order sentence $\zeta$ over $\tau$, there is a first-order sentence $\xi$ over $\tau^+$, with the same quantifier structure, such that

\vspace{7pt}
\begin{enumerate}
\item $\mathfrak{A}^T_m\models \zeta\hspace{3pt}$ iff  $\hspace{3pt}\widetilde{\mathfrak{A}}^T_m\models \xi$;\smallskip
\item $\mathfrak{B}^T_m\models \zeta\hspace{3pt}$ iff  $\hspace{3pt}\widetilde{\mathfrak{B}}^T_m\models \xi$. 
 \end{enumerate}\medskip
 
 And Theorem \ref{main-main} is immediate, by the same argument in the proof of Theorem \ref{main1}.

The following is a direct corollary of Theorem \ref{main-main}:

\begin{cor}
Let $S_1$ and $S_2$ be two $\Gamma$-labelled forests and $\sigma_g$ includes a $k$-ary relation symbol where $k\geq 2$. Over the class of all finite $\sigma_g$-structures,
\[
         \mathrm{if}\hspace{3pt} S_1\npreceq_e S_2, \hspace{2pt} \mathrm{then}\hspace{3pt} \fo\{S_1\}\nsubseteq \fo\{S_2\}.
\]
\end{cor}
\begin{proof}
We can use the $k$-ary relation ($k\geq 2$) to encode the binary relations we defined in the structures for Theorem \ref{main-main} and let all the other relations be empty. 
\end{proof}

\section{Summary}\label{finite-Rosen}
It is natural to classify fragments of first-order logic based on quantifier structures. So far, most related works are focused on those fragments based on quantifier prefixes, a special kind of quantifier structure, and contribute to our understanding of their different expressiveness. Gr\"adel-McColm's conjecture, which claims the strictness of the first-order prefix hierarchy, generalizes the results of Walkoe \cite{W.Walkoe_partially-ordered_1970}, Keisler \& Walkoe  \cite{kw73diversity} and Chandra \& Harel \cite{ch82fohierarchy}. Rosen proved this conjecture by showing that it holds over infinite structures \cite{rosen05prifix} and raised the question of whether it also holds over finite structures. We define games that characterize the quantifier classes, which generalize the standard  Ehrenfeucht-Fra\"iss\' e games, and prove that a similar and natural hierarchy, i.e.\ the first-order quantifier structure hierarchy, is strict over finite structures. Although these two hierarchies are similar, they are independent in that none implies the other directly. Nevertheless, our constructions do provide justifications for some special cases of  Gr\"adel-McColm's conjecture over finite structures. For example, from the constructions introduced in this paper we can see that there is a property that is expressible in $\fo(\exists\forall\forall)$, but not in $\fo(\forall\forall\forall)$, $\fo(\forall\forall\exists)$, $\fo(\forall\exists\exists)$ and $\fo(\exists\exists\exists)$. But we don't know whether it is expressible in $\fo(\forall\exists\forall)$ or not. 

Recall that the mapping $\iota$ between quantifier structures is not necessarily injective when we define the embedding relation $\preceq_e$ (Definition \ref{forests-embedding-1}). 
Now, we change the definition of quantifier structure embedding a little bit, i.e.\ require $\iota$ to be injective. 
As a consequence, the quantifier classes are also changed correspondingly. 
And a finite version of Rosen's main theorem \cite{rosen05prifix} can be stated as the following, based on such change. Recall that we assume that all the structures are finite. 
\[
   \text{Let }p\in\Gamma^* \text{ and }S \text{ be a } \Gamma \text {-labelled forest}, \fo(p)\nsubseteq \fo\{S\} \text{ if }p\notin \mathscr{W}(S).
\] 
As has been mentioned in the introduction, a proof of this theorem will solve Gr\"adel-McColm's conjecture over finite structures,  which is still open at present. Another possible way to resolve this conjecture is to define  \EF style games for the prefix classes. In any cases, we need new techniques and insights.

Let $\mathcal{T}$ be a $\Gamma$-labelled tree and $\mathcal{T}^i$ be a forest composed of $i$ disjoint copies of $\mathcal{T}$. Let $\mu_t$ be a map from a $\Gamma$-labelled tree to a natural number such that $\mu_t(\mathcal{T})$ is the minimum number for $m$ such that $\fo\{\mathcal{T}^m\}=\fo\{\mathcal{T}^{m+1}\}$.\footnote{The existence of such a number $m$ is determined by Lemma \ref{finite-logic-equiv-fixed-rank}.} Then, the following is conceivable, under the new definition of quantifier structure embedding: 

Let $\mathcal{T}$ be a  $\Gamma$-labelled tree and $S$ be an arbitrary $\Gamma$-labelled forests. Over the class of all finite digraphs, for any $m\leq \mu_t(\mathcal{T})$,
\[
         \mathrm{if}\hspace{3pt} \mathcal{T}^m\npreceq_e S, \hspace{2pt} \mathrm{then}\hspace{3pt} \fo\{\mathcal{T}^m\}\nsubseteq \fo\{S\}.
\]
  
Note that it generalizes not only Gr\"adel-McColm's conjecture over finite structures but also the finite version of Rosen's main theorem. Furthermore, it is also interesting to study $\mu_t(\mathcal{T})$. For example, how fast will it grow w.r.t. the size of $\mathcal{T}$?    

Another question is whether we can prove similar hierarchies in other logics. One candidate is so called independence-friendly logic  (IFL). What makes it interesting is that IFL has the form of first-order logic, while has the expressive power equals existential second-order logic (ESO). It is well-known that ESO captures the complexity class NP over finite structures. Clearly, establishing a natural and strict hierarchy for the NP problems would be very interesting. 

\section*{Acknowledgement}
  I am grateful to Anuj Dawar for giving me a lot of suggestions and ideas which improve
the paper significantly and to Bjarki Holm for useful discussions and proofreading. I am
also grateful to the anonymous referees for their valuable comments.

\end{document}